\documentclass[cmp]{svjour}  

\usepackage{bm,epsfig,amssymb,mathrsfs,amsmath,float,color,yfonts}
\usepackage{multirow}

\usepackage[normalem]{ulem}
\usepackage{graphicx}

 \def\map#1{\mathcal #1}

\def\<{\langle}\def\>{\rangle}

\def\tr{\operatorname{tr}}
\def\Tr{\operatorname{Tr}}
\def\:{\hbox{\bf
    :}}

\def\R{\mathbb R}
\def\N{\mathbb N}
\def\C{\mathbb C}

\def\spc#1{\mathcal{#1}}
\def\seq#1{\textfrak{#1}}

\def\set#1{\mathsf{#1}}

\def\vec#1{\bm{#1}}

\def\Ker{\mathop{\rm Ker}}  
\def\Supp{\mathop{\rm Supp}} 

\newenvironment{proofof}[1]{\vspace*{5mm} \par 
\noindent{\it Proof of #1:\hspace{2mm}}}{\endproof
\hfill$\Box$ 
}

\newtheorem{theo}{{Theorem}}
\newtheorem{condition}{{Condition}}
\newtheorem{lem}{{Lemma}}
\newtheorem{prop}{{Proposition}}
\newtheorem{cor}{{Corollary}}
\newtheorem{defi}{{Definition}}

\def\Label#1{\label{#1}\ [\ \text{#1}\ ]\ }
\def\Label{\label}

\begin{document}

\title{Attaining the ultimate precision limit in quantum state estimation}
\titlerunning{Attaining the ultimate precision limit in quantum state estimation}
\author{Yuxiang Yang\inst{1}\fnmsep\inst{2}\fnmsep\inst{5}, Giulio Chiribella\inst{1}\fnmsep\inst{3} \fnmsep\inst{4}\fnmsep\inst{5}, and Masahito Hayashi\inst{6}\fnmsep\inst{7}\fnmsep\inst{8}}
\institute{
Department of Computer Science, The University of Hong Kong, Pokfulam Road, Hong Kong.
\and
Institute for Theoretical Physics, ETH Z\"urich, Switzerland\\
\email{yangyu@phys.ethz.ch}
\and
Department of Computer Science,  Parks Road, Oxford, OX1 3QD,  UK.\\
\email{giulio.chiribella@cs.ox.ac.uk}
\and
Perimeter Institute for Theoretical Physics, Waterloo, Ontario N2L 2Y5, Canada
 \and
The University of Hong Kong Shenzhen Institute of Research and Innovation, 5/F, Key Laboratory Platform Building, No.6, Yuexing 2nd Rd., Nanshan, Shenzhen 518057, China
\and
Graduate School of Mathematics, Nagoya University,
Furocho, Chikusaku, Nagoya, 464-860, Japan.\\
\email{masahito@math.nagoya-u.ac.jp}
\and
Centre for Quantum Technologies, National University of Singapore, 3 Science Drive 2, Singapore 117542.
\and
Shenzhen Institute for Quantum Science and Engineering, Southern University of Science and Technology, Shenzhen 518055, China.}  
\authorrunning{Y. Yang, G. Chiribella, M. Hayashi}

\maketitle
\abstract{We derive a bound on the precision of  state estimation for finite dimensional quantum systems and prove its attainability in the generic case where the spectrum is non-degenerate. 
     Our results hold under an assumption called local asymptotic covariance, which is weaker than unbiasedness or local  unbiasedness. The derivation  is based on an analysis of the limiting distribution of the estimator's deviation from the true value of the parameter, and takes advantage of  quantum local asymptotic normality,  a useful asymptotic characterization of identically prepared states in terms of Gaussian states.  We first prove our results for the mean square error of a special class of models, called $D$-invariant, and then extend the results to arbitrary models, generic cost functions,  and  global state estimation, where the unknown parameter is not restricted to a local neighbourhood of the true value. The extension includes a treatment of nuisance parameters, i.e. parameters that are not  of interest to the experimenter but nevertheless affect the precision of the estimation.       As an illustration of the general approach, we provide the optimal estimation strategies for the joint measurement of two qubit observables, for the estimation of  qubit states in the presence of amplitude damping noise, and for  noisy multiphase  estimation. }

\medskip

\section{Introduction.}

Quantum estimation theory is one of the pillars of  quantum information science, with a wide range of applications from evaluating the performance of quantum devices \cite{d2001quantum,knill2008randomized} to exploring the foundation of physics \cite{heisenberg1927anschaulichen,schnabel2010quantum}.   In the typical scenario, the problem is specified by a  parametric family of quantum states,  called the {\em model}, and the  objective  is to design  measurement strategies that estimate the  parameters of interest with  the highest possible precision.
The precision measure is often chosen to be the mean square error (MSE),  and is  lower bounded  through generalizations of the Cram\'{e}r-Rao bound of classical statistics \cite{holevo-book,helstrom-book}. Given $n$ copies of a quantum state, such generalizations  imply  that the  product ${\rm MSE}\cdot n$ converges to a positive constant in the large $n$ limit. 

Despite many efforts   made over the years (see, e.g., \cite{holevo-book,helstrom-book,massar1995optimal,gill2000state,fujiwara2001estimation,matsumoto2002new,hayashi2005statistical,gill2013asymptotic} and \cite{szczykulska2016multi} for a review), the attainability of the precision bounds of quantum state estimation has only been proven in a few special cases. 
Consider, as an example, the most widely used bound, namely the symmetric logarithmic derivative Fisher information bound (SLD bound, for short).  
The SLD bound is tight in  the one-parameter case \cite{holevo-book,helstrom-book}, 
but is generally non-tight  in multiparameter estimation.
Intuitively, measuring one parameter may affect the precision in the measurement of  another parameter, and thus it is extremely tricky  to construct the optimal measurement.
Another bound for multiparameter estimation is the right logarithmic derivative Fisher information bound (RLD bound, in short)\cite{holevo-book}.
Its achievability was  shown in the Gaussian states case \cite{holevo-book},
the qubits case \cite{0411073,hayashi2008asymptotic}, and the qudits case \cite{guta-lan,qlan}.
In this sense,  the RLD bound is superior to  the SLD bound.
However, the RLD bound holds only when the family of states to be estimated satisfies an {\em ad hoc} mathematical condition. The most general quantum extension of the classical Cram\'er-Rao bound till now is the Holevo bound \cite{holevo-book}, which gives the maximum among all existing lower bounds for the error of unbiased measurements for the estimation of any family of states.  
The attainability of the Holevo bound was studied  
in the pure states case \cite{matsumoto2002new} and the qubit case \cite{0411073,hayashi2008asymptotic},
and was conjectured to be generic by  one of us  \cite{hayashi2009quantum}.   
Yamagata, Fujiwara, and Gill \cite{yamagata2013quantum} addressed the attainability question in a local scenario, showing that the Holevo bound can be attained under  certain regularity conditions.  
However, the attaining estimator constructed therein depends on the true parameter, and therefore has limited practical interest. 
Meanwhile, the need of  a general, attainable bound on multiparameter quantum estimation is increasing,  as more and more  applications are being investigated \cite{yue2014quantum,knott2016local,zhang2017quantum,liu2017control,proctor2017multi}.

In this work we explore a new route to the study of precision limits in quantum estimation. 
    This new route  allows us to prove the asymptotic attainability of the Holevo bound in generic scenarios, to extend its validity to a broader class of estimators,  and to derive a new set of attainable precision bounds. 
We adopt the condition of \emph{local asymptotic covariance} \cite{hayashi2009quantum} which is less restrictive than the unbiasedness condition \cite{holevo-book} assumed in the derivation of the Holevo bound. 
Under local asymptotic covariance, we characterize the MSE of the \emph{limiting distribution}, 
namely the distribution of the estimator's rescaled deviation from the true value of the parameter in the asymptotic limit of $n\to\infty$.

Our contribution  can be divided into two parts, 
the attainability  of the Holevo bound and the proof that the Holevo bound still holds under the weaker condition of local asymptotic covariance.
To show the achievability part, we employ 
 {\em quantum local asymptotic normality} (Q-LAN),  a useful characterization of $n$-copy $d$-dimensional (qudit) states in terms of multimode  Gaussian states.
The qubit case was derived in \cite{0411073,hayashi2008asymptotic}
and the case of full parametric models was derived by Kahn and Guta 
when the state has non-degenerate spectrum \cite{guta-lan,qlan}.
Here we extend this characterization to a larger class of models, called $D$-invariant models, using a technique of symplectic diagonalization.  For models that are not $D$-invariant,
we derive an  achievable bound, expressed in terms of a quantum Fisher information-like quantity that  can be straightforwardly evaluated. Whenever the model consists of qudit states with non-degenerate spectrum, this quantity turns out  to be equal to the quantity in the Holevo bound \cite{holevo-book}.  Our evaluation has compact uniformity and order estimation of the convergence, which will allow us to prove the achievability of the bound even in the global setting.


 We stress that, until now, the most general proof of the Holevo bound required the condition of local unbiasedness.  In particular, no previous  study showed the validity of the Holevo bound under the weaker condition of local asymptotic covariance in  the multiparameter scenario.
To avoid employing the (local) unbiasedness condition,
we focus on the discretized version of the RLD Fisher information matrix, introduced by Tsuda and Matsumoto \cite{tsuda2005quantum}. 
Using this version of the RLD Fisher information matrix,
we manage to handle the local asymptotic covariance condition and to show the validity of the Holevo bound in this broader scenario. 
 Remarkably, the validity of the bound   does not require finite-dimensionality of the system or non-degeneracy of the states in the model. 
  Our result also provides a simpler way of evaluating the Holevo bound, whose original expression  involved  a difficult optimization over a set  of operators.  

The advantage of  local asymptotic covariance over local unbiasedness is the following.
For  practical applications, the estimator needs to attain the lower bound globally, i.e., at all points in the parameter set.
However, it is quite difficult to meet this desideratum  under the condition of local unbiasedness, even if we employ a two-step method based on a first rough estimate of the state, followed by the  measurement that is optimal  in the neighbourhood of the estimate. 
In this paper,   
we construct a locally asymptotic covariant estimator that achieves the Holevo bound at every point, for any qudit submodel except those with  degenerate states. 
Our construction proceeds in two steps. In the first step, we perform a full tomography of the state, using the protocol proposed in \cite{haah2017sample}. In the second step, we implement a locally optimal estimator based on  Q-LAN \cite{guta-lan,qlan}. 
The two-step estimator works even when 
the estimated parameter is not assumed to be in a local neighbourhood of the true value.  
The key tool to prove this property is 
our precise evaluation of the optimal local estimator with 
compact uniformity and order estimation of the convergence. 
Our method can be extended from the MSE to  arbitrary cost functions.
A comparison between the approach adopted in  this work (in green)  and conventional approaches to quantum state estimation (in blue) can be found in Figure \ref{fig:compare}.

  Besides the attainability of the Holevo bound, the method  can be used to derive a broad class of bounds for quantum state estimation. 
Under suitable assumptions, we  characterize the tail  of the limiting distribution, providing a bound on the probability that the estimate falls out of a confidence region.
The limiting distribution is a good approximation of the (actual) probability distribution of the estimator, up to a term vanishing in $n$.    
 Then, we  derive a bound for quantum estimation with \emph{nuisance parameters}, {\em i.e.}  parameters that are not of interest to the experimenter but may affect the  estimation of the other parameters.  For instance, the strength of noise in a phase estimation scenario can be regarded as a nuisance parameter.  Our  bound applies also to arbitrary estimation models, thus extending nuisance parameter bounds derived for specific cases (see, e.g., \cite{suzuki,kok2017role,yang2017quantum}).  In the final part of the paper, the above bounds are illustrated in  concrete examples, including the joint measurement of two qubit observables,  the estimation of  qubit states in the presence of amplitude damping noise, and  noisy multiphase  estimation. 
 


 \begin{figure}[t!]
\begin{center}
  \includegraphics[width=\linewidth]{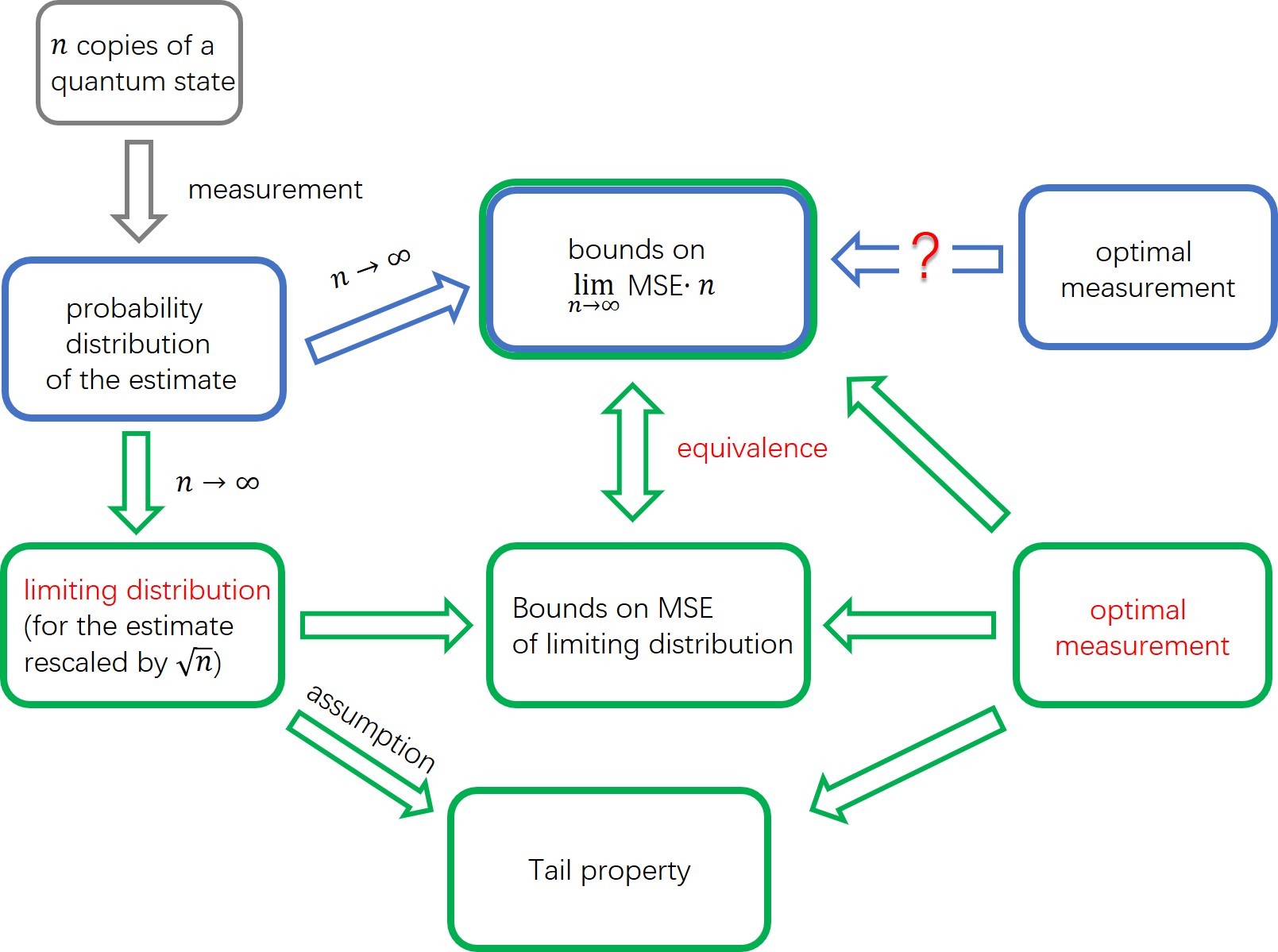}
  \end{center}
\caption{\Label{fig:compare}
  {\bf Comparison between the approach  of this work (in green) and the traditional  approach of quantum  state estimation (in blue).}   In the traditional approach, one derives precision bounds based on the probability distribution function (PDF) for measurements on the original set of quantum states. The bounds are evaluated in the large $n$ limit and the task is to find a sequence of measurements that achieves the limit bound. In this work, we first characterize the limiting distribution and then work out a bound in terms of the limiting distribution.  This construction also provides the optimal measurement in the limiting scenario, which can be used to prove the asymptotic attainability of the bound.  The analysis of the limiting distribution also provides tail bounds, which approximate the tail bounds for finite $n$ up to a small correction, under the assumption that the cost function and the model satisfy a certain relation (see Theorem \ref{thm-tail}).}
\end{figure}


\begin{table}
\begin{center}
    \begin{tabular}{ | p{5cm} | c | p{3.5cm} |}
    \hline
    Term & Notation & Definition \\ \hline
   Limiting distribution & $\wp_{\vec{t}_0,\vec{t}|\seq{m}}$ & Eqs. (\ref{limiting-single}) and (\ref{limiting}) \\
     \hline
   SLD quantum Fisher information at $\vec{t}_0$ & {$J_{\vec{t}_0}$} & {Eq. (\ref{SLD-QFI}) }\\
     \hline
        \multirow{2}{*}{RLD quantum Fisher information at $\vec{t}_0$} & \multirow{2}{*}{$\widetilde{J}_{\vec{t}_0}$} & Eq. (\ref{RLD-QFI}) (for D-invariant models) \\
     \hline
       D-matrix at $\vec{t}_0$ & $D_{\vec{t}_0}$ & Eq. (\ref{D-m}) \\
     \hline
     Gaussian state & $G[\vec{\alpha},\vec{\gamma}]$ & Eq. (\ref{Gaussian-state})\\
     \hline
(Multi-mode) displaced thermal state & {$\Phi[(\vec{\alpha}^R,\vec{\alpha}^I),\vec{\beta}]$} 
     &  {Eq. (\ref{multidisplaced-thermal})}\\
     \hline
     Gaussian shift operator & $ T_{\vec{\alpha}}$ & Eq. (\ref{Gaussian-shift-operator})\\
     \hline
     $2m$-dimensional symplectic form  & {$\Omega_m$} 
     & {Eq. (\ref{M2})}\\
     \hline
A $2m$-dimensional diagonal matrix & {$E_m(\vec{x})$} 
& {Eq. (\ref{M1})}\\
\hline
Covariance matrix of a probability distribution $\wp$ & \multirow{2}{*}{$V[\wp]$} 
& \multirow{2}{*}{Eq. (\ref{cov-def})}\\
\hline
    \end{tabular}\caption{Table of terms and notations.}\Label{table-notation}
\end{center}
\end{table}

The remainder of the paper is structured as follows. In Section \ref{sec-single} we  introduce the main ideas  in the one-parameter case. 
Our discussion of the one-parameter case requires no regularity condition for the parametric model.   
Then we devote several sections to introducing and deriving tools for the multiparameter estimation. 
In Section \ref{s3}, we briefly review the Holevo bound and Gaussian states, and derive some relations that will be useful in the rest of the paper. 
In Sections \ref{sec:Tool}, we introduce Q-LAN.   
In Section \ref{sec-RLD} we introduce the $\epsilon$-difference RLD Fisher information matrix, which will be a key tool for deriving our bounds in the multiparameter case. 
In Section \ref{sec:multi}, we derive the general bound on the precision of multiparameter estimation.   In Section \ref{sec-general}, we address state estimation in the presence of nuisance parameters and derive a precision bound for this scenario.
Section \ref{sec-tail} provides bounds on the tail probability.
In Section \ref{sec-global}, we extend our results to global estimation and to generic cost functions. 
In Section \ref{sec:applications}, the general method is illustrated through examples.  The conclusions are drawn  in Section \ref{sec-conclusion}.

{\em Remark on the notation.} In this paper, we use $z^*$ for the complex conjugate of $z\in\C$ and $A^\dag$ for the Hermitian conjugate of an operator $A$.
For convenience of the reader, we list other frequently appearing notations and their definitions in Table \ref{table-notation}.

\section{Precision bound under  local  asymptotic covariance: one-parameter case.}\Label{sec-single}
In this section, we discuss estimation of a single parameter under the local asymptotic covariance condition, 
without any assumption on  the parametric model. 

\subsection{Cram\'{e}r-Rao inequality without regularity assumptions}\label{S-one}

Consider a  one-parameter model $\map M$, of the form 
\begin{align}\label{quditmodel}
\map M  =   \left\{ \rho_t \right\}_{    t  \in  \set \Theta} 
\end{align} 
where $\set \Theta$ is a subset of $\mathbb{R}$.
In the literature it is typically assumed that the parametrization is differentiable. When this is the case, one can define 
the symmetric logarithmic derivative operator (SLD in short) at $t_0$ via the equation
\begin{align}
\frac{d \rho_{{t_0}}}{d t_0}=
\frac12\left(
\rho_{{t}_0}{L}_{t_0}+{L}_{t_0}\rho_{{t}_0}\right).
\label{LLP}
\end{align}
Then, 
the SLD Fisher information is defined as 
\begin{align}
J_{{t_0}}:=\Tr \left[  \rho_{{t_0}} L_{t_0}^2\right] \, .\Label{DefSLD1}
\end{align}
The SLD $L_{t_0}$ is not unique in general, 
but the SLD Fisher information 
$J_{{t_0}}$ is uniquely defined because it does not depend on the choice of 
the SLD $L_{t_0}$ among the operators satisfying  \eqref{LLP}. 
When the parametrization is $C^1$-continuous and $\epsilon >0$ is a small number, 
one has
\begin{align}
F\left(\rho_{t_0} ||\rho_{t_0+\epsilon}\right)
= 1-\frac{J_{t_0}}{8}\epsilon^2+ o(\epsilon^2)\label{LLP2B} \, ,
\end{align}
where 
\begin{align}
F(\rho\|\rho'):= \Tr |\sqrt{\rho}\sqrt{\rho'}|
\end{align}
is the fidelity between two density matrices $\rho$ and $\rho'$. 
It is called Bhattacharya or Hellinger coefficient in the classical case
\cite{bhattacharyya1943measure,hellinger1909neue}.

Here we do not assume that the parametrization \eqref{quditmodel} is differentiable. Hence, 
the SLD Fisher information cannot be defined by \eqref{DefSLD1}.
Instead, following the intuition of \eqref{LLP2B}, we define the SLD Fisher information $J_{t_0}$  as the limit
\begin{align}
J_{t_0}:=\liminf_{\epsilon\to 0}\frac{8 \big[ \,1-F\left(\rho_{t_0} ||\rho_{t_0+\epsilon}\right)   \big]}{\epsilon^2 } \, .
\label{LLP2}
\end{align}
In the $n$-copy case,  we have the following lemma: 
\begin{lem}\label{MGR}
\begin{align}
\liminf_{n \to \infty}
\frac{8\left(1-F\left(\rho_{t_0}^{\otimes n}||\rho_{t_0+\frac{\epsilon}{\sqrt{n}}}^{\otimes n}\right)\right)}{\epsilon^2}
{=}  
\frac{8\left(1-
e^{-\frac{J_{t_0} \epsilon^2}{8}}
\right)}
{\epsilon^2}.
\end{align}
\end{lem}
\begin{proof}
Using the definition \eqref{LLP2}, we have
\begin{align}
& 
\liminf_{n \to \infty}
\frac{8\left(1-F\left(\rho_{t_0}^{\otimes n}||\rho_{t_0+\frac{\epsilon}{\sqrt{n}}}^{\otimes n}\right)\right)}{\epsilon^2}
\nonumber \\
&{=}  
\liminf_{n \to \infty}
\frac{8\left(1-
\left(1-\frac{J_{t_0} \epsilon^2}{8n} 
+o(\frac{ \epsilon^2}{n}) 
\right)^n\right)}
{\epsilon^2}
{=}  
\frac{8\left(1-
e^{-\frac{J_{t_0} \epsilon^2}{8}}
\right)}
{\epsilon^2}.
\end{align}
\end{proof}
In other words, the SLD Fisher information is constant over $n$ if we replace $\epsilon$ by $\epsilon/\sqrt n$.

To estimate the parameter $t\in\set \Theta$, we perform on the input state a quantum measurement, which is mathematically described by a {\em positive operator valued measure (POVM)} with outcomes in $\cal{X}\subset\R$. An outcome $x$ is then mapped to an estimate of $t$ by an estimator $\hat{t}(x)$. It is often assumed that the measurement   is unbiased, in the following sense: 
a  POVM  $M$ on a single input copy is called \emph{unbiased}
when
\begin{align}\Label{unbiased}
\int_{\cal{X}} \hat{t}(x) \, \Tr \left[ \rho_{t} \,M(d {x})\right]=t, \qquad\forall t\in\set \Theta.
\end{align}
For a POVM $M$, we define the mean square error (MSE)  $V_{t}(M)$ as
\begin{align}
V_{{t}}(M)
:=\int_{\cal{X}} \left(\hat{t}(x)-t\right)^2 \Tr \left[ \, \rho_{{t}} M(d{x}) \,\right] \, .
\end{align}

Then, we have the fidelity version of the Cram\'{e}r-Rao inequality:

\begin{theo}\Label{L16}
For an unbiased measurement $M$ satisfying
\begin{align}
\int_{{\cal X}} \hat{t}(x)P_t(dx)=t\Label{UBC0}
\end{align}
for any $t$,
we have
\begin{align}
\frac{1}{2}
\big[ \, V_{t_0}(M)+V_{t_0+\epsilon}(M)+\epsilon^2  \,  \big]  \ge 
\frac{\epsilon^2}{8 \big[ \, 1-F(\rho_{t_0} ||\rho_{t_0+\epsilon}) \,\big]} \, .\label{O6}
\end{align}
When $\lim_{\epsilon \to 0}V_{t_0+\epsilon}(M)=V_{t_0}(M)$,
taking the limit $\epsilon \to 0$, we have
\begin{align}
V_{t_0}(M) \ge J_{t_0}^{-1} .
\label{O9}
\end{align}
\end{theo}

The proof uses the notion of fidelity  between two classical probability distributions: 
for two given distributions $P$ and $Q$ on a probability space ${\cal X}$, 
we define the fidelity $F(P\|Q)$ as follows.
Let $f_P$ and $f_Q$ be the  Radon-Nikod\'{y}m derivatives of $P$ and $Q$ 
with respect to $P+Q$, respectively.
Then, the fidelity $F(P\|Q)$ can be defined as
\begin{align}\label{deffid}
F(P\|Q)  :  = \int_{{\cal X}} \sqrt{f_P(x)}\sqrt{f_Q(x)} (P+Q)(dx)  \, .
\end{align}
With the above definition, the fidelity satisfies an information processing inequality: for every classical channel $\map G$, one has $F(\map G (P) \|   \map G(Q))  \ge  F(  P\|Q)$.  For  a family of probability distributions $\{P_\theta\}_{\theta \in  \set \Theta}$, we define the Fisher information as  
\begin{align}
J_{t_0}:=\liminf_{\epsilon\to 0}\frac{8 \big[ \,1-F\left(P_{t_0} ||P_{t_0+\epsilon}\right)   \big]}{\epsilon^2 } \, .
\end{align}
When the probability distributions are over a discrete set, their Fisher information coincides with the quantum SLD of the corresponding diagonal matrices.

\begin{proofof}{Theorem \ref{L16}}
Without loss of generality, we assume $t_0=0$.
We define the probability distribution $P_t$
by $P_t(B):= \Tr\left[ \,\rho_t M(B) \,\right]$.
 Then, the information processing inequality of the fidelity \cite{fuchs1996distinguishability} yields the bound $F(\rho_{t_0} ||\rho_{t_0+\epsilon}) \le F(P_0\|P_\epsilon)$. Hence, 
it is sufficient to show \eqref{O6} for the probability distribution family $\{P_t\}$.

Let $f_0$ and $f_\epsilon$ 
be the  Radon-Nikod\'{y}m derivatives of $P_0$ and $P_{\epsilon}$ with respect to $P_{0}+P_{\epsilon}$. Denoting the estimate by $\hat{t}$, we have
\begin{align*}
V_{\epsilon} (\hat{t})
=&
\int_{{\cal X}}
(\hat{t}(x)-\epsilon)^2
P_{\epsilon}(d x) 
=
\int_{{\cal X}}
\hat{t}(x)^2
P_{\epsilon}(d x) 
-2\epsilon \int_{{\cal X}}
\hat{t}(x)
P_{\epsilon}(d x) 
+\epsilon^2 \\
=&
\int_{{\cal X}}
\hat{t}(x)^2
P_{\epsilon}(d x) 
-2\epsilon^2 +\epsilon^2
=
\int_{{\cal X}}
\hat{t}(x)^2
P_{\epsilon}(d x) 
-\epsilon^2 ,
\end{align*}
and therefore
\begin{align}
& 2 V_{0} (\hat{t}) 
+2 V_{\epsilon} ( \hat{t}) +2\epsilon^2 
=
\int_{{\cal X}}2 
\hat{t}(x)^2
(P_{0}(d x) +P_{\epsilon}(d x) )
  \nonumber \\
= &
\int_{{\cal X}}2 
\hat{t}(x)^2
(f_0(x)+f_\epsilon(x)) (P_{0}+P_{\epsilon})(d x) 
  \nonumber \\
\ge &
\int_{{\cal X}}
\hat{t}(x)^2
(\sqrt{f_0(x)}+\sqrt{f_\epsilon(x)})^2 (P_{0}+P_{\epsilon})(d x) .
\label{O1}
\end{align}
Also, \eqref{deffid} implies the relation
\begin{align}
\int_{{\cal X}}
(\sqrt{f_0(x)}-\sqrt{f_\epsilon(x)})^2 
(P_{0}+P_{\epsilon})(d x) 
=2-2F(P_0\|P_\epsilon).\label{O2}
\end{align}
Hence, Schwartz inequality implies 
\begin{align}
& 
\int_{{\cal X}}\hat{t}(x)^2
(\sqrt{f_0(x)}+\sqrt{f_\epsilon(x)})^2 (P_{0}+P_{\epsilon})(d x) 
\cdot 
\int_{{\cal X}}
(\sqrt{f_0(x)}-\sqrt{f_\epsilon(x)})^2 
(P_{0}+P_{\epsilon})(d x) 
 \nonumber \\
\ge &
\Big(
\int_{{\cal X}}\hat{t}(x)
(\sqrt{f_0(x)}+\sqrt{f_\epsilon(x)})
(\sqrt{f_0(x)}-\sqrt{f_\epsilon(x)})
 (P_{0}+P_{\epsilon})(d x) 
\Big)^2 
\nonumber \\
=&
\Big(
\int_{{\cal X}}\hat{t}(x)
(f_0(x)-f_\epsilon(x))
 (P_{0}+P_{\epsilon})(d x) 
\Big)^2 
\nonumber \\
=&
\Big(
\int_{{\cal X}}
\hat{t}(x)
 (P_{0}-P_{\epsilon})(d x) 
\Big)^2 
=\epsilon^2 \label{O3}
\end{align}
Combining \eqref{O1}, \eqref{O2},  and \eqref{O3}
we have
\eqref{O6}.
\end{proofof}

\subsection{Local asymptotic covariance} 
When many copies of the state  $\rho_t$ are available, the estimation of $t$ can be reduced to a local neighbourhood of a fixed point $t_0 \in \set \Theta$.  
Motivated by Lemma \ref{MGR},  we adopt the following parametrization of the $n$-copy state
\begin{align}\label{local-model} 
\rho_{t_0,t}^{n}
:=\rho_{t_0+\frac{t}{\sqrt{n}}}^{\otimes n} \, ,  \qquad  t\in \set \Delta_n   :=      \sqrt n  \,       (   \set \Theta-t_0) \, ,    
\end{align} 
having used the notation $a \,  \set \Theta  +  b : =  \{  a  x+   b \,  |  x\in\set \Theta\}$, for two arbitrary constants $a, b\in  \R$.  
With this parametrization, the  local  $n$-copy model is
$  \left\{\rho_{t_0,t}^{n}           
    \right\}_{t  \in  \set \Delta_n} $.
Note that,  for every $t\in\R$, there exists a sufficiently large $n$ such that $ \set \Delta_n$ contains $t$.     As a consequence,  one has   $ \bigcup_{n\in \N}\,  \set \Delta_n =  \R$. 
 
Assuming $t_0$ to be known, the task is to estimate the local parameter $t\in\R$, 
by performing a measurement   on the $n$-copy state $\rho_{t_0,t}^n$ 
and then mapping the obtained data to an estimate $\hat{t}_n$.   The whole estimation strategy can be described by a sequence of POVMs $\seq{m}:=\{M_n\}$.   
For every Borel set $\set{B}\subset\R$, we adopt the standard notation
\begin{align*}
M_n(\set{B}):=\int_{\hat{t}_n\in\set{B}}~M_n(d\hat{t}_n)  \, . 
\end{align*}

In the  existing works on quantum state estimation,
 the error criterion  is defined in terms of the difference
between the global estimate $t_0+\frac{\hat{t}_n}{\sqrt{n}}$
and the global true value $t_0+\frac{t}{\sqrt{n}}$.  
Instead, here we focus on the difference between the local estimate $\hat{t}_n$  and the true value of the local parameter $t$. With this aim in mind, we consider the probability distribution  
\begin{align}\Label{pdf}
\wp^{n}_{t_0,t|M_n}(\set{B})
:=\Tr  \left [\rho_{t_0,t}^{n} M_n
\left( \frac{\set{B}}{\sqrt{n} }+{t}_0\right)
\right] \, .
\end{align}

We focus  on the behavior of $\wp^{n}_{t_0,t|M_n}$ in the large $n$ limit, assuming the following condition:

\begin{condition}[Local asymptotic covariance for a single-parameter] 
A sequence of measurements $\seq{m}=\{M_n\}$ satisfies local asymptotic covariance 
\footnote{The counterpart of this condition in classical statistics is known as asymptotic equivalent-in-law or regular. See, for instance, page 115 of \cite{van1998asymptotic}.} 
when
\begin{enumerate}
\item
The distribution $\wp^{n}_{t_0,t|M_n}$ 
(\ref{pdf}) converges  
 to a distribution $\wp_{t_0,t| \seq{m}}$, called 
the limiting distribution, namely
\begin{align}\Label{limiting-single}
\wp_{t_0,t| \seq{m}}\left(\set{B}\right):=
\lim_{n \to \infty}\wp^{n}_{t_0,t|M_n}\left(\set{B}\right)
\end{align}
for any Borel set $\set{B}$.

\item the limiting distribution satisfies  the relation
\begin{align}
\wp_{t_0,t| \seq{m}}(\set{B}+t)=\wp_{t_0,0| \seq{m}}(\set{B})\Label{LL}
\end{align}
for any $t\in\R$, which is equivalent to the condition
\begin{align}
\lim_{n \to \infty}
\Tr  \left [\rho_{t_0,t}^{n} M_n
\left( \frac{\set{B}}{\sqrt{n} }+{t}_0+\frac{t}{\sqrt{n} }\right)
\right]
=
\lim_{n \to \infty}
\Tr  \left [\rho_{t_0,0}^{n} M_n
\left( \frac{\set{B}}{\sqrt{n} }+{t}_0\right)
\right].
\end{align}
\end{enumerate}
\end{condition}

Using the limiting distribution, we can faithfully approximate the tail probability as
\begin{align}\Label{limit-real}
\mathbf{Prob}\left\{  \left|\hat{t}_n-t\right| > \epsilon\right\}&= \wp_{t_0,t|\seq{m}}((-\infty,-\epsilon)\cup(\epsilon,\infty))+\epsilon_n
\end{align}
where  the $\epsilon_n$ term  vanishes with $n$ for every fixed $\epsilon$.

For convenience, one may be tempted to require  the existence of a probability density function (PDF) 
of the limiting distribution
$\wp_{t_0,t| \seq{m}}$.
However, the existence of a PDF is already guaranteed by  the following lemma.
\begin{lem}\label{LGD}
When a sequence $\seq{m}:=\{M_n\}$ of POVMs  satisfies local asymptotic covariance, the  limiting distribution $\wp_{t_0,t| \seq{m}} $ admits a PDF, denoted by $\wp_{t_0,0|\seq{m},d}$. 
\end{lem}

The proof is provided  in Appendix \ref{app:lemma1}.

\subsection{MSE bound for the limiting distribution}\Label{S2-2-2}

As a figure of merit, we focus on the mean square error (MSE)
$V[\wp_{t_0,t| \seq{m}}]$
of the limiting distribution $\wp_{t_0,t| \seq{m}}$, namely
\begin{align*}
V[\wp_{t_0,t| \seq{m}}]:=\int_{-\infty}^{\infty} (\hat{t}-t)^2\, \wp_{t_0,t| \seq{m}}\left(d\hat{t}\right) \, .
\end{align*}
Note that local asymptotic covariance implies that the MSE is independent of $t$. 

The main result of the section is  the following theorem: 
\begin{theo}[MSE bound for single-parameter estimation]\Label{ThY}
When a sequence $\seq{m}:=\{M_n\}$ of POVMs  satisfies local asymptotic covariance, the MSE of its limiting distribution is lower bounded as
\begin{align}
V[\wp_{t_0,t| \seq{m}}]
\ge J^{-1}_{t_0} \, ,\Label{TY}
\end{align}
where $J_{t_0}$ is the SLD Fisher information of the model $\{\rho_t\}_{t\in\set \Theta}$. 
The PDF of $\wp_{t_0,t| \seq{m}}$ is upper bounded by $\sqrt{J_{t_0}}$.
When the PDF of $\wp_{t_0,t| \seq{m}}$ is differentiable with respect to $t$,
equality in (\ref{TY}) holds if and only if
$\wp_{t_0,t| \seq{m}}$
is the normal distribution with average zero and variance 
$J_{t_0}^{-1}$.
\end{theo}


\begin{proofof}{Theorem \ref{ThY}}
When 
the integral $\int_{\mathbb{R}} \hat{t} 
\wp_{t_0,0|\seq{m}}(d \hat{t})$
does not converge,
$V[\wp_{t_0,t| \seq{m}}]$ is infinite and satisfies \eqref{TY}.
Hence, we can assume that the above integral converges.
Further, we can assume that the outcome $\hat{t}$ satisfies
the unbiasedness condition 
$\int_{\mathbb{R}} \hat{t} 
\wp_{t_0,t|\seq{m}}(d \hat{t})=t$.
Otherwise, we can replace $\hat{t}$ by 
$\hat{t}_0:=\hat{t}-\int_{\mathbb{R}} \hat{t}'
\wp_{t_0,0|\seq{m}}(d \hat{t}')$
because the estimator $\hat{t}_0$ 
has a smaller MSE than $\hat{t}$
and satisfies the unbiasedness condition due to 
the covariance condition.
Hence, Theorem \ref{L16} guarantees 
\begin{align}
V[\wp_{t_0,t| \seq{m}}]
=V[\wp_{t_0,0| \seq{m}}]
&\ge
\Big(\liminf_{\epsilon \to 0} \frac{8\left(1- F(\wp_{t_0,0|\seq{m}}||\wp_{t_0,\epsilon|\seq{m}})\right)}{\epsilon^2} \Big)^{-1}.
\Label{TU9}
\end{align}
Applying Lemma \ref{ALL8} to $\{\wp_{t_0,t|\seq{m}}\}$, we have
\begin{align}
&\liminf_{\epsilon \to 0} \frac{8\left(1- F(\wp_{t_0,0|\seq{m}}||
\wp_{t_0,\epsilon|\seq{m}})\right)}{\epsilon^2}
\nonumber \\
\stackrel{(a)}{\le} &
  \liminf_{\epsilon \to 0}\liminf_{n \to \infty}
\frac{8\left(1-F\left(\wp^n_{t_0,0|M_n}||\wp^n_{t_0,\epsilon|M_n} \right)\right)}{\epsilon^2} 
\nonumber \\
\stackrel{(b)}{\le} & \liminf_{\epsilon\to 0}\liminf_{n \to \infty}
\frac{8\left(1-F\left(\rho_{t_0}^{\otimes n}||\rho_{t_0+\frac{\epsilon}{\sqrt{n}}}^{\otimes n}\right)\right)}{\epsilon^2}
\nonumber \\
\stackrel{(c)}{=} & 
 \liminf_{\epsilon\to 0}
\frac{8\left(1-
e^{-\frac{J_{t_0} \epsilon^2}{8}}
\right)}
{\epsilon^2}
= J_{t_0}\Label{TUX} \, .
\end{align}
The inequality $(a)$ holds by Lemma \ref{ALL8} from Appendix \ref{app:asymptoticstuff}, and the inequality $(b)$ comes from the data-processing inequality of the fidelity. 
The equation $(c)$ follows from Lemma \ref{MGR}.
Finally, substituting Eq. (\ref{TUX}) into Eq. (\ref{TU9}), we have the desired bound (\ref{TY}).

Now, we denote the PDF of $\wp_{t_0,0|\seq{m}}$ by $\wp_{t_0,0|\seq{m},d}$.
In Appendix \ref{app:lemma1}  the proof of Lemma \ref{LGD} shows that
we can apply Lemma \ref{L30} to $\{\wp_{t_0,t| \seq{m}}\}_t$.
Since the Fisher information of $\{\wp_{t_0,t| \seq{m}}\}_t$
is upper bounded by $ J_{t_0}\Label{TU}$,
this application guarantees that 
$\wp_{t_0,t| \seq{m},d}(x)\le \sqrt{J_{t_0}}$ for $x \in \mathbb{R}$.

When the PDF $\wp_{t_0,t| \seq{m},d}$ is differentiable, 
to derive the equality condition in Eq. (\ref{TY}),
we show \eqref{TU9} in a different way.
Let $l_{t_0,t}(x)$ be the logarithmic derivative of $\wp_{t_0,t|\seq{m},d}(x)$, defined as 
$l_{t_0,t}(x)\cdot \wp_{t_0,t|\seq{m},d}(x): =\frac{\partial \log \wp_{t_0,t|\seq{m},d}(x)}{\partial t}
=\frac{\partial \log \wp_{t_0,0|\seq{m},d}(x-t)}{\partial t}$.  
By Schwartz inequality, we have
\begin{align}
V[\wp_{t_0,0| \seq{m}}]
&\ge\frac{\left|\int_{-\infty}^{\infty} \hat{x}\,l_{t_0,0}(\hat{x})
\,\wp_{t_0,0|\seq{m}}(\hat{x})d\hat{x}\right|^2}
{\int_{-\infty}^{\infty} l_{t_0,0}^2(\hat{x}) \wp_{t_0,0|\seq{m}}(\hat{x})d\hat{x}}.\Label{TU2E}
\end{align}
The numerator on the right hand side of Eq. (\ref{TU2E}) can be evaluated by noticing that
\begin{align*}
\int_{-\infty}^{\infty} \hat{x}\,l_{t_0,0}(\hat{x})\,\wp_{t_0,0|\seq{m}}(\hat{x})
d\hat{x}
&=\left.\frac{\partial}{\partial x}\int_{-\infty}^{\infty} \hat{x}\,\wp_{t_0,x|\seq{m}}(\hat{x})
d\hat{x}\right|_{x=0}.
\end{align*}
By local asymptotic covariance, this quantity can be evaluated as
\begin{align}
\int_{-\infty}^{\infty} \hat{x}\,l_{t_0,0}(\hat{x})\,\wp_{t_0,0|\seq{m}}(\hat{x})d\hat{x}
=\left.\frac{\partial }{\partial x}\left[\int_{-\infty}^{\infty}( \hat{x}+x)\,\wp_{t_0,0|\seq{m}}(\hat{x})
d\hat{x}\right]\right|_{x=0}=1.\Label{TU1C}
\end{align}
Hence, (\ref{TU2E})  coincides with \eqref{TU9}.
The denominator on the right hand side of (\ref{TU2E}) equals 
the right hand side of \eqref{TU9}.
The equality in Eq. (\ref{TU2E}) holds if and only if
$\int_{-\infty}^{\infty} \hat{x}^2 \wp_{t_0,0|\seq{m},d}( \hat{x})d \hat{x} =J^{-1}_{t_0}$ and
$\frac{d}{dx} \log \wp_{t_0,0|\seq{m},d}(\hat{x})$ is proportional to $\hat{x}$, which implies that 
$\wp_{t_0,0|\seq{m}}$ 
is the normal distribution with average zero and variance $J^{-1}_{t_0}$. 
\end{proofof}

\medskip

The RHS of (\ref{TY}) can be regarded as the
limiting distribution version of the SLD quantum Cram\'er-Rao bound.  Note that, when  the  limiting PDF is differentiable and the bound is attained,  
  the probability distribution  $\wp^{n}_{t_0,t|M_n}$  is approximated
 (in the pointwise sense) by a normal distribution with average zero and variance $\frac{1}{n J_{t_0}}$. 
 Using this fact, we will show that there exists a sequence
of POVMs  that attains the equality (\ref{TY}) at all points uniformly.   The optimal sequence of POVMs is constructed explicitly in Section \ref{sec:multi}.

\subsection{Comparison between local asymptotic covariance and other conditions}

We conclude the section by  discussing the relation between  asymptotic covariance and  other conditions that are often imposed on measurements. This subsection  is not necessary for understanding the technical results in the next sections and can be skipped at a first reading.

Let us start with the unbiasedness condition.  Assuming unbiasedness, one can derive  the quantum  Cram\'er-Rao  bound  on the MSE \cite{holevo-book}. 
Holevo showed the attainability of the quantum Cram\'er-Rao  bound when estimating displacements in Gaussian systems \cite{holevo-book}.   

The  disadvantage of  unbiasedness is that it is  too restrictive, as  it is satisfied only by  a small class of measurements. 
Indeed, the unbiasedness condition for the estimator $M$ requires 
the condition
$ \Tr \left[E \frac{d^i \rho_t}{d t^i}\right]\, |_{t=t_0}=0$ for $i \ge 2$ with 
$ E:= \int \hat{t} M(d\hat{t})$ as well as the condition $ \Tr\left[   E \frac{d \rho_t}{d t}\right] \,|_{t=t_0}=1$.
   In certain situations, the above conditions might be incompatible.
    For example,
 consider a family of  qubit states  $\rho_t:=\frac12\left(I+\vec{n}_t\cdot\vec{\sigma}\right)$. 
When the Bloch vector $\vec{n}_t$ has a non-linear dependence on $t$ and
the set of higher order derivatives $\frac{d^i \rho_t}{d t^i}|_{t=t_0}$ with 
$i \ge 2$ spans the space of traceless Hermittian matrices, 
no unbiased estimator can exist.  
In contrast,  local asymptotic covariance is only related  to 
the first  derivative $\frac{d \rho_t}{d t}|_{t=t_0}$ 
because 
the contribution of higher order derivatives to the variable $\hat{t}_n$
has order $o\left( \frac{1}{\sqrt{n}}\right)$ and 
vanishes under the condition of the local asymptotic covariance.

One can see that the unbiasedness condition implies local asymptotic covariance 
with the parameterization 
$\rho_{t_0+\frac{t}{\sqrt{n}}}$ 
in the following sense. 
When we have $n$ (more than one) input copies, 
we can construct unbiased estimator by applying 
a single-copy unbiased estimator $M$ satisfying Eq. (\ref{unbiased}) to all copies
as follows.
For the $i$-th outcome $x_i$, we take the rescaled average 
$
\sum_{i=1}^n \frac{x_i}{n}$, 
which satisfies the unbiasedness \eqref{unbiased} for the parameter $t$
as well.
When the single-copy estimator $M$ has variance $v$ at $t_0$, 
{which is lower bounded by the Cram\'{e}r-Rao inequality,} 
this estimator has  variance $v/n$ at $t_0$.
In addition,
the average \eqref{BUL}
 of the obtained data satisfies the  local asymptotic covariance 
because 
the rescaled estimator $\sqrt{n} ( (\sum_{i=1}^n \frac{x_i}{n})-t_0)$
follows the Gaussian distribution with variance $v$
in the large $n$ limit by the central limit theorem;
the center of the Gaussian distribution is pinned at the true value of the parameter by unbiasedness; 
the shape of the Gaussian is independent of the value $t$ and depends only on $t_0$; thus locally asymptotic covariance holds.

The above discussion can be extended to the  multiple-copy case as follows.
  Suppose that $M_{\ell}$ is an unbiased measurement for the $\ell$-copy state
$\rho_{t_0+\frac{t}{\sqrt{n}}}^{\otimes \ell}$ with respect to the parameter $t_0+\frac{t}{\sqrt{n}}$, 
where $\ell$ is an arbitrary finite integer.   From the measurement  $M_{\ell}$  we can construct a measurement for the  $n$-copy state with $n= k\ell+i$ and $i <\ell $
by applying the measurement $M_\ell$ $k$ times and discarding the remaining $i$ copies.   In the following, we consider the limit where the total number $n$ tends to infinity, while $\ell$ is kept fixed.
When the variance of $M_\ell$ at $t_0$ is $v/\ell$,
the average 
\begin{align}
\sum_{i=1}^k \frac{x_i}{k}\label{BUL}
\end{align}
of the $k$ obtained data $x_1, \ldots, x_k$ satisfies local asymptotic covariance, i.e., 
the rescaled estimator $\sqrt{n} ((\sum_{i=1}^k \frac{x_i}{k})-t_0)$
follows the Gaussian distribution with variance $v$ in the large $n$ limit.
Therefore, for any unbiased estimator, there exists an estimator satisfying locally asymptotic covariance that has the same variance.


Another common condition, less restrictive than unbiasedness, is  \emph{local unbiasedness}. This condition depends on the true parameter $t_0$ and consists of the following two requirements 
\begin{align}\Label{asy-unbiased3}
\int \hat{t}\Tr \left[  \,  \rho_{t}^{\otimes \ell}\,  M_\ell(d\hat{t})\, \right  ] \, \Big|_{t=t_0}=t_0\\
\Label{asy-unbiased4}
\frac{d}{d t}\int \hat{t}\Tr\left[  \,  \rho_{t}^{\otimes \ell} \, M_\ell(d\hat{t})\right] \,\Big|_{t=t_0}=1\, , 
\end{align}
{where $\ell$ is a fixed, but otherwise arbitrary, integer. }
The derivation of the quantum Cram\'er-Rao bound still holds, because it uses only the condition  (\ref{asy-unbiased4}). 
    When  the parametrization $\rho_t$ is  $C^1$ continuous, the first derivative $\frac{d}{d t}\int \hat{t}\Tr \left[  \rho_{t}^{\otimes \ell}M_\ell (d\hat{t})\right]$ is continuous at $t=t_0$, and
the locally unbiased condition at $t_0$ yields
the local asymptotic covariance at $t_0$ in the 
way as Eq. \eqref{BUL}.    

Another relaxation of the unbiasedness  condition is  \emph{asymptotic unbiasedness}  \cite{hayashi2005statistical}
\begin{align}\Label{asy-unbiased1}
\lim_{n\to\infty}\int \hat{t}\Tr \left[ \,  \rho_{t}^{\otimes n} \, M_n(d\hat{t})\right] =t  \qquad &\forall t\in\set \Theta\\
\Label{asy-unbiased2}\lim_{n\to\infty}\frac{d}{d t}\int \hat{t}  \,  \Tr\left[  \,  \rho_{t}^{\otimes n}\, M_n(d\hat{t})\right] =1 \qquad   &\forall t\in\set \Theta \, .
\end{align}
The condition of asymptotic unbiasedness leads to a precision bound on MSE \cite[Chapter 6]{hayashi2017quantum}. 
The bound is given by the SLD Fisher information,  and therefore it is  attainable for Gaussian states.  
  However,    no attainable bound for qudit systems has been derived so far under the condition of asymptotic unbiasedness. 
Interestingly, one cannot directly use the  attainability for Gaussian systems to derive an attainability result for qudit systems,   despite the  asymptotic equivalence between Gaussian systems and qudit systems  stated by  {\em quantum local asymptotic normality (Q-LAN) }(see \cite{guta-lan,qlan} and Section \ref{subsec:condition}).  
The problem is that the error of Q-LAN goes to 0 for large $n$,  but the error in the derivative may not go to zero, and therefore  the condition (\ref{asy-unbiased2}) is not guaranteed to hold.  

In order to guarantee attainability of the quantum Cram\'er-Rao bound, one could  think of further loosening the condition of the asymptotic unbiasedness. An attempt to avoid the problem of the Q-LAN error could be to remove  condition (\ref{asy-unbiased2}) and keep only condition (\ref{asy-unbiased1}). 
This leads to an enlarged class of estimators, called {\em weakly asymptotically unbiased}.    The problem with these estimators is that no general MSE bound is known to hold at every point $x$.  For example,  one can find   \emph{superefficient} estimators \cite{hayashiCMP,lecam}, which violate  the Cram\'er-Rao bound  on a set of points.    Such a set must be of zero measure in the limit $n\to\infty$, but the violation of the bound may occur in a considerably large set when $n$ is finite.
In contrast,  local asymptotic covariance  guarantees the MSE  bound \eqref{TY} at every point $t$ where the local asymptotic convariance condition is satisfied.

\begin{table}
\begin{center}
    \begin{tabular}{ | p{2cm} | l | p{4cm} |}
    \hline
    {\bf Condition} & {\bf Definition} & {\bf Limitation} \\ \hline
    Unbiasedness  & $\int \hat{t}\Tr \left[  \, \rho_{t}^{\otimes \ell} \, M(d\hat{t})\right]=t \qquad \qquad \qquad \qquad  \forall t\in \set \Theta$ &  Too restrictive (unbiased estimator may not exist). \\
     \hline
    {Local}  &
    {\parbox[t]{5cm}{$\displaystyle \int \hat{t}\Tr \left[  \, \rho_{t_0}^{\otimes \ell}\,  M_\ell(d\hat{t})\right ] \, =t_0$}} 
    & Estimator depends on  \\
     {unbiasedness} &       &  the true parameter $t_0$.\\ 
     & 
    {$\displaystyle\frac{d}{dt}\int \hat{t}\Tr\left[  \, \rho_{t}^{\otimes \ell} \,M_\ell(d\hat{t})\,  \right]  |_{t=t_0}=1$}  & 
      \\ 
      \hline     
    {Asymptotic}  &{{$\displaystyle\lim_{n\to\infty}\int \hat{t}\Tr\left[ \, \rho_{t}^{\otimes n}\, M_n(d\hat{t}) \, \right] =t \quad ~ \qquad \forall t\in\set \Theta$}} & Attainability unknown  \\
     {unbiasedness} &&  for finite-dimensional systems. \\ 
     & {$\displaystyle\lim_{n\to\infty}\frac{d}{dt}\int \hat{t}\Tr \left[ \,\rho_{t}^{\otimes n} \, M_n(d\hat{t})\, \right]=1  \qquad \forall t \in \set \Theta$}  & \\
       \hline
    Weak  
    & \multirow{3}{*}{$\displaystyle\lim_{n\to\infty}\int \hat{t} \, \Tr\left[\,  \rho_{t}^{\otimes n} \, M_n(d\hat{t}) \,\right]=t  \qquad \quad ~\forall t\in\set \Theta$} 
    & \multirow{3}{*}{No lower bound to the MSE }\\
asymptotic &&   \multirow{3}{*}{is known to hold at every point.}\\ 
unbiasedness    &&\\
    \hline
    \end{tabular}\caption{Alternative conditions for deriving MSE bounds.}
\end{center}
\end{table}

\section{Holevo bound and Gaussian states families}\Label{s3}
\subsection{Holevo bound}\Label{s31} 
When studying  multiparameter estimation in quantum systems, 
we need to address the tradeoff between the precision of estimation of different parameters.
 This is done using two types of quantum extensions of Fisher information matrix: the SLD and the right logarithmic derivative (RLD). 
 
Consider  a multiparameter family of density operators  $\{\rho_{\vec{t}}\}_{\vec{t} \in  \set \Theta}$, where $\set\Theta$ is an open set in $\R^k$,  $k$ being  the number of parameters.   
Throughout this section, we assume that $\rho_{\vec{t}_0}$ is invertible 
and that  the parametrization is  $C^1$ in all parameters.   
Then,  the SLD $L_j$ and the RLD $\tilde{L}_j$ for the parameter $t_j$ are defined through the following equations
\begin{align*}
\frac{\partial \rho_{\vec{t}}}{\partial t_j}=\frac12\left(\rho_{\vec{t}}{L}_j+{L}_j\rho_{\vec{t}}\right), \quad
\frac{\partial \rho_{\vec{t}}}{\partial t_j}=\rho_{\vec{t}}\tilde{L}_j \, ,
\end{align*}
see {\em e.g.}    \cite{helstrom-book,holevo-book} and \cite[Section II]{hayashi2008asymptotic}.  
It can be seen from the definitions that the SLD $L_j$ can always be chosen to be Hermitian, 
while the RLD $\tilde{L}_j$ is in general not  Hermitian.

The SLD quantum Fisher information matrix $J_{\vec{t}}$ and
the RLD quantum Fisher information matrix $\tilde{J}_{\vec{t}}$ are the $k\times k$ matrices
defined as 
\begin{align}\Label{SLD-QFI}
\left(J_{\vec{t}}\right)_{ij}:=\Tr  \left[  \frac {\rho_{\vec{t}}(L_iL_j+L_jL_i)}2\right] ,\quad
\left(\tilde{J}_{\vec{t}}\right)_{ij}:=\Tr  \left[  
\tilde{L}_j^\dagger \rho_{\vec{t}} \tilde{L}_i \right] \, .
\end{align}
Notice that the SLD quantum Fisher information matrix $J_{\vec{t}}$ is a real symmetric matrix,
but
the RLD quantum Fisher information matrix $\tilde{J}_{\vec{t}}$ is not a real matrix in general.

A POVM $M$ is called an unbiased estimator for the family $\spc{S}=\{\rho_{\vec{t}}\}$ 
when the relation
\begin{align*}
\vec{t}=E_{\vec{t}}(M):=\int \vec{x} \Tr \big[ \rho_{\vec{t}} M(d \vec{x})\big] 
\end{align*} 
holds for any parameter $\vec{t}$.
For a POVM $M$, we define the mean square error (MSE) matrix $V_{\vec{t}}(M)$ as
\begin{align}
\Big( \, V_{\vec{t}}(M) \, \Big)_{i,j}:=\int (x_i- t_i) (x_j- t_j)
\Tr \big[  \, \rho_{\vec{t}} M(d \vec{x})\, \big] 
\end{align}

It is known that an unbiased estimator $M$ satisfies
the SLD type and RLD type of Cramer-Rao inequalities  
\begin{align}
V_{\vec{t}}(M)\ge J_{\vec{t}}^{-1}, \qquad {\rm and}  \qquad
V_{\vec{t}}(M)\ge \tilde{J}_{\vec{t}}^{-1} \, ,
\end{align}
respectively \cite{holevo-book}. 
Since it is not always possible  to minimize the MSE matrix under the unbiasedness condition, we minimize 
the weighted MSE $\tr  \big[ W V_{\vec{t}}(M) \big]$ for a given weight matrix $W \ge 0$, where $\tr$ denotes the trace of $k\times k$ matrices.
When a POVM $M$ is unbiased, one has the {\em RLD bound} \cite{holevo-book}  
\begin{align}
 \tr \big[ W V_{\vec{t}}(M) \big]   \ge 
 \map{C}_{{\rm R},{\cal M}}(W,\vec{t})  
 \end{align}
 with
 \begin{align}
  \nonumber  \map{C}_{{\rm R},{\cal M}}(W,\vec{t})  &:=
\min \Big\{\tr [V W]  ~\Big|~  V \ge \tilde{J}^{-1}_{\vec{t}}\Big\}  \\
&  =
\tr \Big [ \, W {\sf Re} (\tilde{J}_{\vec{t}})^{-1}\, \Big]  + \tr\Big| \sqrt{W} {\sf Im} (\tilde{J}_{\vec{t}})^{-1} \sqrt{W}\Big| \, . 
\Label{BFR}
\end{align}
In particular, when $W>0$,
the lower bound  \eqref{BFR} is attained  by the matrix 
$V= {\sf Re} (\tilde{J}_{\vec{t}})^{-1} + \sqrt{W}^{-1} | \sqrt{W} {\sf Im} (\tilde{J}_{\vec{t}})^{-1} \sqrt{W}|\sqrt{W}^{-1} $.

The RLD bound has a particularly tractable  form when the model is D-invariant:  
\begin{defi}  The model $\{\rho_{\vec{t}}\}_{\vec{t} \in  \set \Theta}$ is
\emph{D-invariant at $\vec{t}$}  when the space spanned by the SLD operators is invariant under the linear map  $\map{D}_{\vec{t}}$. For any operator $X$, $\map{D}_{\vec{t}}(X)$ is defined via the following equation 
\begin{align}
\frac { \rho_{\vec{t}}\map{D}_{\vec{t}}(X)+\map{D}_{\vec{t}}(X)\rho_{\vec{t}} } 2=i[X,\rho_{\vec{t}}]
\end{align} 
where $[A,B]  =  AB-BA$ denotes the  commutator. 
When the model is D-invariant at any point,
it is simply called D-invariant.
\end{defi}

For a D-invariant model, the RLD quantum Fisher information can be computed in terms of  the  {\em  D-matrix}, namely    the skew-symmetric matrix defined as 
\begin{align}\Label{D-m}
(D_{\vec{t}})_{j,k}:=i\Tr \Big [  \rho_{\vec{t}}[L_j, L_k] \Big]  \, .
\end{align}
Precisely, the RLD quantum Fisher information has the expression \cite{holevo-book} 
\begin{align}\Label{RLD-QFI}
\left(\widetilde{J}_{\vec{t}}\right)^{-1}=\left(J_{\vec{t}}\right)^{-1}+\frac{i}{2}\left(J_{\vec{t}}\right)^{-1}D_{\vec{t}}\left(J_{\vec{t}}\right)^{-1} \, .
\end{align}
Hence, \eqref{BFR}  becomes
\begin{align}\Label{RLD-QFIMRT}
\map{C}_{{\rm R},{\cal M}}(W,\vec{t})  = \tr  \left[  \, W \left(J_{\vec{t}}\right)^{-1}\right]  +
\frac{1}{2} \tr \left|\sqrt{W}\left(J_{\vec{t}}\right)^{-1}D_{\vec{t}}\left(J_{\vec{t}}\right)^{-1}
\sqrt{W}\right|.
\end{align}
For D-invariant models, the RLD bound is larger  and thus it is a better bound than the bound derived by using the SLD Fisher information matrix (the SLD bound).
However, in the one-parameter case, when the model is not D-invariant,
the RLD bound is not tight, and it is common to use the SLD bound in the one-parameter case.
Hence, both quantum extensions of the Cram\'er-Rao bound have advantages and disadvantages.

To unify both extensions, 
Holevo \cite{holevo-book} derived the following bound, which 
improves the RLD bound when the model is not D-invariant.
For 
a $k$-component vector $\vec{X}$ of operators, define the $k\times k$ matrix $Z_{\vec{t}}(\vec{X})$  as
\begin{align}\Label{HYU}
 \Big(Z_{\vec{t}}(\vec{X})\Big)_{ij}:=
 \Tr  \Big[  \rho_{\vec{t}}X_iX_j \Big] \, .
\end{align}
Then, Holevo's bound is as follows: for any weight matrix $W$, one has 
\begin{align}
\nonumber \inf_{M \in \set{UB}_{{\cal M}}} \tr \Big [  W V_{\vec{t}}(M) \Big]  
&\ge \map{C}_{{\rm H},{\cal M}}(W,\vec{t})  \\
 &:= \min_{\vec{X}} \min_V \Big \{\tr \big[ W V\big] ~ \Big|~ V \ge Z_{\vec{t}}(\vec{X})\Big \} \Label{H-quantity1} \\
&= \min_{\vec{X}} 
\tr  \Big[  W  {\sf Re} (Z_{\vec{t}}(\vec{X}))  \Big]  +
\tr \Big|\sqrt{W} {\sf Im}  (Z_{\vec{t}}(\vec{X}))\sqrt{W}\Big| \Label{H-quantity2},
\end{align}
where
$\set{UB}_{{\cal M}}$ denotes the set of all unbiased measurements under the model ${\cal M}$,
$V$ is a real symmetric matrix, and 
$\vec{X}=(X_i)$ is a $k$-component vector of Hermitian operators  satisfying
\begin{align}
 \Tr\left. \left[  
 X_i
 \frac{\partial\rho_{\vec{t}}}{\partial t_j}
  \right]  \right|_{\vec{t}=\vec{t}_0 }  =\delta_{ij}, \ \forall i, j \le k.\Label{BCR-1}
\end{align}
$\map{C}_{{\rm H},{\cal M}}(W,\vec{t})$ is called the Holevo bound. 
When $W>0$, there exists 
a vector $\vec{X}$ achieving the minimum in  \eqref{H-quantity1}.
Hence, similar to the RLD case,
the equality in \eqref{H-quantity1} holds for $W>0$
only when 
\begin{align}
 V_{\vec{t}}(M) =  {\sf Re}\,   (Z_{\vec{t}}(\vec{X}))+
\sqrt{W}^{-1} |\sqrt{W} {\sf Im  }  \, (Z_{\vec{t}}(\vec{X}))\sqrt{W}| \sqrt{W}^{-1}.\Label{BCR}
\end{align}
Moreover, we have the following proposition.
\begin{prop}[\protect{\cite[Theorem 4]{hayashi2008asymptotic}}]\Label{LGT}
Let $\spc{S}=\{\rho_{\vec{t}}\}_{\vec{t}\in \set{\Theta}}$  be a generic $k$-parameter qudit model and let $\spc{S}'=\{\rho_{\vec{t},\vec{p}}\}_{\vec{t},\vec{p}}$ be a 
$k'$-parameter  model containing $\spc{S}$ 
as $\rho_{\vec{t}}=\rho_{(\vec{t},\vec{0})}$.
When $\spc{S}'$ is D-invariant and the inverse $\tilde{J}_{\vec{t}'}^{-1} $ of 
the RLD Fisher information matrix of the model $\spc{S}'$ exists, we have
\begin{align}
\map{C}_{{\rm H},{\cal M}}(W,\vec{t})
&= \min_{\vec{X}:{\cal M}'} \min_V \Big \{\tr \big[W V \big]  ~ \Big |~ V \ge Z_{\vec{t}}(\vec{X})\Big \} \Label{H-quantity11}\\
&=\min_{P}\left\{\tr \left [  P^TWP(J_{\vec{t}'}^{-1})\right] +\frac12\tr\left|\sqrt{P^TWP}J_{\vec{t}'}^{-1}
D_{\vec{t}'}J_{\vec{t}'}^{-1}\sqrt{P^TWP}\right|\right\}.\Label{Ho2}
\end{align}
In \eqref{H-quantity11},
$\min_{\vec{X}:{\cal M}'}$ denotes the minimum for vector $ \vec{X}$ whose components $X_i$
are  linear combinations of the SLDs operators in the model ${\cal M}'$.
In \eqref{Ho2}, the minimization is taken over all $k\times k'$ matrices satisfying the constraint $(P)_{ij}:=\delta_{ij}$ for $i,j\le k$,  $J_{\vec{t}'}$ 
and $D_{\vec{t}'}$ are the SLD Fisher information matrix  and the D-matrix [cf. Eqs. (\ref{SLD-QFI}) and (\ref{D-m})] for the extended model  $\spc{S}'$
at $\vec{t}':=(\vec{t},\vec{0})$.
\end{prop}

The Holevo bound is always tighter than the RLD bound:  
\begin{align}
\map{C}_{{\rm R},{\cal M}}(W,\vec{t})
\le\map{C}_{{\rm H},{\cal M}}(W,\vec{t}) \,.
\end{align}
The equality holds if and only if the model ${\cal M}$ is D-invariant \cite{suzuki2016explicit}.

In the above proposition, it is not immediately clear whether the Holevo bound depends on the choice of the extended model $\spc{S}'$. In the following, we show that there is a minimum D-invariant extension of $\spc{S}$, and thus the Holevo bound is independent of the choice of $\spc{S}'$.
The minimum D-invariant subspace in the space of Hermitian matrices is given as follows.
Let ${\cal V}$ be the subspace spanned SLDs $\{L_i\}$ of the original model ${\cal M}$
at $\rho_{\vec{t}}$.
Let ${\cal V}'$ be the subspace spanned by $\cup_{j=0}^{\infty}{\cal D}_{\vec{t}}^j({\cal V})$.
Then, the subspace ${\cal V}'$ is D-invariant and contains ${\cal V}$.
What remains is to show that ${\cal V}'$ is the minimum D-invariance subspace.
Let ${\cal V}''$ be the orthogonal space with respect to ${\cal V}'$ for the inner product defined 
by $\Tr \rho X^\dagger Y$.
We denote  by $P'$ and $P''$ the projections into $\cal V'$ and $\cal V''$ respectively.
Each component $X_i$ of a vector of operators $\vec{X}$ can be expressed as $X_i=P'X_i + P''X_i$.
Then, the two vectors $\vec{X}':=(P' X_i)$ and $\vec{X}'':=(P'' X_i)$
satisfy the inequality 
$Z_{\vec{t}}(\vec{X})=Z_{\vec{t}}(\vec{X}')+Z_{\vec{t}}(\vec{X}'')
\ge Z_{\vec{t}}(\vec{X}')$.
Substituting Eq. (\ref{SLD-QFI}) into Eq. (\ref{BCR-1}) and noticing that $P''X_i$ has no support in $\cal V$, we get that only the part $P'X_i $ contributes the condition \eqref{BCR-1} and  
the minimum in \eqref{H-quantity2} is attained when $\vec{X}''=0$.
Hence, the minimum is achieved when each component of the vector $\vec{X}$ is included in 
the minimum D-invariant subspace ${\cal V}'$.
Therefore, since the minimum D-invariant subspace can be uniquely defined,
the Holevo bound does not depend on the choice
of the D-invariant model $\spc{S}'$ that extends $\spc{S}$.

\subsection{Classical and quantum Gaussian states}\Label{s32}
For a  classical system of dimension $d^{\rm C}$, a Gaussian state is  
a $d^{\rm C}$-dimensional normal distribution $N[\vec{\alpha}^{\rm C},\Gamma^{\rm C}]$ 
with mean $\vec{\alpha}^{\rm C}$ and covariance matrix $\Gamma$. The corresponding 
 random variable will be denoted as  $\vec {Z}=(Z_1, \ldots, Z_{d^{\rm C}})$ and will take values $\vec{z}  =  (z_1, \dots, z_{d^{\rm C}})\in \R^{d^{\rm C}}$.

 For quantum systems we will restrict our attention to a subfamily of Gaussian states, known as displaced thermal states. 
     For a quantum system made of a single mode,   the {\em displaced thermal states} are  defined as
\begin{align}\Label{D-alpha}
\rho_{\alpha,\beta}:=T^{\rm Q}_{\alpha}\,\rho^{{\rm thm}}_{\beta}(T^{\rm Q}_{\alpha})^\dag, \qquad T^{\rm Q}_\alpha=\exp(\alpha \hat{a}^\dag-\bar{\alpha} \hat{a})  \, ,
\end{align}
where $\alpha \in \C$ is the {\em displacement},  $T^{\rm Q}_\alpha$ is the {\em displacement operator},  $\hat{a}$ is the annihilation operator satisfying the relation $[\hat a, \hat a^\dag] =1$, and  $\rho^{{\rm thm}}_{\beta}$ is a {\em thermal state},  defined as
\begin{align}
\rho^{{\rm thm}}_{\beta}&:=(1-e^{-\beta})\sum_{j=0}^{\infty}e^{-j\beta}\,  |j\>\<j| \, ,
\end{align}
where  the basis $\{  |j\> \}_{j \in \N}$ consists of the eigenvectors of $\hat a^\dag \hat a$ and $\beta \in  (0,\infty) $ is a real parameter, hereafter  called the {\em thermal parameter}.

For a quantum system of $d^{\rm Q}$ modes, the  products of single-mode displaced thermal states 
 will be denoted as
\begin{align}\Label{multidisplaced-thermal}
\Phi[\vec{\alpha}^{\rm Q} , \vec{\beta}^{\rm Q}]:= \bigotimes_{j=1}^{d^Q} \rho_{\alpha_j,\beta_j} \,.  
\end{align}
where $\vec{\alpha}^{\rm Q}  = (\alpha_j)_{j=1}^{d^{\rm Q}}$ is the vector of displacements and $\vec{\beta}^{\rm Q}  = (\beta_j)_{j=1}^{d^{\rm Q}} $ is the vector of thermal parameters.    In the following we will regard $\vec{\alpha}$ as a vector in $\R^{2 d^{\rm Q}}$, using the notation $\vec \alpha  =  (\alpha^{\rm R}_1,  \alpha^{\rm I}_1 , \cdots,  \alpha^{\rm R}_{d^{\rm Q}},   \, \alpha^{\rm I}_{d^{\rm Q}})$, $\alpha_j^{\rm R}  :  ={\sf Re}  (\alpha_j)$, $\alpha_j^{\rm I}  :  =  {\sf  Im}   ( \alpha_j)$.

For a hybrid system of $d^{\rm C}$  classical variables and $d^{\rm Q}$ quantum modes, we define  the {\it canonical classical-quantum (c-q) Gaussian states} $G[\vec{\alpha},\Gamma]$ as
\begin{align}\Label{Gaussian-state}
G[\vec{\alpha},\Gamma]:=N[\vec{\alpha}^{\rm C},\Gamma^{\rm C}]\otimes 
\Phi[\vec{\alpha}^{\rm Q},\vec{\beta}^{{\rm Q}}] \, ,
\end{align}
with $\vec{\alpha}= \vec{\alpha}^{\rm C} \oplus \vec{\alpha}^{\rm Q}  \in  \R^{d^{\rm C} + 2 d^{\rm Q}}$, $\vec{\alpha}^{\rm C} \in\R^{d^{\rm C}}$, $\vec{\alpha}^{\rm Q}  \in  \R^{2d^{\rm Q}}$, and ${\Gamma}= \Gamma^{\rm C}  \oplus  \Gamma^{\rm Q}$,   
where   
\begin{align}
\nonumber  \Gamma^{\rm Q}  &:=
E_{d^{\rm Q}}\left(\vec{N}\right)+\frac{i}{2}\Omega_{d^{\rm Q}},  \\
  E_{d^{\rm Q}}(\vec{N}) &:= 
\left(
\begin{array}{ccccc}
N_{1} & 0 & & &  \\
0 &N_{1} & & &  \\
 &      & \ddots & &  \\
 &      &  & N_{d^{\rm Q}}&0  \\
 &      &  & 0&N_{d^{\rm Q}}
 \end{array}
\right)    
 \qquad N_{j}:=\frac{e^{-\beta_{j}}}{1-e^{-\beta_{j}}}  \label{M1}\\
 \Omega_{d^{\rm Q}}   &:=
\left(
\begin{array}{ccccc}
0 & 1 & & &  \\
-1 & 0 & & &  \\
 &      & \ddots & &  \\
 &      &  & 0&1  \\
 &      &  & -1&0 
 \end{array}
\right)\quad
.\label{M2}
\end{align}

Equivalently, the canonical Gaussian states can be expressed as 
 \begin{align}
 G[\vec{\alpha},{\Gamma}]=T_{\vec{\alpha}} \, G[\vec{0},{\Gamma}] \, T^\dag_{\vec{\alpha}} \, ,
 \end{align}
where $T_{\vec{\alpha}}$ is the Gaussian shift operator
\begin{align}\Label{Gaussian-shift-operator}
 T_{\vec{\alpha}}:=\left(\bigotimes_{k=1}^{d^C} T^{\rm C}_{\alpha^C_k}\right)\otimes\left(\bigotimes_{j=1}^{d^{\rm Q}}T^{\rm Q}_{\alpha^{\rm Q}_j }\right)\,,
 \end{align}
$T^{\rm Q}_{\alpha_k^{\rm Q}}$ is given by Eq. (\ref{D-alpha}), and $T^{\rm C}_{\alpha_j^{\rm C}}$ is the map $z_j\to z_j +  \alpha_j^{\rm C}$.  
   For the classical part, 
  we have adopt the notation 
  \begin{align*}
 \Tr  \left[ N[\vec{\alpha}^{\rm C},\Gamma^{\rm C}]    \, \exp\left(\sum_{j=1}^{d^{\rm C}} i \xi_j  Z_j\right)     \right]: =   \int  d z_1 \dots  d z_{d^{\rm C}}    \,  N[\vec{\alpha}^{\rm C},\Gamma^{\rm C}]      (\vec{z})       \,   \exp\left(\sum_{j=1}^{d^{\rm C}} i \xi_j  z_j\right)     \, .
 \end{align*}
 With this notation, the canonical Gaussian state $ G[\vec{\alpha},{\Gamma}]$ is uniquely identified by  the characteristic equation \cite{holevo-book}
\begin{align}
\Tr  \left[   G[\vec{\alpha},{\Gamma}]  \, \exp\left(\sum_{j=1}^d i \xi_j R_j\right)\right]  
= \exp \left [i \sum_{j} \xi_j \alpha_j - \frac{1}{2} \sum_{j,k} \xi_j\xi_k  \,  {\sf Re}  [ \Gamma_{j,k}]  \right] \, ,\label{MFT}
\end{align}
with  
\begin{align}
\nonumber d    &:  =  d^{\rm C}  + 2  d^{\rm Q} \\
\nonumber R_j & :=     Z_j  \, , \qquad \forall j \in \{1,\dots,  d^{\rm C}\}   \\
 R_{2j-1}  &:=Q_j  :  =  \frac {a_j+ a_j^\dag }{\sqrt 2}   \, ,\quad  R_{2j}: =P_j  :  =  \frac {a_j - a_j^\dag }{\sqrt 2 i}  \, ,\qquad \forall j \in \left\{d^{\rm C}  + 1  , \dots,  d\right\} \, . 
 \end{align}
The formulation in terms of the characteristic equation \eqref{MFT} can be used  to generalize the notion of canonical Gaussian state \cite{Petz}.  
Given a $d$-dimensional Hermitian matrix (correlation matrix) $\Gamma={\sf Re} ( \Gamma)+i{\sf Im} ( \Gamma)$ whose real part ${\sf Re} ( \Gamma)$ is positive semi-definite, we define  the operators $\vec{R}:=(R_1, \ldots, R_d)$ 
 via the commutation relation 
\begin{align}
[R_k, R_j]= i {\sf Im}  ( \Gamma_{k,j}).
\end{align}
We define the {\it general Gaussian state} $ G[\vec{\alpha},{\Gamma}]$ on the operators
$\vec{R}$ as the linear functional 
on the operator algebra generated by $R_1, \ldots, R_d$ 
satisfying the characteristic equation \eqref{MFT} \cite{Petz}. Note that, although $\Gamma$ is not necessarily positive semi-definite, 
its real part ${\sf Re} (\Gamma)$ is positive semi-definite.
Hence, the right-hand-side of Eq. (\ref{MFT}) is contains a negative semi-definite quadratic form, in the same way as for the standard Gaussian states.

For general Gaussian states, we have the following lemma.
\begin{lem}\Label{3LL}
Given a Hermitian matrix $\Gamma$,
there exists an invertible real matrix $ T$ such that
the Hermitian matrix $T \Gamma T^{T} $ is the correlation matrix of a canonical Gaussian state.
In particular, when ${\rm Im} (\Gamma)$ is invertible, 
$T \Gamma T^{T} = E_{d^Q}\left(\vec{N}\right)+\frac{i}{2}\Omega_{d^Q}$
and the vector $\beta$ is unique up to the permutation.   
Further, when two matrices  $T$ and $\tilde{T}$ satisfy the above condition,
the canonical Gaussian states 
$G[T^{-1}\vec{\alpha}, T \Gamma T^{T}]$
and
$G[\tilde{T}^{-1}\vec{\alpha}, \tilde{T} \Gamma \tilde{T}^{T}]$
are unitarily equivalent.
\end{lem}
The proof is provided  in Appendix \ref{A3LL}. 

In the above lemma,  we can transform $\Gamma$ into the block form   $ \Gamma^C  \oplus  \Gamma^Q$
where $\Gamma^C$ is real by applying orthogonal transformation. The unitary operation on the classical part is given as a scale conversion.
Hence, an invertible real matrix $ T$ can be realized by 
the combination of a scale conversion and a linear conversion, 
which can be implemented as a unitary on the Hilbert space. 
Hence, a general Gaussian state can be given as the resultant linear functional 
on the operator algebra 
after the application of the linear conversion to a canonical Gaussian state. 
This kind of construction is unique up to unitarily equivalence. 
Indeed, Petz \cite{Petz} showed a similar statement by using Gelfand-Naimark-Segal (GNS) construction. 
Our derivation directly shows the uniqueness without using the GNS construction.

\begin{lem}\Label{lem-QLAN-QFI1}
The Gaussian states family $\{
G[\vec{\alpha},{\Gamma}] 
\}_{\vec{\alpha}}$ is D-invariant. 
The SLDs are calculated as $L_{\vec{\alpha},j}= \sum_{k=1}^d (( {\sf Re} (\Gamma))^{-1})_{j,k} R_k$.
The $D$-operator at any point $\vec{\alpha}$ is given as
$ \map{D}(R_j)= \sum_{k} 2 {\sf Im} (\Gamma)_{j,k}R_k$.
The inverse of the RLD Fisher information matrix $\tilde{J}_{\vec{\alpha}}$ is calculated as
\begin{align}\Label{QFI-QLAN1}
(\tilde{J}_{\vec{\alpha}})^{-1}=
\Gamma.
\end{align}
\end{lem}
This lemma shows the inverse of the RLD Fisher information matrix is given by the correlation matrix. 

\begin{proof}
Due to the coordinate conversion give in Lemma \ref{3LL},
it is sufficient to show the relation \eqref{QFI-QLAN1} for the canonical Gaussian states family.
In that case, the desired statement  has already  been shown by Holevo in \cite{holevo-book}.
\end{proof}

Therefore, as shown in Appendix \ref{ALNS}, 
a $D$-invariant Gaussian model can be characterized as follows:

\begin{lem}\Label{LNS}
Given an $d \times d $ strictly positive-definite Hermitian matrix $\Gamma=A+ iB$ ($A, B$ are real matrices)
and a $d\times k$ real matrix $T$ with $k \le d$, 
the following conditions are equivalent.
\begin{description}
\item[(1)]
The linear submodel ${\cal M}:=\{G[ T\,\vec{t}, {\Gamma}]\}_{\vec{t}\in \mathbb{R}^{k}}$ (with displacement $T\,\vec{t}$)
is D-invariant.  
\item[(2)]
The image of the linear map $A^{-1}T$ is invariant for the application of $B$.
\item[(3)]
There exist a unitary operator $U$ and a Hermitian matrix $\Gamma_0$ 
such that 
\begin{align}
U G[ T\,\vec{t}, {\Gamma}] U^\dagger
= G[ \vec{t}, \Gamma_T]\otimes G[ 0, {\Gamma}_0],\Label{FHY}
\end{align}
where $\Gamma_T:= (T^T A^{-1} T)^{-1}+ i 
(T^T A^{-1} T)^{-1}(T^T B T)(T^T A^{-1} T)^{-1}$.
\end{description}
\end{lem}

\subsection{Measurements on Gaussian states family}\Label{s33}
We discuss the stochastic behavior of the outcome of the  measurement 
on the c-q system generated by $\vec{R}=(R_j)_{j=1}^d$
when the state is given as a general Gaussian state $G[\vec{\alpha}, {\Gamma}]$.
To this purpose, we introduce the notation
$\wp_{\vec{\alpha}| M}(B):= \Tr \Big [  G  [\vec{\alpha}, {\Gamma}] M(B)\Big] $ for a POVM $M$.
Then, we have the following lemma.
\begin{lem}\Label{ANH}
Let $\vec{X}  =  (X_i)_{i=1}^k$  be  the  vector defined by  $X_i : = \sum_{j=1}^d P_{i,j} R_j$,  where $P$ is a real $k\times d$ matrix.   
For a weighted matrix $W>0$, there exists a POVM $M_{P|W}^\Gamma$ with outcomes in  $\mathbb{R}^k$
such that
$\int x_i M_{P|W}^\Gamma (d \vec{x})= X_i$
and
$\wp_{\vec{\alpha}| M_{P|W}^\Gamma}$ is the normal distribution  
with average $ (\sum_{j=1}^d P_{i,j}\alpha_j)_{i=1}^k$ and covariance matrix
$$  {\sf Re}  (Z_{\vec{\alpha}}(\vec{X}))+\sqrt{W}^{-1} |\sqrt{W} {\sf Im}  (Z_{\vec{\alpha}}(\vec{X}))\sqrt{W}| \sqrt{W}^{-1}.$$
In this case, the weighted covariance matrix is 
$$\tr  \Big[ W {\sf Re} ( Z_{\vec{\alpha}}(\vec{X})) \Big] +\tr \Big|\sqrt{W} {\sf Im} ( Z_{\vec{\alpha}}(\vec{X}))\sqrt{W} \Big|. $$
\end{lem}
The proof is provided in Appendix \ref{AANH}. 

In the above lemma, when $\vec{X}=\vec{R}$, we simplify $M_{P|W}^{\Gamma}$ to $M_{W}^\Gamma $.
This lemma is useful for estimation in the Gaussian states family 
${\cal M}':=\{G[\vec{t}, {\Gamma}]\}_{\vec{t}\in \mathbb{R}^{d}}$.
In this family, 
we consider the covariant condition.
\begin{defi}
A POVM $M$ is a {\it covariant} estimator for the family 
$\{G[\vec{t}, {\Gamma}]\}_{\vec{t}\in \mathbb{R}^{k}}$
when
the distribution $\wp_{\vec{t}|M}(\set{B}):=\Tr G[\vec{t}, {\Gamma}] M(\set{B})$
satisfies the condition $\wp_{\vec{0}|M}(\set{B})=\wp_{\vec{t}|M}(\set{B}+\vec{t})$
for any $\vec{t}$.
This condition is equivalent to 
\begin{align*}
M(\set{B}+\vec{t})=T_{\vec{t}} M(\set{B}) T_{\vec{t}}^\dagger  \qquad \forall \vec{t} \in \R^k \, .
\end{align*}
\end{defi}
Then, we have the following lemma for this Gaussian states family.

\begin{cor}[\cite{holevo-book}]\Label{LFD}
For any weight matrix $W\ge 0$ and the above Gaussian states family ${\cal M}'$, we have
\begin{align}
& \inf_{M \in \set{UB}_{{\cal M}'}}\tr  \Big[   W V_{\vec{t}}(M) \Big]  
=\inf_{M \in \set{CUB}_{{\cal M}'}}\tr  \Big[ W V_{\vec{t}}(M)\Big]  \nonumber \\
= & \map{C}_{{\rm R},{\cal S'}}(W,\vec{t})
= \tr \Big [  {\sf Re} ( \Gamma) W\Big]   + \tr\Big |\sqrt{W} {\sf Im} ( \Gamma)\sqrt{W}\Big |,\Label{DIT2}
\end{align}
where $\set{CUB}_{{\cal M}'}$ are the sets of covariant unbiased estimators for the model ${\cal M}'$, respectively.
Further, 
when $W>0$, the above infimum is attained by 
the covariant unbiased estimators $M_W^\Gamma$ whose output distribution is 
the normal distribution with average $\vec{t}$ and covariance matrix
 ${\sf Re} ( (\Gamma) +\sqrt{W}^{-1} |\sqrt{W} {\sf Im} ( \Gamma) \sqrt{W}|\sqrt{W}^{-1}$.
\end{cor}

This corollary can be shown as follows.
Due to Lemma \ref{lem-QLAN-QFI1}, the lower bound \eqref{RLD-QFIMRT} 
of the weighted MSE $\tr W V_{\vec{t}}(M) $ of 
unbiased estimator $M$ is calculated as the RHS of \eqref{DIT2}.
Lemma \ref{ANH} guarantees the required performance of $M_{W}^\Gamma $.
To discuss the case when $W$ is not strictly positive definite,
we consider $W_\epsilon:=W+\epsilon I$.
Using the above method, we can construct an unbiased and covariant estimator 
whose output distribution is 
the $2d^Q$-dimensional distribution of average $\vec{t}$ and covariance 
 ${\sf Re} ( (\Gamma) + \sqrt{W_\epsilon}^{-1} 
|\sqrt{W_\epsilon} {\sf Im} ( \Gamma) \sqrt{W_\epsilon}|\sqrt{W_\epsilon}^{-1}$.
The weighted MSE matrix is 
$\tr  \Big[ W {\sf Re} ( \Gamma) \Big]+ 
\tr  \Big[ \sqrt{W_\epsilon}^{-1} W \sqrt{W_\epsilon}^{-1} 
|\sqrt{W_\epsilon} {\sf Im} ( \Gamma) \sqrt{W_\epsilon}| \Big]$, which converges to 
the bound \eqref{DIT2}.

By combining Proposition \ref{LGT}, this corollary can be extended to 
a linear subfamily of 
$k'$-dimensional Gaussian family 
$\{G[\vec{t}', {\Gamma}]\}_{\vec{t}'\in \mathbb{R}^{k'}}$.
Consider a linear map $T$ from $\mathbb{R}^k$ to 
$\mathbb{R}^{k'}$.
We have the following corollary for the subfamily 
${\cal M}:=\{G[ T(\vec{t}), {\Gamma}]\}_{\vec{t}\in \mathbb{R}^{k}}$.

\begin{cor}\Label{LFD2}
For any weight matrix $W\ge 0$, we have
\begin{align}
& \inf_{M \in \set{UB}_{{\cal M}}}\tr \Big[W V_{\vec{t}}(M) \Big] 
=\inf_{M \in \set{CUB}_{{\cal M}}}\tr \Big[  W V_{\vec{t}}(M) \Big] 
= \map{C}_{{\rm H},{\cal M}}(W,\vec{t}).
\Label{DIT3}
\end{align}
Further, 
when $W>0$, 
we choose a vector $\vec{X}$ to realize the minimum in  \eqref{H-quantity11}.
The above infimum is attained by 
the covariant unbiased estimators $M_W$ whose output distribution is 
the normal distribution with average $\vec{t}$ and covariance matrix
 ${\sf Re} ( (Z_{\vec{t}}(\vec{X}))+
\sqrt{W}^{-1} |\sqrt{W} {\sf Im} ( Z_{\vec{t}}(\vec{X}))\sqrt{W}| \sqrt{W}^{-1}$.
\end{cor}

Proposition \ref{LGT} guarantees that 
$\map{C}_{{\rm H},{\cal M}}(W,\vec{t})$ with \eqref{H-quantity11} can be given 
when the components $\vec{X}$ are given a linear combination of $R_1, \ldots, R_{k'}$.
Hence, the latter part of the corollary with $W>0$ follows from \eqref{H-quantity1} and Lemma \ref{ANH},
implies this corollary for $W>0$.
The case with non strictly positive $W$ can be shown by considering $W_\epsilon$ in the same way as Corollary \ref{LFD}.

\section{Local asymptotic normality}\Label{sec:Tool}
The extension from one-parameter estimation to multiparameter estimation is quite nontrivial.
Hence we first develop the concept of local asymptotic normality which is the key tool to constructing the optimal measurement in multiparameter estimation.
Since we could derive the tight bound of MSE for the Gaussian states family,
it is a natural idea to approximate the general case by Gaussian states family, and local asymptotic normality will serve as the bridge between these general qudit families and Gaussian state families.

\subsection{Quantum local asymptotic normality with specific parametrization}\Label{subsec:condition}

For a quantum system of dimension $d<\infty$,  also known as qudit,  we consider {\em generic states},  described by density matrices with  full rank and non-degenerate spectrum.
To discuss quantum local asymptotic normality, we need to define a specific coordinate system.
For this aim, we consider the neighborhood of a fixed  density matrix $\rho_{\vec{\theta}_0}$, assumed to be diagonal in the canonical  basis of $\C^d$, and parametrized as 
\begin{align*}
\rho_{\vec{\theta}_0} =    
 \sum_{j=1}^{d}  \theta_{0,j}\,   |j\>\<j|\,
\end{align*}
with spectrum ordered as  ${\theta}_{0,1}>\cdots>{\theta}_{0,d-1}>{\theta}_{0,d}>0$. 
In the neighborhood of $\rho_{\vec{\theta}_0} $, 
we parametrize the states of the system as
\begin{align}\Label{state}
\rho_{\vec{\theta}_0+\frac{\vec{\theta}}{\sqrt{n}}}=
U_{\vec{\theta}^R,\vec{\theta}^I} \, \rho_0(\vec{\theta}^C) \, U^\dag_{\vec{\theta}^R,\vec{\theta}^I}
\end{align}
for 
$\vec{\theta}  :=  ( \vec{\theta}^C,\vec{\theta}^R,\vec{\theta}^I)$ with 
$(\vec{\theta}^R,\vec{\theta}^I)\in\R^{d(d-1)}$ and $\vec{\theta}^C \in\R^{d-1}$, 
where  $\rho_0(\vec{\theta}^C)$ is the  diagonal density matrix
\begin{align}\Label{rhomu}
\rho_0 (\vec{\theta}^C)  =    
 \sum_{j=1}^{d}  \,  \left(\theta_{0,j}+\frac{{\theta}^C_j}{\sqrt{n}}\right) \,   |j\>\<j|, \,\qquad \theta^C_d:=-\sum_{k=1}^{d-1}\theta^C_k, 
\end{align}
and $U_{\vec{\theta}^R,\vec{\theta}^I}$ is the unitary matrix defined by
\begin{align}\Label{parameters}
U_{\vec{\theta}^R,\vec{\theta}^I} &=
\exp\left[\sum_{1\le j<k\le d}\frac{i\left(
\theta^{\rm I}_{j,k}F^{\rm I}_{j,k}+\theta^{\rm R}_{j,k}F^{\rm R}_{k,j}\right)}{\sqrt{n(\theta_{0,j}-\theta_{0,k})}}\right].
\end{align}
 Here $\vec{\theta}^R$ and $\vec{\theta}^I$ are vectors of real parameters 
 $ ({\theta}^R_{j,k})_{1\le j<k\le d}$ and $( {\theta}^I_{j,k})_{1\le j<k\le d}$,  
 and  $F^{\rm I}$   ($F^{\rm R}$)
 is the matrix  defined by $(F^{\rm I})_{j,k}:=i\delta_{j,k}-i\delta_{k,j}$   ($(F^{\rm R})_{k,j}:=\delta_{j,k}+\delta_{k,j})$, where
$\delta_{j,k}$   is the delta function.
We note that by this definition the components of $\vec{\theta}^R$ and $\vec{\theta}^I$ are in one-to-one correspondence.
The parameter $\vec{\theta}  =  ( \vec{\theta}^C,\vec{\theta}^R,\vec{\theta}^I)$ will be referred to as \emph{the Q-LAN coordinate},
and the state with this parametrization, which was used by Khan and Guta in \cite{qlan,guta-lan,guta-qubit}, will be denoted by $\rho^{\rm KG}_{\vec{\theta}}$.

Q-LAN establishes an asymptotic correspondence between multicopy qudit states and Gaussian shift models. 
Using the parameterization $\vec{\theta}=(\vec{\theta}^C,\vec{\theta}^R,\vec{\theta}^I)$, we have the multicopy qudit models and Gaussian shift models are equivalent in terms of the RLD quantum Fisher information
matrix: 
\begin{lem}\Label{lem-QLAN-QFI}
The RLD quantum Fisher information matrices of the qudit model and the corresponding Gaussian model in Eq. (\ref{ER1}) are both equal to
\begin{align}\Label{QFI-QLAN}
\left(\widetilde{J}^Q_{\vec{\theta}}\right)^{-1}=E_{d(d-1)/2}\left(\vec{\beta}'\right)+\frac{i}{2}\Omega_{d(d-1)/2}\qquad e^{-\beta'_i}=\frac14\coth{\frac{\beta_i}2}.
\end{align}
\end{lem}
The calculations can be found in Appendix \ref{app-QFI}.
The quantum version of local asymptotic normality has been derived in several different forms \cite{qlan,guta-lan,guta-qubit} with applications in quantum statistics \cite{guctua2008optimal,gill2013asymptotic}, benchmarks \cite{guctua2010quantum} and data compression \cite{yang2018compression}. Here we use the version of \cite{qlan}, which states that $n$ identical copies of a qudit state can be locally approximated by a c-q Gaussian state in the large $n$ limit.
  The approximation is in the following sense: 
 
 \begin{defi}[Compact uniformly asymptotic equivalence of models]
For every $n\in\N^*$, let  $\{\rho_{\vec{t},n}\}_{\vec{t}\in\set{\Theta}_n}$ and $\{\widetilde{\rho}_{\vec{t},n}\}_{\vec{t}\in\set{\Theta}_n}$ be two models of density matrices acting on  Hilbert spaces $\spc H$ and $\spc K$  respectively
where the set of parameters $\set{\Theta}_n$ may depend on $n$.
We say that the two families are {\em asymptotically equivalent} for $\vec{t} \in \set{\Theta}_n$, 
denoted as 
$\rho_{\vec{t},n} \cong \widetilde{\rho}_{\vec{t},n}~(\vec{t} \in \set{\Theta}_n)$,  
if there exists a quantum channel $\map{T}_n$  (i.e. a completely positive trace preserving map) mapping trace-class operators on $\spc H$ to trace-class operators on $\spc K$ and a quantum channel $\map{S}_n$ mapping trace-class operators on $\spc K$ to trace-class operators on $\spc H$, which are independent of $\vec{t}$ and satisfy the conditions
\begin{align}
&\sup_{\vec{t}\in\set{\Theta}_n }
\left\|
\map{T}_n \left(\rho_{\vec{t},n}\right)
-\widetilde{\rho}_{\vec{t},n} \right\|_1 \xrightarrow{n\to\infty} 0 \Label{LAN-TB}\\
&\sup_{\vec{t}\in\set{\Theta}_n}
\left\|
\rho_{\vec{t}, n}
-\map{S}_n
\left(\widetilde{\rho}_{\vec{t}, n}\right)
\right\|_1 \xrightarrow{n\to\infty} 0 
\Label{LAN-SB} \, .
\end{align}

Next, we extend 
asymptotic equivalence to 
compact uniformly asymptotic equivalence.
In this extension, we also describe the order of the convergence. 

Given a sequence $\{a_n\}$ converging to zero, for every $\vec{t}'$ in a compact set $\set{K}$ consider two models $\{\rho_{\vec{t},\vec{t}',n}\}_{\vec{t}\in\set{\Theta}_n,
}$ 
and $\{\widetilde{\rho}_{\vec{t},\vec{t}',n}\}_{\vec{t}\in\set{\Theta}_n}$. We say that they
are {\em asymptotically equivalent} for $\vec{t} \in \set{\Theta}_n$
compact uniformly with respect to $\vec{t}'$ with order $a_n$, 
denoted as 
$\rho_{\vec{t},\vec{t}',n} 
\stackrel{\vec{t}'}{\cong} \widetilde{\rho}_{\vec{t},\vec{t}',n}~
(\vec{t} \in \set{\Theta}_n, a_n)$,
if for every $\vec{t}'\in\set{K}$ there exists a quantum channel $\map{T}_{n,\vec{t}'}$  
mapping trace-class operators on $\spc H$ to trace-class operators on $\spc K$ and a quantum channel $\map{S}_{n,\vec{t}'}$ mapping trace-class operators on $\spc K$ to trace-class operators on $\spc H$
such that  
\begin{align}
&\sup_{\vec{t}' \in K} \sup_{\vec{t}\in\set{\Theta}_n }
\|
\map{T}_{n,\vec{t}'} (\rho_{\vec{t},\vec{t}',n})
-\widetilde{\rho}_{\vec{t},\vec{t}',n} \|_1 = O(a_n) \Label{LAN-T}\\
&\sup_{\vec{t}' \in K} \sup_{\vec{t}\in\set{\Theta}_n}
\|
\rho_{\vec{t},\vec{t}',n}
-\map{S}_{n,\vec{t}'}
(\widetilde{\rho}_{\vec{t},\vec{t}',n})
\|_1 = O(a_n).
\Label{LAN-S} 
\end{align}
Notice that the channels
$\map{T}_{n,\vec{t}'}$ and $\map{S}_{n,\vec{t}'} $
depend on $\vec{t}'$ and are independent of $\vec{t}$.
\end{defi} 
In the above terminology, Q-LAN establishes an asymptotic equivalence between families of $n$ copy qudit states and Gaussian shift models. 
 Precisely, one has the following 
\begin{prop}[Q-LAN for a fixed parameterization; Kahn and Guta \cite{qlan,guta-lan}]\Label{lem-QLAN}
For any $x <1/9$, we define the set $\set{\Theta}_{n,x}$ of $\vec{\theta}$ as 
$$\set{\Theta}_{n,x}:=\left\{\vec{\theta}~|~\|\vec{\theta}\|\le n^{x}\right\}$$ ($\|\cdot\|$ denotes the vector norm). 
Then, we have the following compact uniformly asymptotic equivalence;
\begin{align}
(\rho^{\rm KG}_{\vec{\theta}_0+\vec{\theta}/\sqrt{n}})^{\otimes n}
\stackrel{\vec{\theta}_0}{\cong}
G[\vec{\theta},\Gamma_{\vec{\theta}_0}]:=
N[\vec{\theta}^C,\Gamma_{\vec{\theta}_0}^C] \otimes \Phi[(\vec{\theta}^R,\vec{\theta}^I),\vec{\beta}_{\vec{\theta}_0}]
~(\vec{\theta} \in \set{\Theta}_{n,x},n^{-\kappa} ),
 \Label{ER1}
\end{align}
where $\kappa$ is a parameter to satisfy $\kappa\ge 0.027$, and
$N[\vec{\theta}^C,\Gamma_{\vec{\theta}_0}] $ 
is the multivariate normal distribution with mean $\vec{\theta}^C$ and covariance matrix 
$\Gamma_{\vec{\theta}_0,k,l}:= (J_{\vec{\theta}_0}^{-1})_{k,l}$ 
for $k,l=1, \ldots, d-1$, and 
$(\vec{\beta})_{\vec{\theta}_0,j,k}:=\frac{(\rho_{\vec{\theta}_0})_{k,k}}{ (\rho_{\vec{\theta}_0})_{j,j}}$. 
 \end{prop}

The conditions (\ref{LAN-T}) and (\ref{LAN-S}) are not enough to translate precision limits for one family into precision limits for the other. This is because such limits are often expressed in terms of the derivatives of the density matrix, whose asymptotic behaviour is not fixed by (\ref{LAN-T}) and (\ref{LAN-S}). In the following we will establish an asymptotic equivalence in terms of the RLD quantum Fisher information.

\subsection{Quantum local asymptotic normality with generic parametrization}\Label{subsec:condition2}

In the following, we explore to which  extent can we extend Q-LAN in Proposition \ref{lem-QLAN}. Precisely, we derive a Q-LAN  equivalence as in Eq. (\ref{ER1}) which is not restricted to  the parametrization of Eqs. (\ref{rhomu}) and (\ref{parameters}). 

In the previous subsection, we have discussed the specific parametrization given in \eqref{state}.
In the following, we discuss a generic parametrization.
Given an arbitrary D-invariant model $\rho^{\otimes n}_{\vec{t}_0+\frac{\vec{t}}{\sqrt{n}}}$ with vector parameter  $\vec{t}$, 
we have the following theorem.

\begin{theo}[Q-LAN for an arbitrary parameterization]\Label{Th3}
Let $\{\rho_{\vec{t}}\}_{\vec{t}\in \set{\Theta}}$ be a $k$-parameter D-invariant qudit model. 
Assume that $\rho_{\vec{t}_0}$ is a non-degenerate state,
the parametrization is $C^2$ continuous, and $\tilde{J}_{\vec{t}_0}^{-1} $ exists.
Then, there exist a constant $c({\vec{t}_0})$ such that
the set 
\begin{align}
\set{\Theta}_{n,x,c({\vec{t}_0})}:=
\left\{\vec{t}~|~\|\vec{t}\|\le c({\vec{t}_0}) n^{x}\right\}\label{neighborhood-c}
\end{align} with $x <1/9$
satisfies 
\begin{align}
\rho^{\otimes n}_{\vec{t}_0+\frac{\vec{t}}{\sqrt{n}}}
\stackrel{\vec{t}_0}{\cong}
G[ \vec{t},\tilde{J}_{\vec{t}_0}^{-1}]
~(\vec{t} \in \set{\Theta}_{n,x,c({\vec{t}_0})}\cap \mathbb{R}^k
,n^{-\kappa} ),
 \Label{ER7}
\end{align}
where $\tilde{J}_{\vec{t}_0}^{-1}$ is the RLD Fisher information at $\vec{t}_0$ and $\kappa$ is a parameter to satisfy $\kappa\ge 0.027$.
\end{theo}

\begin{proof}
We choose the basis $\{|i\rangle\}_{i=1}^d$ to diagonalize the state $\rho_{\vec{t}_0}$.
We denote the Q-LAN parametrization based on this basis by $\rho^{{\rm KG}|\vec{t}_0}_{\vec{\theta} } $, where this parametrization depends on $\vec{t}_0$.
It is enough to consider the neighborhood $U(\vec{t}_0)$ of $\vec{t}_0$.
There exists a map $f_{\vec{t}_0}$  on $U(\vec{t}_0)$ such that
$ \rho_{\vec{t}_0+\vec{t}}= \rho^{{\rm KG}|\vec{t}_0}_{\vec{\theta}_0(\vec{t}_0)+f_{\vec{t}_0}(\vec{t})}$,
where $\vec{\theta}_0(\vec{t}_0) $ is the parameter to describe the diagonal elements of  $\rho_{\vec{t}_0}$.
Since the parametrization $ \rho_{\vec{t}}$ is $C^2$-continuous,
the function $f$ is also $C^2$-continuous.
Proposition \ref{lem-QLAN} guarantees that 
\begin{align}
\rho_{\vec{t}_0+ \vec{t}/\sqrt{n}}^{\otimes n}
&=
\left(\rho^{{\rm KG}|\vec{t}_0}_{\vec{\theta}_0(\vec{t}_0)+f_{\vec{t}_0}(\vec{t}/\sqrt{n})}\right)^{\otimes n} \nonumber\\
&\stackrel{\vec{\theta}_0}{\cong}
G[\sqrt{n} f_{\vec{t}_0}(\vec{t}/\sqrt{n}),\Gamma_{\vec{\theta}_0(\vec{t}_0)}]
~(\vec{t} \in \set{\Theta}_{n,x,c({\vec{t}_0})}\cap \mathbb{R}^k
,n^{-\kappa} )
 \Label{ER1T}
\end{align}
with suitable choice of the constant $c({\vec{t}_0})$.
Denoting by $f_{\vec{t}_0}' $  the Jacobian matrix of $f_{\vec{t}_0}$,
since $f$ is  $C^2$-continuous and $f_{\vec{t}_0}(0)=0$, 
the norm $\| \sqrt{n} f_{\vec{t}_0}(\vec{t}/\sqrt{n})- f_{\vec{t}_0}'(0)\vec{t} \|_1$ is evaluated as
$O(\frac{\|\vec{t}\|^2}{\sqrt{n}}) $.
Hence, 
the trace norm $
\|G[\sqrt{n} f_{\vec{t}_0}(\vec{t}/\sqrt{n}),\Gamma_{\vec{\theta}_0(\vec{t}_0)}]
-
G[f_{\vec{t}_0}'(0)\vec{t},\Gamma_{\vec{\theta}_0(\vec{t}_0)}]\|_1
$ is also $O(\frac{\|\vec{t}\|^2}{\sqrt{n}}) $, which is at most $O(n^{-5/18}) $
because $\vec{t} \in \set{\Theta}_{n,x,c({\vec{t}_0})}$.
Since $O(n^{-5/18})$ is smaller than $n^{-\kappa}$,
the combination of this evaluation and \eqref{ER1T} yields 
\begin{align}
\rho_{\vec{t}_0+ \vec{t}/\sqrt{n}}^{\otimes n}
\stackrel{\vec{\theta}_0}{\cong}
G[f_{\vec{t}_0}'(0)\vec{t},\Gamma_{\vec{\theta}_0(\vec{t}_0)}]
~(\vec{t} \in \set{\Theta}_{n,x,c({\vec{t}_0})}\cap \mathbb{R}^k
,n^{-\kappa} ).
 \Label{ER6}
\end{align}
The combination of Lemma \ref{LNS} and \eqref{ER6} implies 
\eqref{ER7}.
\end{proof}

\section{
The $\epsilon$-difference RLD Fisher information matrix}
\label{sec-RLD}
In Section \ref{S-one} we  evaluated the limiting distribution in the one-parameter case, using the fidelity as  a discretized version of the SLD Fisher information. 
 In order to tackle the  multiparameter case, we need to develop a similar discretization for the RLD Fisher information matrix, which is the relevant quantity for the multiparameter setting (cf. Section \ref{s3}).  
In this section we define  a discretized version of the RLD Fisher information matrix, extending to the multiparameter case the single-parameter definition introduced by Tsuda  and Matsumoto  \cite{tsuda2005quantum},  who in turn extended the corresponding  classical notion   \cite{chapman1951minimum,hammersley1950estimating}.

\subsection{Definition} 

 Let  $\map M  =  \{  \rho_{\vec{t}}\}_{\vec t \in  \set \Theta}$ be a $k$-parameter model, with the property that  $\rho_{\vec{t}_0}$ is invertible.  If the parametrization  $\rho_{\vec{t}}$ is differentiable,  
 the RLD quantum Fisher information matrix $\tilde{J}_{\vec{t}}$ can be rewritten as 
the following  $k\times k$ matrix
 \begin{align}
\Label{RLDNuova} 
\left(\tilde{J}_{\vec{t}_0}\right)_{ij}
=\Tr  \left[  
\tilde{L}_j^\dagger \rho_{\vec{t}_0} \tilde{L}_i \right]  =  \Tr  \left[  \left.
  \frac{\partial \rho_{\vec{t}}}{\partial t_j}\right|_{\vec{t}=\vec{t}_0}    \rho_{\vec{t}_0}^{-1}   \left. \frac{\partial \rho_{\vec{t}}}{\partial t_i}\right|_{\vec{t}=\vec{t}_0}  \right] \, .
\end{align}

The $\epsilon$-difference   RLD quantum  Fisher information matrix $\tilde{J}_{\vec{t}_0,\epsilon}$ is defined by  replacing the partial derivatives  with finite increments: 
\begin{align}
\left(\tilde{J}_{\vec{t}_0,\epsilon}\right)_{i,j}
:=&   
\Tr   \left[   \left(  \frac {  \rho_{\vec{t}_0+ \epsilon \vec{e}_i}- \rho_{\vec{t}_0} }\epsilon \right)  \, 
\rho_{\vec{t}_0}^{-1}  \,
    \left(  \frac { \rho_{\vec{t}_0+ \epsilon \vec{e}_j}- \rho_{\vec{t}_0}  }{\epsilon } \right) \right] \nonumber\\
= &
\frac {
  \Tr \Big [ \rho_{\vec{t}_0+ \epsilon \vec{e}_i} \, \rho_{\vec{t}_0}^{-1} \, \rho_{\vec{t}_0+ \epsilon \vec{e}_j}  \Big] -1}{\epsilon^2} \, ,\label{KHT}
\end{align}
 where $\vec{e}_j$ is the unit vector  with  $1$ in the $j$-th entry and zero in the other entries.    
Notice that one has  
\begin{align}\left(\tilde{J}_{\vec{t}_0,\epsilon}\right)_{i,i}   =  \frac{\exp  \big[ D_2(\rho_{\vec{t}_0+ \epsilon \vec{e}_i}||\rho_{\vec{t}_0})  \Big] -1}{\epsilon^2} \, , 
\end{align} 
where $D_2(\rho||\sigma):=\log  \Tr \big[\rho^2\sigma^{-1}\big
]$ is the (Petz's) $\alpha$-Renyi entropy for $\alpha=2$. 
 
When the parametrization $\rho_{\vec{t}} $ is differentiable, 
 one has 
 \begin{align}
\lim_{\epsilon \to 0}\tilde{J}_{\vec{t}_0,\epsilon}= \tilde{J}_{\vec{t}_0} \, ,\Label{3-3-0}
\end{align}
where $\tilde{J}_{\vec{t}_0}$ is the RLD quantum Fisher information matrix \eqref{RLDNuova}.

When the parametrization  is not differentiable, we define 
 the RLD Fisher information matrix $\tilde{J}_{\vec{t}_0}$
to be the limit \eqref{3-3-0}, provided that the limit  exists.  
All throughout  this section, we impose no condition on the  parametrization $\rho_{\vec{t}} $, except for the requirement that $\rho_{\vec{t}_0}$ be invertible.


\subsection{ The $\epsilon$-difference RLD Cram\'er-Rao inequality} 
 A discrete version of  the RLD quantum  Cram\'{e}r-Rao inequality can be derived under the assumption of {\em $\epsilon$-locally unbiasedness}, defined as follows: 
\begin{defi} A POVM $M$ with outcomes in $\R^k$ is      $\epsilon$-locally unbiased  at $\vec{t}_0$ if the expectation value 
  $E_{\vec{t}_0}(M)$ 
  satisfies the conditions 
\begin{align*}
E_{\vec{t}_0}(M) 
= \vec t_{0} ,\qquad  {\rm and}  \qquad 
E_{\vec{t}_0+ \epsilon \vec e_j}(M)
= \vec t_{0} +\epsilon  \vec e_j  \qquad \forall j  \in  \{1,\dots, k\} \, .
\end{align*} 
\end{defi}

Under the $\epsilon$-locally unbiasedness condition, 
Tsuda et al. \cite{tsuda2005quantum} derived a lower bound on the MSE for the one-parameter case. In the following theorem, we extend the bound to the multiparameter case. 

\begin{theo}[$\epsilon$-difference RLD Cram\'er-Rao inequality]\Label{LO5}
The MSE matrix for an $\epsilon$-locally unbiased POVM $M$ at $\vec{t}_0$ satisfies the bound
 \begin{align}
 V_{\vec{t}_0}(M)\ge  (\tilde{J}_{\vec{t}_0,\epsilon})^{-1}. \Label{3-3-1}
 \end{align}
\end{theo}
\begin{proof}
For simplicity, we assume that $\vec{t}_0=\vec{0}$.
For two vectors $\vec{a} \in \mathbb{C}^k$ and $\vec{b}\in \mathbb{C}^{k}$,
 we define the two observables
$X:= \int \sum_i a_i x_i  \,M(d \vec{x}) $ 
and $Y:= \sum_j  b_j  \frac {
\rho_{\vec{t}_0+ \epsilon \vec e_j}-\rho_{\vec{t}_0}}\epsilon  \rho_{\vec{t}_0}^{-1}$.
Then, the Cauchy-Schwartz inequality implies
\begin{align}
\nonumber \Tr \Big[   X^\dag X   \,  \rho_{\vec t_0}\Big]  \,  \Tr \Big[   Y^\dag Y   \,  \rho_{\vec t_0}\Big]       & \ge   \left|\Tr\Big[ X^\dagger Y \rho_{\vec{t}_0} \Big] \right|^2  \\
\nonumber & =  \left |  \sum_{i,j}    \overline a_i  b_j   ~\int   x_i  \Tr  \left[  M(d \vec x) \frac {
\rho_{\vec{t}_0+ \epsilon \vec e_j}-\rho_{\vec{t}_0}}\epsilon\right]       \right|^2\\
&   =  |\< \vec a| \vec b\>|^2 \, ,
  \end{align}
  the second equality  following from 
 $\epsilon$-locally unbiasedness  at $\vec{t}_0$.
 Note that one has  $\Tr [Y^\dag Y  \rho_{\vec t_0}]  =   \<  \vec b |  \,   \tilde J_{\vec t_0, \epsilon} \, |\vec  b \>$ and
 \begin{align}
\nonumber \langle \vec{a} |V_{\vec{t}_0}(M) |\vec{a} \rangle-
\Tr \Big[ X^\dagger X \rho_{\vec{t}_0}\Big]
 &=
\int   \Tr \left[\Big(\sum_i a_i x_i- X\Big )^\dagger \Big(\sum_j a_j x_j- X\Big) M(d \vec{x}) \rho_{\vec{t}_0}\right] \\
&\ge 0 \,.
\end{align}
Choosing $\vec{b}:= (\tilde{J}_{\vec{t}_0,\epsilon})^{-1} \vec{a}$, we have
\begin{align}
\nonumber  \langle \vec{a} |V_{\vec{t}_0}(M) |\vec{a} \rangle ~
\langle \vec{a} | \tilde{J}_{\vec{t}_0,\epsilon}^{-1} |\vec{a} \rangle  &\ge  \Tr \Big[   X^\dag X   \,  \rho_{\vec t_0}\Big]  \,  \Tr \Big[   Y^\dag Y   \,  \rho_{\vec t_0}\Big]     \\
\nonumber    & \ge   \left|\Tr\Big[ X^\dagger Y \rho_{\vec{t}_0} \Big] \right|^2 \\
   & =\left|\langle \vec{a} |  (\tilde{J}_{\vec{t}_0,\epsilon})^{-1}  |  \vec{a}\rangle \right|^2 \, , 
\end{align}
which implies $\<\vec{a} |V_{\vec{t}_0}(M) |\vec{a} \rangle  \ge  \langle \vec{a} |  (\tilde{J}_{\vec{t}_0,\epsilon})^{-1}  |  \vec{a}\rangle$.   Since $\vec a$ is arbitrary, the last inequality implies \eqref{3-3-1}.
\end{proof}

We will call  \eqref{3-3-1} the {\em $\epsilon$-difference RLD Cram\'{e}r-Rao inequality}.  

The $\epsilon$-difference RLD Cram\'er-Rao inequality can be used to derive an information processing inequality, 
which states that the $\epsilon$-difference RLD Fisher information matrix is non-increasing  under the application of measurements. 
For a family of probability distributions $\{P_{\vec t}\}_{\vec t \in \set \Theta}$, 
we assume that 
$P_{\vec{t}+\epsilon \vec{e}_j }$ is absolutely continuous with respect to $P_{\vec{t}}$ for  every $j  $. 
Then, 
the $\epsilon$-difference RLD Fisher information is defined as 
\begin{align}
\left(  J_{\vec{t},\epsilon}\right)_{ij} : =  \int   \left(   \frac{p_{\vec{t}+\epsilon \vec e_j }(x)-1}{\epsilon}\right)  \, \left(
\frac{p_{\vec{t}+\epsilon \vec e_j }(x)-1}{\epsilon}\right)  ~   P_{\vec{t}}(dx)\label{GCR}
\end{align}
 where  $p_{\vec{t}+\epsilon \vec{e}_j }$  and   $p_{\vec{t}+\epsilon \vec{e}_i }$  are the Radon-Nikod\'{y}m derivatives of 
$P_{\vec{t}+\epsilon \vec{e}_j }$ and $P_{\vec{t}+\epsilon \vec{e}_i }$ 
with respect to $P_{\vec{t}}$, respectively.
We note that the papers \cite{chapman1951minimum,hammersley1950estimating}
defined its one-parameter version when the distributions are absolutely continuous with respect to the Lebesgue measure.
Hence, 
when an estimator $\hat{\vec{t}}$ for the distribution family $\{P_{\vec t}\}_{\vec t \in \set \Theta}$
is $\epsilon$-locally unbiased at $\vec{t}_0$,
in the same way as \eqref{3-3-1}, 
we can show the $\epsilon$-difference Cram\'er-Rao inequality;
 \begin{align}
 V_{\vec{t}_0}[\hat{t}]\ge  ({J}_{\vec{t}_0,\epsilon})^{-1}. \Label{3-3-1F}
 \end{align}

For  a family of quantum states $\{\rho_{\vec t}\}_{\vec t \in \set \Theta}$ and a POVM $M$, we denote  by $J_{\vec{t},\epsilon}^M$ the $\epsilon$-difference Fisher information matrix of  the probability distribution family $\{ P^M_{\vec t}\}_{\vec{t} \in \set \Theta}$ defined by $P^M_{\vec t}  :  =\Tr \big [ M\rho_{\vec{t}}\big ]$. With this notation, we have the following lemma:
 
\begin{lem}\Label{LO6}
For  every family of quantum states $\{\rho_{\vec t}\}_{\vec t \in \set \Theta}$ and every  POVM $M$, one has the information processing inequality
 \begin{align}
\tilde{J}_{\vec{t}_0,\epsilon} \ge J_{\vec{t}_0,\epsilon}^M \, ,\Label{3-3-3}
 \end{align}
 where $\tilde{J}_{\vec{t}_0,\epsilon}$ is the $\epsilon$-difference RLD Fisher information of the model  $\{\rho_{\vec t}\}_{\vec t \in \set \Theta}$. 
\end{lem}

\begin{proof}
 Consider the estimation of $\vec{t}$ from the probability distribution family $\{ P^M_{\vec t}\}_{\vec{t} \in \set \Theta}$.  Following the same arguments used for the achievability of the Cram\'er-Rao bound with  locally unbiased estimators (see, for instance, Chapter 2 of Ref. \cite{hayashi2017quantum}),  it is possible to show that   there exists an  $\epsilon$-locally unbiased estimator $\hat {\vec t}$ at $\vec{t}_0$  
such that 
\begin{align}
V_{\vec{t}_0}\left(\hat {\vec t}  \right) =  (J_{\vec{t}_0,\epsilon}^M)^{-1} \, .\Label{3-3-2}
\end{align}
Combining the POVM $M$ with the $\epsilon$-locally unbiased estimator $\hat {\vec t}$ we obtain a new POVM $M'$, which is $\epsilon$-locally unbiased. Applying Theorem \ref{LO5}  to the POVM $M'$ we obtain   
 \begin{align}
( \tilde{J}_{\vec{t}_0,\epsilon} )^{-1}\le  V_{\vec{t}_0}\left(   M' \right)   =  V_{\vec{t}_0}\left(\hat {\vec t}  \right)   =   (J_{\vec{t}_0,\epsilon}^M)^{-1} \, ,
\end{align}    
which implies  \eqref{3-3-3}.    
\qed
\end{proof}

We  stress  that \eqref{3-3-3}   is a matrix inequality for {\em Hermitian} matrices: in general,  $\tilde{J}_{\vec{t}_0,\epsilon}$ has complex entries.
  Also note that  any classical  process can be regarded as a POVM.  
Hence, in the same way as \eqref{3-3-3}, 
using the $\epsilon$-difference Cram\'er-Rao inequality \eqref{3-3-1F},
we can show the inequality 
\begin{align}
J_\epsilon \ge J_\epsilon' \label{MMM}
\end{align}
for an classical process ${\cal E}$
when 
$J_\epsilon$ is the $\epsilon$-difference Fisher information matrix on the distribution family
$\{P_{\vec{t}}\}_{\vec{t}\in \set\Theta}$
and 
$J_\epsilon'$ is the $\epsilon$-difference Fisher information matrix on the distribution family
$\{{\cal E}(P_{\vec{t}})\}_{\vec{t}\in \set\Theta}$.
 
\if

\subsection{Proof of Theorem \ref{ThY}}\Label{P-ThY}
In Subsection \ref{S2-2-2}, we have shown Theorem \ref{ThY} by assuming .
If we employ the concept of $\epsilon$-difference RLD Fisher information matrix,
we can show Theorem \ref{ThY} as follows.
Applying Theorem \ref{LO5} to $\{\wp_{t_0,t|\seq{m}}\}$ and using the observation in Remark \ref{RE1}, 
we have
\begin{align}
V[\wp_{t_0,t| \seq{m}}]=V[\wp_{t_0,0| \seq{m}}]
\ge
\epsilon^2/(e^{D_2(\wp_{t_0,t+\epsilon|\seq{m}}||\wp_{t_0,t|\seq{m}})}-1).
\end{align}
Further, since the relative R\'{e}nyi entropy is monotonely increasing with respect to the order parameter,
we have
\fi

\subsection{Extended models}
The  lemmas in the previous subsection can be generalized to the case where  an extended model ${\cal M}':=
\{\rho_{\vec{t}'}\}_{\vec{t}'=(\vec{t},\vec{p})}$ contains the original model ${\cal M} $
as $\rho_{\vec{t}}=\rho_{(\vec{t},\vec{0})}$.
Choosing $\vec{t}_0'=(\vec{t}_0,\vec{0})$, we denote 
the $\epsilon$-difference RLD Fisher information matrix at $\vec{t}_0'$ for the family ${\cal M}'$
by $\tilde{J}_{\vec{t}_0',\epsilon}$. 

\begin{lem}\Label{LO7-2}
For an $\epsilon$-locally unbiased estimator $M$ at $\vec{t}_0'$,
there exists a $k \times k'$ matrix $P$
such that
$P_{ij} =\delta_{ij}$ for $ i,j\le k$ and
\begin{align*}
V_{\vec{t}_0}(M) &\ge P  \tilde{J}_{\vec{t}_0',\epsilon}^{-1}P^T.
\end{align*}
\end{lem}
\begin{proofof}{Lemma \ref{LO7-2}}
For an $\epsilon$-locally unbiased estimator $M$ at $\vec{t}_0$,
there exists a $k \times k'$ matrix $P$
such that
\begin{align}
P_{ij}& =\delta_{ij} \hbox{ for }i,j\le k \Label{FDT}\\
\epsilon P_{ij}&= \int x_i \Tr M(d \vec{x})
( \rho_{\vec{t}_0'+ \epsilon \vec e_j}-\rho_{\vec{t}_0'})
\hbox{ for }i\le k, k+1 \le j \le k'. 
\end{align}
Now, we introduce a new parametrization 
$\tilde{\rho}_{\eta}:= 
\rho_{\vec{t}_0'+ \sum_{i,j}\eta_i P^{-1}_{j,i}\vec e_j}$.
Since 
$\frac{\partial \theta_j}{\partial \eta_i}
= P^{-1}_{j,i}$, the $\epsilon$-difference RLD quantum Fisher information
under the parameter $\eta$ is $ (P^{-1})^{T} \tilde{J}_{\vec{t}_0',\epsilon} P^{-1}$.
Applying Theorem \ref{LO5} to the parameter $\eta$, we obtain
\begin{align}
V_{\vec{t}_0}(M) \ge P  \tilde{J}_{\vec{t}_0',\epsilon}^{-1}P^T. \Label{3-3-E}
\end{align}
Combining \eqref{FDT} and \eqref{3-3-E}, we obtain the desired statement.
\end{proofof}

In the same way as Lemma \ref{LO6}, 
Lemma \ref{LO7-2} yields the following lemma.

\begin{lem}\Label{LO8}
For any POVM $M$, 
there exists a $k \times k'$ matrix $P$
such that
$P_{ij} =\delta_{ij}$ { for }$ i,j\le k$ and
\begin{align}
(J_{\vec{t}_0,\epsilon}^M)^{-1} 
&\ge P  \tilde{J}_{\vec{t}_0',\epsilon}^{-1}P^T
. \Label{3-3-F}
\end{align}
\end{lem}

\subsection{Asymptotic case}\Label{SRT}
We denote by $\tilde{J}_{\vec{t}_0,\epsilon}^{n}$ the 
$\epsilon$-difference RLD Fisher information matrix
 of the $n$-copy states $\{\rho_{\vec{t}}^{\otimes n}\}_{\vec{t}  \in \set \Theta}$.

In the following we provide the analogue of Lemma \ref{MGR} for the RLD Fisher information matrix.
\begin{lem}\Label{LLO}
When  the parametrization is $C^1$ continuous, the following relations hold 
 \begin{align}
\lim_{n \to\infty}\frac{1}{n}\tilde{J}_{\vec{t}_0,   \frac {\epsilon }{\sqrt{n}}}^{n} 
&= \tilde{J}^{[\epsilon]}_{\vec{t}_0} \Label{KBI} \\
\lim_{\epsilon \to 0}
\tilde{J}^{[\epsilon]}_{\vec{t}_0} 
&= \tilde{J}_{\vec{t}_0}\Label{2KBI} ,
 \end{align}
 where  $J^{[\epsilon]}_{\vec{t}_0}$ is the matrix defined by
$\left( \tilde{J}^{[\epsilon ]}_{\vec{t}_0} \right)_{i,j}:= \frac{1}{\epsilon^2} \left[  e^{\epsilon^2   \left(  \tilde{J}_{\vec{t}_0}\right)_{i,j}}-1\right]$.
\end{lem}
\begin{proofof}{Lemma \ref{LLO}}
Eq. \eqref{2KBI} holds trivially.    Using \eqref{3-3-0}, we have 
 \begin{align*}
\lim_{n\to \infty}\frac{1}{n} \left( \tilde{J}_{\vec{t}_0, \frac {\epsilon}{\sqrt{n}}}^n  \right)_{i,j}
&
=\lim_{n\to \infty}
\frac{1}{\epsilon^2} \left(
\Tr \left[ 
\rho_{\vec{t}_0+ \frac{\epsilon}{\sqrt{n}} \vec e_j}^{\otimes n}
(\rho_{\vec{t}_0}^{\otimes n})^{-1}
\rho_{\vec{t}_0+ \frac{\epsilon}{\sqrt{n}} \vec e_i}^{\otimes n}\right]
-1\right) \\
&= \lim_{n\to \infty}
\frac{1}{\epsilon^2} \left[\left
(1+ \frac{\epsilon^2}{n}  \left(  \tilde{J}_{\vec{t}_0}\right)_{i,j} +O\left(\frac{1}{n^2}\right)\right)^n
-1\right] \\
&=\lim_{n\to \infty}
\frac{1}{\epsilon^2} \left[
\left(1+ \frac{\epsilon^2}{n} \left(  \tilde{J}_{\vec{t}_0}\right)_{i,j} +O\left(\frac{1}{n^2}\right)\right)^n
-1\right]\\ 
& = 
\frac{1}{\epsilon^2} \left[
e^{\epsilon^2   \left(\tilde{J}_{\vec{t}_0}\right)_{i,j}}-1 \right  ] \, ,
\end{align*}
which implies \eqref{KBI}.
\end{proofof}

\section{Precision bounds for multiparameter estimation}\Label{sec:multi}

\subsection{Covariance conditions}
First, we introduce the condition for our estimators.
The correspondence between qudit states and Gaussian states also extends to the estimator level. 
We consider a generic state family ${\cal M}=\{\rho_{\vec{t}}\}_{\vec{t} \in \set{\Theta}}$, with the parameter space $\set{\Theta}$ being an open subset of $\R^k$. 
Similar to the single-parameter case, 
given a point $\vec{t}_{0} \in \set{\Theta}$, 
we consider a local model $\rho^n_{\vec{t}_0,\vec{t}}:=\rho^{\otimes n}_{\vec{t}_0+\vec{t}/\sqrt{n}}$.
Throughout this section, we assume that $\rho_{\vec{t}_0}$ is invertible.
For a sequence of POVM $\seq{m}:=\{M_n\}$, we introduce the condition of local asymptotic covariance as follows:

\begin{condition}[Local asymptotic covariance] \Label{cond}
We say that a sequence of measurements $\seq{m}:=\{M_n\}$ satisfies local asymptotic covariance at $\vec{t}_{0} \in \set{\Theta}$ under the state family ${\cal M}$, 
if the probability distribution 
\begin{align}\Label{finitedistribution}
\wp^{n}_{\vec{t}_{0},\vec{t}|M_n}(\set{B})
:=\Tr\rho_{\vec{t}_0+\frac{\vec{t}}{\sqrt{n}}} ^{\otimes n}M_n
\left( \frac{\set{B}}{\sqrt{n}}+\vec{t}_0\right) 
\end{align}
converges to a limiting distribution 
\begin{align}\Label{limiting}
\wp_{\vec{t}_{0},\vec{t} |\seq{m}}(\set{B}):=
\lim_{n\to\infty} \wp^n_{\vec{t}_{0},\vec{t}|M_n}(\set{B}),
\end{align}
the relation
\begin{align}
\wp_{\vec{t}_{0},\vec{t}| \seq{m}}(\set{B}+\vec{t})
=
\wp_{\vec{t}_{0},\vec{0}| \seq{m}}(\set{B})
\Label{LLe}
\end{align}
holds for any $\vec{t}\in\R^k$\ \footnote{The range of $\vec{t}$ is determined via the constraint $\vec{t}_0+\vec{t}/\sqrt{n}\in\set{\Theta}$.
Just as in the one-parameter case, $\vec{t}$ can take any value in $\R^k$ when $n$ is large enough. The range of the local parameter is then $\vec{t}\in\R^k$.}.
When we need to express the outcome of 
$\wp^n_{\vec{t}_{0},\vec{t}|M_n} $ or
$\wp_{\vec{t}_{0},\vec{t}| \seq{m}}$, we denote it by $\hat{\vec{t}}$. 

Further, we say that a sequence of measurements $\seq{m}:=\{M_n\}$ satisfies local asymptotic covariance under the state family ${\cal M}$
when it satisfies local asymptotic covariance 
at any element $\vec{t}_{0} \in \Theta$ under the state family ${\cal M}$.
\end{condition}


Under these preparations, we obtain the following theorem by using Theorem \ref{Th3}.
\begin{theo}\Label{CTh3}
Let $\{\rho^{\otimes n}_{\vec{t}}\}_{\vec{t}\in\set{\Theta}}$ be a $k$-parameter D-invariant qudit model
with $C^2$ continuous parametrization.
Assume that 
$\tilde{J}^{-1}_{\vec{t}_0}$ exists, 
$\rho_{\vec{t}_0}$ is a non-degenerate state,
and a sequence of measurements $\seq{m}:=\{M_n\}$ satisfies local asymptotic covariance at $\vec{t}_{0} \in \set{\Theta}$.
Then there exists a covariant POVM $\widetilde{M}^G$
such that
\begin{align}\Label{Gaussian-limiting}
\Tr \widetilde{M}^G(\set{B})G[\vec{t},\tilde{J}^{-1}_{\vec{t}_0}]=
\wp_{\vec{t}_0,\vec{t}|\seq{m}}(\set{B})
\end{align}
for any vector $\vec{t}$ and any measurable subset $\set{B}$. Here $\tilde{J}_{\vec{t}_0}$ is the RLD Fisher information of the qudit model at $\vec{t}_0$. 
\end{theo}

To show Theorem \ref{CTh3}, we will use  the following lemma.

\begin{lem}\Label{L54}
For a function $f$, an operator $F$, and a c-q Gaussian state in the canonical form $G[\vec{\alpha},\Gamma]$, the relation
\begin{align}\Label{Gaussian-limiting2}
\Tr F G[\vec{\alpha},\Gamma]=
f(\vec{\alpha})
\end{align}
holds for any vector $\vec{\alpha}$
if and only if 
\begin{align}\Label{Gaussian-limiting3}
F=\int d\vec{y}\,\map{F}^{-1}_{\vec{\xi}\to\vec{y}}\left(\sqrt{\pi^k}e^{\frac14\sum_j\frac{\xi_j^2}{1-\gamma_j}}\map{F}_{\vec{\alpha}\to\vec{\xi}}\left( 
f (\vec{\alpha})
\right)\right)|\vec{y}\>\<\vec{y}|.
\end{align} 
 Here $\vec{\xi}$ and $\vec{y}$ are $k$-dimensional vectors, $|\vec{y}\>$ is a (multimode) coherent state, $\gamma_j$ are thermal parameters of the Gaussian, and $\map{F}^{-1}_{\vec{\xi}\to\vec{y}}(g)$ denotes the inverse of the Fourier transform 
$\map{F}_{\vec{\xi}\to\vec{y}}(g):=\int d\vec{\xi}\ e^{i\vec{\xi}\cdot\vec{y}} g$.
Therefore,
for a given function $f(\vec{\alpha})$, there uniquely exists an operator $F$
to satisfy \eqref{Gaussian-limiting2}.
\end{lem}
The proof can be found in Appendix \ref{app-lemma12}.
Now, we are ready to prove Theorem \ref{CTh3}.
\begin{proofof}{Theorem \ref{CTh3}}
We consider without loss of generality $G[\vec{t},\tilde{J}_{\vec{t}_0}^{-1}]$ to be in the canonical form, noticing that any Gaussian state is unitarily equivalent to a Gaussian state in the canonical form as shown by Lemma \ref{3LL}.
For any measurable set $\set{B}$, we define 
the operator $\widetilde{M}^G(\set{B})$ as
\begin{align}\Label{app-MG8}
\widetilde{M}^G(\set{B}):=\int d\vec{y}\,\map{F}^{-1}_{\vec{\xi}\to\vec{y}}\left(\sqrt{\pi^k}e^{\frac14\sum_j\frac{\xi_j^2}{1-\gamma_j}}\map{F}_{\vec{t}\to\vec{\xi}}\left( 
\wp_{\vec{t}_0,\vec{t}|\seq{m}}(\set{B})
\right)\right)|\vec{y}\>\<\vec{y}|.
\end{align}
From the above definition, it can be verified that $\widetilde{M}^G(\set{B})$ satisfies the definition of a POVM: first, it is immediate to see that the term $\map{F}^{-1}_{\vec{\xi}\to\vec{y}}\left(\sqrt{\pi^k}e^{\frac14\sum_j\frac{\xi_j^2}{1-\gamma_j}}\map{F}_{\vec{t}\to\vec{\xi}}\left( 
\wp_{\vec{t}_0,\vec{t}|\seq{m}}(\set{B})
\right)\right)$ equals the convolution of $\wp_{\vec{t}_0,\vec{t}|\seq{m}}(\set{B})$ and a Gaussian function by employing the convolution theorem. Since both functions are non-negative, the outcome of convolution is still non-negative, which implies that $\widetilde{M}^G(\set{B})\ge 0$. Second, by linearity we have $\widetilde{M}^G(\set{B}_1\cup\set{B}_2)=\widetilde{M}^G(\set{B}_1)+\widetilde{M}^G(\set{B}_2)$ for any disjoint sets $\set{B}_1$ and $\set{B}_2$. Finally, the equality $\widetilde{M}^G(\set{B})$ can be shown by combining the linearity with the fact that $\wp_{\vec{t}_0,\vec{t}|\seq{m}}$ is a probability distribution function.

Moreover, Lemma \ref{L54} guarantees that
\begin{align}\Label{H32}
\Tr \widetilde{M}^G(\set{B})
 G[\vec{t},\tilde{J}_{\vec{t}_0}^{-1}]=
\wp_{\vec{t}_0,\vec{t}|\seq{m}}(\set{B}).
\end{align}
What remains to be shown is that the POVM $\{\widetilde{M}^G(\set{B})\}$ satisfies the covariance condition.
Eq. \eqref{H32} guarantees that
\begin{align*}
\Tr T_{\vec{t}'} \widetilde{M}^G(\set{B}) T_{\vec{t}'}^\dagger
 G[\vec{t},\tilde{J}_{\vec{t}_0}^{-1}]=
\Tr \widetilde{M}^G(\set{B})
 G[\vec{t}-\vec{t}',\tilde{J}_{\vec{t}_0}^{-1}]=
\wp_{\vec{t}_0,\vec{t}-\vec{t}'|\seq{m}}(\set{B}),
\end{align*}
and
\begin{align*}
\Tr \widetilde{M}^G(\set{B}+\vec{t}') 
 G[\vec{t},\tilde{J}_{\vec{t}_0}^{-1}]=
\wp_{\vec{t}_0,\vec{t}|\seq{m}}(\set{B}+\vec{t}')
=
\wp_{\vec{t}_0,\vec{t}-\vec{t}'|\seq{m}}(\set{B}).
\end{align*}
The uniqueness of the operator to satisfy the condition \eqref{Gaussian-limiting2}
implies the covariance condition
\begin{align*}
\widetilde{M}^G(\set{B}+\vec{t}')
= T_{\vec{t}'} \widetilde{M}^G(\set{B}) T_{\vec{t}'}^\dagger.
\end{align*}
\end{proofof}

\subsection{MSE bound for the D-invariant case}\Label{subsec:D-inv}
Next, we derive the lower bound of MSE of the limiting distribution
for any D-invariant model.
As an extension of the mean square error, 
we introduce the {\em mean square error  matrix (MSE matrix)}, defined as 
\begin{align}\Label{cov-def}
V_{i,j}[\wp]:= \int\,x_i x_j \wp(dx)
\end{align}
for a generic probability distribution $\wp$.  
Since the set of symmetric matrices is not totally ordered,
we will consider the minimization of the expectation value   $\tr W V[\wp_{\vec{t}_{0},\vec{t}| \seq{m}}]$ 
for a certain weight matrix $W\ge 0$. For short, we will refer to the quantity  $\tr W V[\wp_{\vec{t}_{0},\vec{t}| \seq{m}}]$ as the {\em weighted MSE}.

Under local asymptotic   covariance, one can   derive lower bounds on the 
covariance matrix of the limiting distribution and construct optimal measurements to achieve them.
In general, the attainability of the conventional quantum Cram\'er-Rao bounds is a challenging issue.
  For instance, a well-known bound is the symmetric logarithmic derivative (SLD) Fisher information bound
\begin{align}\Label{SLDB}
\tr W V[\wp_{\vec{t}_{0},\vec{t}| \seq{m}}]\ge 
\tr W J_{\vec{t}_0}^{-1},
\end{align}
where $J_{\vec{t}_0}$ is the SLD Fisher information.
The SLD bound is attainable in the single-parameter case, i.e. when $k=1$, yet it is in general not attainable for multiparameter estimation (see, for instance, later in Subsection \ref{subsec-ob} for a concrete example).

 In the following, we derive an {\em attainable} lower bound on the  weighted MSE. To this purpose,  we define the set $\set{LAC}(\vec{t}_{0})$ of 
local asymptotic covariant sequences of measurements at the point $\vec{t}_0 \in \set{\Theta}$. For a model ${\cal M}$, we focus on the minimum value 
\begin{align}\Label{global-boundB}
\map{C}_{\spc{S}}(W,\vec{t}_{0}):=\min_{\seq{m}' \in \set{LAC}(\vec{t}_{0}) }
 \tr W V [\wp_{\vec{t}_{0},\vec{t} |\seq{m}'}].
\end{align}
When $k\ge 2$, a better choice is the RLD quantum Fisher information bound. 
The main result of this section is an attainable bound on the weighted MSE, relying on the RLD quantum Fisher information.
\begin{theo}[Weighted MSE bound for D-invariant models]\Label{thm-MSE}
Assume that $\tilde{J}_{\vec{t}_0}^{-1} $ exists.
Consider any sequence of locally asymptotically covariant measurements $\seq{m}:=\{M_n\}$.
The limiting distribution is evaluated as
\begin{align}\Label{RLDB0}
V[\wp_{\vec{t}_{0},\vec{t}| \seq{m}}] 
\ge (\tilde{J}_{\vec{t}_0})^{-1},
\end{align}
where $\tilde{J}_{\vec{t}_0}$ is the RLD quantum Fisher information.
When the model is $C^1$ continuous and D-invariant,
we have the bound for the weighted MSE with weight matrix $W\ge 0$
of the limiting distribution as
\begin{align}\Label{RLDB}
\tr W V[\wp_{\vec{t}_{0},\vec{t}| \seq{m}}] 
\ge \tr WJ_{\vec{t}_0}^{-1}+\frac12\tr\left|\sqrt{W}J_{\vec{t}_0}^{-1}D_{\vec{t}_0}J_{\vec{t}_0}^{-1}\sqrt{W}\right|,
\end{align}
where $J_{\vec{t}_0}$ is the SLD quantum Fisher information (\ref{SLD-QFI}) and $D_{\vec{t}_0}$ is the D-matrix (\ref{D-m}).
When $\spc{S}$ is a D-invariant qudit model
and the state $\rho_{\vec{t}_0}$ is not degenerate, we have
\begin{align}\Label{RLDB2}
\map{C}_{\spc{S}}(W,\vec{t}_{0})
= \tr WJ_{\vec{t}_0}^{-1}+\frac12\tr\left|\sqrt{W}J_{\vec{t}_0}^{-1}D_{\vec{t}_0}J_{\vec{t}_0}^{-1}\sqrt{W}\right|.
\end{align}
Moreover, if $W>0$
and $\wp_{\vec{t}_0,\vec{0}| \seq{m}}$ has a differentiable PDF, 
the equality in (\ref{RLDB}) holds if and only if $\wp_{\vec{t}_0,\vec{t}| \seq{m}}$ is the normal distribution with average zero
and covariance
\begin{align}\label{V-tw}
V_{\vec{t_0}|W}:=J_{\vec{t}_0}^{-1}+\frac12
\sqrt{W}^{-1}
\left|\sqrt{W}J_{\vec{t}_0}^{-1}D_{\vec{t}_0}J_{\vec{t}_0}^{-1}\sqrt{W}\right|
\sqrt{W}^{-1}.
\end{align}
Further, when $\{\rho_{\vec{t}}\}_{\vec{t}\in\set{\Theta}}$ is a qudit-model with $C^2$ continuous parametrization,
the equality in \eqref{RLDB} holds, i.e.,
there exist a sequence of POVMs $M_W^{\vec{t}_0,n}$,
a compact set $K$, and constant $c(\vec{t}_0)$
 such that
\begin{align}
\limsup_{n \to \infty}\sup_{\vec{t}_0 \in K}
\sup_{\vec{t} \in \set{\Theta}_{n,x,c(\vec{t}_0)} } n^{\kappa}\| \wp^n_{\vec{t}_0,\vec{t}| M_W^{\vec{t}_0,n}}
- N[\vec{t}, V_{\vec{t_0}|W}]
\|_1 <\infty,\Label{DJI}
\end{align}
where $\kappa$ is a parameter to satisfy $\kappa\ge 0.027$
\end{theo}

In the following, we prove Theorem \ref{thm-MSE} following three steps. The first step is to derive the bound (\ref{RLDB}). The second step is to show that, to achieve the equality, the limiting distribution needs to be a Gaussian with certain covariance. The last step is to find a measurement attaining the equality. 
In this way,   when the state is not degenerate, we can construct the measurement using Q-LAN\footnote{We note that state estimation involving degenerate states was only solved in a few special cases (see, e.g., \cite{hayashi2016fourier} for the case of maximally mixed qubits).}.


\begin{proofof}{Theorem \ref{thm-MSE}}
\par\noindent{\it Impossibility part \footnote{Note that when the state is not degenerate the proof of the impossibility part of Theorem \ref{thm-MSE} can be simplified using the correspondence between $n$-copy qudit states and Gaussian states as established by the Q-LAN.} (Proofs of \eqref{RLDB0} and \eqref{RLDB}):}

To give a proof, 
we focus on the $\epsilon$-difference RLD Fisher information matrix
$\tilde{J}_{\vec{t}_0,\epsilon}$ at $\vec{t}_0$
for a quantum states family $\{\rho_{\vec{t}}\}_{\vec{t} \in \set{\Theta}}$.
We denote the $\epsilon$-difference Fisher information matrices
for the distribution family
$\{\wp^n_{\vec{t}_{0},\vec{t}|M_n}\}_{\vec{t}}$ 
and
$\{\wp_{\vec{t}_{0},\vec{t} |\seq{m}}\}_{\vec{t}}$
by 
$J_{\vec{t},\epsilon}^n$  and $J_{\vec{t},\epsilon}^{\seq{m}}$, respectively.
Also, we employ the notations given Section \ref{SRT}. 

Applying \eqref{3-3-3} to the POVM $M_n$, we have
 \begin{align}
\frac1n\tilde{J}_{\vec{t}_0, \epsilon/\sqrt{n}}^n \ge 
J^n_{\vec{0},\epsilon},\Label{3-3-4}
 \end{align}
where $\tilde{J}_{\vec{t}_0,\epsilon/\sqrt{n}}^{n}$ denotes the 
$(\epsilon/\sqrt{n})$-difference RLD Fisher information matrix
 of the $n$-copy states $\{\rho_{\vec{t}}^{\otimes n}\}_{\vec{t}  \in \set \Theta}$. 
Combining the above with \eqref{KBI} of Lemma \ref{LLO},
we find that 
$J^n_{\vec{0},\epsilon} <
(\tilde{J}^{[\epsilon]}_{\vec{t}_0}  +I)$
for sufficiently large $n$.
Hence, we can apply Lemma \ref{LL8} (see Appendix \ref{app:asymptoticstuff}) to 
the distribution family
$\{\wp^n_{\vec{t}_{0},\vec{t}|M_n}\}_{\vec{t}}$. 
Then, for any complex vector $\vec{a}$, we have
 \begin{align}
 \liminf_{n \to \infty} 
\langle \vec{a}| J_{\vec{0},\epsilon}^n -J_{\vec{0},\epsilon}^{\seq{m}}|\vec{a}\rangle
\ge 0 .
\Label{3-3-5}
 \end{align}
By taking the limit $n\to \infty$, the combination of \eqref{3-3-4}, \eqref{KBI} of Lemma \ref{LLO}, and \eqref{3-3-5}
implies 
 \begin{align}
\tilde{J}^{[\epsilon]}_{\vec{t}_0}  \ge J_{\vec{0},\epsilon}^{\seq{m}}.
\Label{3-3-6}
 \end{align}
Here, in the same way as the proof of Theorem \ref{ThY},
we can assume that the outcome $\hat{\vec{t}}$ satisfies
the unbiasedness condition. 
Hence, the $\epsilon$-difference Carm\'{e}r-Rao inequality (\ref{3-3-1F}) implies that 
\begin{align}
V[\wp_{\vec{t}_{0},\vec{t}| \seq{m}}] 
\ge \left(J_{\vec{0},\epsilon}^{\seq{m}}\right)^{-1}
\ge \left(
\tilde{J}^{[\epsilon]}_{\vec{t}_0}  
\right)^{-1}.
\Label{3-3-8}
\end{align}
By taking the limit $\epsilon \to 0$, 
\eqref{2KBI} of Lemma \ref{LLO} 
implies 
\begin{align}
V[\wp_{\vec{t}_{0},\vec{t}| \seq{m}}] 
\ge \left(
\tilde{J}_{\vec{t}_0}  
\right)^{-1}.
\Label{3-3-8B}
\end{align}
When the model is $C^1$ continuous and D-invariant, 
adding the conventional discussion for MSE bounds (see, e.g., Chapter 6 of \cite{holevo-book}) to \eqref{3-3-8},
we obtain \eqref{RLDB}.

\noindent{\it Achievability part (Proof of \eqref{RLDB2}):}

Next, we discuss the attainability of the bound
when $W>0$ and $\wp_{\vec{t}_0,\vec{0}| \seq{m}}$ has a differentiable PDF.
In this case, we have the Fisher information matrix 
$J_{\vec{0}}^{\seq{m}}$ of the location shift family 
$\{\wp_{\vec{t}_0,\vec{t}| \seq{m}}\}_{\vec{t}}$.
Taking limit $\epsilon \to 0$ in \eqref{3-3-8}, we have
\begin{align}
V[\wp_{\vec{t}_{0},\vec{t}| \seq{m}}] 
\ge \left(J_{\vec{0}}^{\seq{m}}\right)^{-1}
\ge \left(
\tilde{J}_{\vec{t}_0}  
\right)^{-1}.
\Label{3-3-8F}
\end{align}
The equality of (\ref{RLDB}) holds if and only if
$V[\wp_{\vec{t}_{0},\vec{t}| \seq{m}}] 
= V_{\vec{t_0}|W}$
and the equality in the first inequality of \eqref{3-3-8F} holds. 
Due to the same discussion as the proof of Theorem \ref{ThY},
the equality in the first inequality of \eqref{3-3-8F} holds only when 
all the components of the logarithmic derivative of  the distribution family 
$\{\wp_{\vec{t}_{0},\vec{t} |\seq{m}}\}_{\vec{t}}$
equal the linear combinations of the estimate of $t_i$.
This condition is equivalent to the condition that 
the distribution family 
$\{\wp_{\vec{t}_{0},\vec{t} |\seq{m}}\}_{\vec{t}}$
is a distribution family of shifted normal distributions.
Therefore, when $W>0$, the equality condition of Eq. (\ref{RLDB}) is that
$\wp_{\vec{t}_0,\vec{t}| \seq{m}}$ is the normal distribution with average zero
and covariance matrix
$V_{\vec{t_0}|W}$.

Now, we assume that the state $\rho_{\vec{t}_0}$ is not degenerate.
Then, we use Q-LAN to show that there always exists a sequence of POVM $\seq{m}=\{M_n\}$ satisfying the above property.
We rewrite Eq. \eqref{ER7} of Theorem \ref{Th3} as follows.
\begin{align}
\limsup_{n \to \infty}\sup_{\vec{t}_0 \in K}
\sup_{\vec{t} \in \set{\Theta}_{n,x,c(\vec{t}_0)} } n^{\kappa}
\left\| 
\map{T}^{\rm QLAN}_n \left(\rho_{\vec{t}_0+\vec{t}/\sqrt{n}}^{\otimes n}\right)
-G[{\vec{t}},\tilde{J}_{\vec{t}_0}^{-1}] 
\right\|_1 <\infty,\Label{ERT}
\end{align}
where the notation is the same as Theorem \ref{Th3}.
Then, we choose the covariant POVM 
$M_{W}^{\tilde{J}_{\vec{t}_0}^{-1}}$ on the Gaussian system given in Lemma \ref{ANH}.
When the POVM $M_{W}^{\tilde{J}_{\vec{t}_0}^{-1}}$ applied to the system with the state 
$G[{\vec{\theta}},\Gamma] $, the outcome obeys the distribution
$N[\vec{t}, V_{\vec{t_0}|W}] $.
Therefore, due to \eqref{ERT}, 
when we choose $\seq{m} $ to be 
the sequence of POVMs 
$M_n(\set{B}):= (\map{T}^{\rm QLAN}_n)^\dag (M_{W}^{\tilde{J}_{\vec{t}_0}^{-1}} (\set{B}))$ 
satisfies the condition \eqref{DJI}.

Notice that when $W$ has null eigenvalues, $\sqrt{W}^{-1}$ is not properly defined. In this case, we consider $W_{\epsilon}:=W+\epsilon\cdot P_{W}^{\perp}$ for $\epsilon>0$, where $P_W^\perp$ is the projection on the kernel of $W$, i.e. $P_W^\perp\vec{x}=\vec{x}$ if and only if $W\vec{x}=0$. 
Denoting by $J_{\epsilon}^{-1}:=J_{\vec{t}_0}^{-1}+\frac12
\sqrt{W_\epsilon}^{-1}
\left|\sqrt{W_\epsilon}J_{\vec{t}_0}^{-1}D_{\vec{t}_0}J_{\vec{t}_0}^{-1}\sqrt{W_\epsilon}\right|
\sqrt{W_\epsilon}^{-1}$ the 
\begin{align*}
\tr WJ_{\epsilon}^{-1}\le \tr W_{\epsilon}J_{\epsilon}^{-1}-\epsilon\tr J_{\epsilon}^{-1}.
\end{align*}
Meanwhile, since $W_\epsilon>0$ we can repeat the above argument to find a qudit measurement that attains $\tr W_{\epsilon}J_{\epsilon}^{-1}$.
Taking the limit $\epsilon\to 0$ the quantity $\tr WJ_{\epsilon}^{-1}$ converges to the equality of Eq. (\ref{RLDB2}). Therefore, we can still find a sequence of measurements with Fisher information $\{J_\epsilon\}$ that approaches the bound.
\end{proofof}

\subsection{Precision bound for the estimation of generic models}\Label{s63}
In the previous subsection, we established the precision bound for D-invariant models, where the bound is attainable and has a closed form. Here we extend the bound to any $n$-copy qudit models. The main idea is to \emph{extend} the model to a larger D-invariant model by introducing   additional parameters. 

When estimating parameters in a generic model $\spc{S}$ (consisting of states generated by noisy evolutions, for instance),  the bound (\ref{RLDB}) may not hold. It is then convenient to \emph{extend} the model to a D-invariant model $\spc{S}'$ which contains $\spc{S}$. Since  the bound (\ref{RLDB}) holds for the new model $\spc{S}'$, a corresponding bound  can be derived for the original model $\spc{S}$. The new model $\spc{S}'$ has some additional parameters other than those of $\spc{S}$, which are fixed in the original model $\spc{S}$.
Therefore, a generic quantum state estimation problem can be regarded as an estimation problem in a D-invariant model with \emph{fixed parameters}. The task is to estimate parameters in a model $\spc{S}'$ (globally) parameterized as $\vec{t}'_0=(\vec{t}_0,\vec{p}_0)\in\set{\Theta}'$,  
where $\vec{p}_0$ is a fixed vector and $\set{\Theta}'$ is an open subset of $\R^{k'}$ that equals $\set{\Theta}$ when restricted to $\R^k$. In the neighborhood of $\vec{t}'_0$, since the vector $\vec{p}_0$ is fixed, we have $\vec{t}'=(\vec{t},\vec{0})$ with $\vec{0}$ being the null vector of $\R^{k'-k}$ and $\vec{t}\in\R^{k}$ being a vector of free parameters.
For this scenario, only the parameters in $\vec{t}$ need to be estimated
and we know the parameters $\vec{p}_0$.
Hence, the MSE of $\vec{t}'$ is of the form
\begin{align*}
 V[\wp_{\vec{t}'_0,\vec{t}'| \seq{m}}]=
 \left(\begin{array}{cc}  V[\wp_{\vec{t}_0,\vec{t}| \seq{m}}]&0  \\ 0 & 0 \end{array}\right)
\end{align*}
for any local asymptotic covariant measurement sequence $\seq{m}$.
Due to the block diagonal form of the MSE matrix,  
to discuss the weight matrix $W$ in the original model $\spc{S}$,  
we consider the weight matrix $W'=P^TWP$ 
in the D-invariant model $\spc{S}'$, 
where $P$ is any $k\times k'$ matrix satisfying the constraint $(P)_{ij}:=\delta_{ij}$ for $i,j\le k$
in the following way.

\begin{theo} [MSE bound for generic models]\Label{prop-fixed}
The models $\spc{S}$ and  $\spc{S}'$ are 
$C^1$ continuous and are given in the same way as Proposition \ref{LGT},
and the notations are the same as Proposition \ref{LGT}.
Also, we assume that $\tilde{J}_{\vec{t}_0}^{-1} $ exists.
Consider any sequence of locally asymptotically covariant measurements $\seq{m}:=\{M_n\}$.
Then, the MSE matrx of 
the limiting distribution is evaluated as $\wp_{\vec{t}_{0},\vec{t}| \seq{m}}$.
There exists a $k \times k'$ matrix $P$ such that
\begin{align}
P_{ij}& =\delta_{ij} \hbox{ for }i,j\le k \\
V[\wp_{\vec{t}_{0},\vec{t}| \seq{m}}] 
& \ge P (\tilde{J}_{\vec{t}_0'})^{-1} P^T,\Label{HB0}
\end{align}
where $\tilde{J}_{\vec{t}_0'}$ is the RLD quantum Fisher information of $S'$.
When the model ${\cal M}'$ is a D-invariant model,
we have the bound for the weighted MSE with weight matrix $W\ge 0$
of the limiting distribution as
\begin{align}\Label{HB}
\tr W V[\wp_{\vec{t}_{0},\vec{t}| \seq{m}}] 
\ge \map{C}_{{\rm H},{\cal M}}(W,\vec{t}_{0}).
\end{align}
with $\map{C}_{{\rm H},{\cal M}}(W,\vec{t}_{0})$ defined in Eq. (\ref{H-quantity2}).
When the model ${\cal M}'$ is a D-invariant qudit model
and the state $\rho_{\vec{t}_0}$ is not degenerate, we have
\begin{align}\Label{HB2}
\map{C}_{{\cal M}}(W,\vec{t}_{0})
= \map{C}_{{\rm H},{\cal M}}(W,\vec{t}_{0}).
\end{align}
Moreover, if $W>0$
and $\wp_{\vec{t}_0,\vec{0}| \seq{m}}$ has a differentiable PDF, 
the equality in (\ref{HB}) holds if and only if $\wp_{\vec{t}_0,\vec{t}| \seq{m}}$ is the normal distribution with average zero
and covariance matrix
$V_{\vec{t_0}|W}:=Re(Z_{\vec{t}_0}(\vec{X}))+ 
\sqrt{W}^{-1} |\sqrt{W} {\sf Im} ( Z_{\vec{t}_0}(\vec{X}))\sqrt{W}| \sqrt{W}^{-1}$,
where $\vec{X}$ is the vector to realize the minimum \eqref{H-quantity11}. 
Further, when the models $\spc{S}$ and  $\spc{S}'$ are qudit-models with $C^2$ continuous parametrization,
the equality in \eqref{HB} holds, i.e.,
there exist a sequence of POVMs $M_W^{\vec{t}_0,n}$,
a compact set $K$, and constant $c(\vec{t}_0)$
 such that
\begin{align}
\limsup_{n \to \infty}\sup_{\vec{t}_0 \in K}
\sup_{\vec{t} \in \Theta_{n,x,c(\vec{t}_0)} } n^{\kappa}\left\| \wp^n_{\vec{t}_0,\vec{t}| M_W^{\vec{t}_0,n}}
- N[\vec{t}, V_{\vec{t_0}|W}]
\right\|_1 <\infty.\Label{DJI2}
\end{align}
\end{theo}
Theorem \ref{prop-fixed} determines the ultimate precision limit for generic qudit models. Now, we compare it with the most general existing bound on quantum state estimation, namely  Holevo's bound \cite{holevo-book}. Let us define the ultimate precision of {\em unbiased} measurements as
\begin{align*}
 \map{C}_{\set{UB}_{{\cal M}}}(W,\vec{t}_{0}):=
 \lim_{n\to\infty}\min_{M_n \in \set{UB}_{{\cal M}} }
  \tr W V \left[
 \wp^n_{\vec{t}_{0},\vec{t}|M_n}\right].
\end{align*}
Since the Holevo bound still holds with the $n$-copy case,
(see \cite[Lemma 4]{hayashi2008asymptotic}) 
we have
\begin{align}\Label{bound-H}
 \map{C}_{\set{UB}_{{\cal M}}}(W,\vec{t}_{0})\ge 
  \map{C}_{{\rm H},{\cal M}}(W,\vec{t}_{0}).
\end{align}
There are a couple of differences between our results and existing results:
The Holevo bound is derived under unbiasedness assumption, which, as mentioned earlier, is more restrictive than local asymptotic covariance. 
Our bound (\ref{HB}) thus applies to a wider class of measurements than the Holevo bound.

Furthermore, Yamagata et. al. \cite{yamagata2013quantum}
showed a similar statement as \eqref{DJI2} of Theorem \ref{prop-fixed} in a local model scenario.
They did not show the compact uniformity of the convergence 
and had no order estimation of the convergence.  
However, our evaluation \eqref{DJI2} guarantees 
the compact uniformity with the order estimation.
Then, they did not discuss an estimator to attain the bound globally.
Later, we will construct an estimator to attain our bound globally based on 
the estimator given in Theorem \ref{prop-fixed}.
Our detailed evaluation with the compact uniformity and the order estimation
enables us to evaluate the performance of such an estimator globally.

\begin{proofof}{Theorem \ref{prop-fixed}}
\par\noindent{\it Impossibility part (Proofs of \eqref{HB0} and \eqref{HB}):}

We denote the $\epsilon$-difference Fisher information matrices
for the distribution family
$\{\wp^n_{\vec{t}_{0},\vec{t}|M_n}\}_{\vec{t}}$ 
and
$\{\wp_{\vec{t}_{0},\vec{t} |\seq{m}}\}_{\vec{t}}$
by 
$J_{\vec{t},\epsilon}^n$ and $J_{\vec{t},\epsilon}^{\seq{m}}$, respectively.
Also, we denote the $\epsilon$-difference type RLD Fisher information matrix
at $\vec{t}_0'= (\vec{t}_0,\vec{0})$ of the family $\{\rho_{\vec{t}'}^{\otimes n}\}_{\vec{t}'}$
by $\tilde{J}_{\vec{t}_0',\epsilon}^{n}$. 
Then, we have \eqref{3-3-5} in the same way.

Applying \eqref{3-3-F} of Lemma \ref{LO8} with $\epsilon\to\epsilon/\sqrt{n}$, 
there exist $k \times k'$ matrices $P_n$ such that
 \begin{align}
(P_{n})_{ij} &=\delta_{ij} \qquad\hbox{ for }i,j\le k \\
\frac{1}{n}\left(P_n (\tilde{J}_{\vec{t}_0', \epsilon/\sqrt{n}}^{n} )^{-1} P_n^T\right)^{-1}
&\ge
J_{\vec{t}_0,\epsilon}^n.\Label{3-3-4H}
 \end{align}
Hence, 
the combination of \eqref{KBI} of Lemma \ref{LLO}, \eqref{3-3-4H}, and \eqref{3-3-5}
implies that
there exists a $k \times k'$ matrices $P$
such that
 \begin{align}
P_{ij}& =\delta_{ij} \qquad\hbox{ for }i,j\le k \\
\left(P(\tilde{J}_{\vec{t}_0'}^{[\epsilon]})^{-1}P^T\right)^{-1}
& \ge J_{\vec{t}_0,\epsilon}^{\seq{m}}.
\Label{3-3-6V}
 \end{align}
Due to the same reason as \eqref{3-3-8}, we have 
\begin{align}
V[\wp_{\vec{t}_{0},\vec{t}| \seq{m}}] 
\ge (J_{\vec{t}_0,\epsilon}^{\seq{m}})^{-1}
\ge
P(\tilde{J}_{\vec{t}_0'}^{[\epsilon]})^{-1}P^T.\Label{3-3-8V}
\end{align}
By taking the limit $\epsilon \to 0$, 
the combination of \eqref{2KBI} of Lemma \ref{LLO} and \eqref{3-3-8V} implies \eqref{HB0}.

When the model ${\cal M}'$ is D-invariant, 
since
\begin{align*}
\tr W V[\wp_{\vec{t}_{0},\vec{t}| \seq{m}}] 
\ge
\tr P^T W P\tilde{J}_{\vec{t}_0'}^{-1},
\end{align*}
we obtain \eqref{HB} by using the expression \eqref{Ho2}
in the same way as \eqref{RLDB}.

\noindent{\it Achievability part (Proof of \eqref{HB2}):}

Since $\rho_{\vec{t}_0'}$ is not degenerate, 
we can show the achievability in the same way as Theorem \ref{thm-MSE}
because we can apply Q-LAN (Theorem \ref{Th3}) for the model ${\cal M}'$.
The difference is the following.
Choosing the matrix $P$ to achieve the minimum \eqref{Ho2}, 
we employ the covariant POVM $M_{P|W}^{\tilde{J}_{\vec{t}_0'}^{-1}}$
instead of the covariant POVM $M_W^{\tilde{J}_{\vec{t}_0}^{-1}}$.
Then, we obtain the desired statement.
\end{proofof}

\section{Nuisance parameters }\Label{sec-general}
For state estimation in a noisy environment, the strength of noise is not a parameter of interest, yet it affects the precision of estimating other parameters. 
In this scenario, the strength of noise is  a \emph{nuisance parameter} \cite{cox1994inference,lehmann2006theory}. To illustrate the difference between nuisance parameters and fixed parameters that are discussed in the previous section, let us consider the case of a qubit clock state under going a noisy time evolution. To estimate the duration of the evolution, we introduce the strength of the noise as an additional parameter and consider the estimation problem in the extended model parameterized by the duration and the noise strength. The strength of the noise is usually unknown.  Although it is not a parameter of interest, its value will affect the precision of our estimation, and thus it should be treated as a \emph{nuisance} parameter. 

\subsection{Precision bound for estimation with nuisance parameters}

In this subsection, 
we consider state estimation of an arbitrary $(k+s)$-parameter model 
$\{ \rho_{\vec{t},\vec{p}}\}_{(\vec{t},\vec{p})\in \tilde{\set{\Theta}}}$,
where $\vec{t}$ and $\vec{p}$ 
are $k$-dimensional and $s$-dimensional parameters, respectively.
Our task is to estimate only the parameters $\vec{t}$ and 
it is not required to estimate the other parameters $\vec{p}$, which is called 
nuisance parameters.
Hence, our estimate is $k$-dimensional.
We say that a parametric family of a structure of nuisance parameters
is a nuisance parameter model,
and denote it by $\tilde{\spc{S}}=\{ \rho_{\vec{t},\vec{p}}\}_{(\vec{t},\vec{p})\in\tilde{\set{\Theta}}}$.
We simplify $(\vec{t},\vec{p})$ by $\tilde{\vec{t}}$.

The concept of local asymptotic covariance can be extended to a model with nuisance parameters
by considering a local model $\rho^n_{\tilde{\vec{t}}_0,\tilde{\vec{t}}}
:=\rho^{\otimes n}_{\tilde{\vec{t}}_0+\tilde{\vec{t}}/\sqrt{n}}$.
Throughout this section, we assume that $\rho_{\tilde{\vec{t}}_0}$ is invertible
and all the parametrizations are at least $C^1$ continuous.


\begin{condition}[Local asymptotic covariance with nuisance parameters] \Label{Ncond}
We say that a sequence of measurements $\seq{m}:=\{M_n\}$ 
to estimate the $k$-dimensional parameter $\vec{t}$
satisfies local asymptotic covariance at $\tilde{\vec{t}}_0 =(\vec{t}_0,\vec{p}_0) \in \tilde{\set{\Theta}}$ 
 under the nuisance parameter model $\tilde{{\cal M}}$
when the probability distribution 
\begin{align}\Label{finitedistributionN}
\wp^{n}_{\tilde{\vec{t}}_0,\tilde{\vec{t}}|M_n}(\set{B})
:=\Tr\rho_{\tilde{\vec{t}}_0+\frac{\tilde{\vec{t}}}{\sqrt{n}}} ^{\otimes n}M_n
\left( \frac{\set{B}}{\sqrt{n}}+\vec{t}_0\right) 
\end{align}
converges to a limiting distribution 
\begin{align}\Label{limitingB}
\wp_{\tilde{\vec{t}}_0,\tilde{\vec{t}} |\seq{m}}(\set{B}):=
\lim_{n\to\infty} \wp^n_{\tilde{\vec{t}}_0,\tilde{\vec{t}}|M_n}(\set{B}),
\end{align}
the relation
\begin{align}
\wp_{\tilde{\vec{t}}_0,\tilde{\vec{t}}| \seq{m}}(\set{B}+\vec{t})
=
\wp_{\tilde{\vec{t}}_0,(\vec{0},\vec{0})| \seq{m}}(\set{B})
\Label{LLe2}
\end{align}
holds for any $\tilde{\vec{t}}=(\vec{t},\vec{p})\in\R^{k+s}$.

Further, we say that a sequence of measurements $\seq{m}:=\{M_n\}$ satisfies local asymptotic covariance under the nuisance parameter model $\tilde{{\cal M}}$
when it satisfies local asymptotic covariance 
at any element $\tilde{\vec{t}}_0 \in \tilde{\set{\Theta}}$ under the nuisance parameter model $\tilde{{\cal M}}$.
\end{condition}

The quantity we want to bound is 
$\tr V[\wp_{\vec{t}_0,\vec{t}|\seq{m}}] W$,
where $\wp_{\vec{t}_0,\vec{t}|\seq{m}}$ is the limiting distribution of a sequence $\seq{m}$  of locally asymptotically covariant measurements  and $W$ is a weight matrix. Since nuisance parameters are not of interest, the weight matrix of the model $\tilde{\spc{S}}$ is a $(k+s)\times (k+s)$ matrix of  the form
\begin{align}\Label{W'}
\tilde{W}=\left(\begin{array}{cc} W & 0\\ 0& 0\end{array}\right).
\end{align}

\begin{lem}\Label{NLGT}
Let $\tilde{\spc{S}}=\{\rho_{\tilde{\vec{t}}}\}_{\tilde{\vec{t}}\in \tilde{\set{\Theta}}}$  
be a $(k+s)$-parameter nuisance parameter model
and let $\spc{S}'=\{\rho_{\vec{t}'}\}_{\vec{t}'=(\tilde{\vec{t}},\vec{q})}$ be a 
$k'$-parameter  model containing $\tilde{\spc{S}}$
as $\rho_{\tilde{\vec{t}}}=\rho_{(\tilde{\vec{t}},\vec{0})}$. 
When $\spc{S}'$ is D-invariant
and the inverse $\tilde{J}_{\tilde{\vec{t}}_0}^{-1} $ exists, we have
\begin{align}
\map{C}_{{\rm NH},{\cal M}}(W,\tilde{\vec{t}}_{0})
&:= \min_{\vec{X}} \min_V \{\tr W V| V \ge Z_{\tilde{\vec{t}}_0}(\vec{X})\} \Label{NH-quantity} \\
&=\min_{P}\left\{\tr P^TWP(J_{\vec{t}_0'}^{-1})+\frac12\tr\left|\sqrt{P^TWP}
J_{\vec{t}_0'}^{-1}D_{\vec{t}_0'}J_{\vec{t}_0'}^{-1}\sqrt{P^TWP}\right|\right\}.\Label{NHo2}
\end{align}
In \eqref{NH-quantity}, $V$ is a real symmetric matrix and 
$\vec{X}=(X_i)$ is a $k$-component vector of operators to satisfy
\begin{align}
 \Tr
 X_i
 \frac{\partial\rho_{\tilde{\vec{t}}}}{\partial \tilde{t}_j}
 \Big|_{\tilde{\vec{t}}=\tilde{\vec{t}}_0 }=\delta_{ij}, \ \forall
\,i \le k,\quad\forall\,j \le k+s.\Label{MBG}
\end{align}
In \eqref{NHo2},
the minimization is taken over all $k\times (k+s)$ matrices satisfying the constraint 
$(P)_{ij}:=\delta_{ij}$ for $i \le k,j\le k+s$,
and,  $J_{\vec{t}_0'}$ and $D_{\vec{t}_0'}$ are the SLD Fisher information matrix  and the D-matrix [cf. Eqs. (\ref{SLD-QFI}) and (\ref{D-m})] for the extended model  $\spc{S}'$ at 
$\vec{t}_0':=(\tilde{\vec{t}}_0,\vec{0})$.
\end{lem}
This lemma is a different statement from \cite[Theorem 4]{hayashi2008asymptotic}.
However, using the method of \cite[Theorem 4]{hayashi2008asymptotic}, we can show this lemma.

 In the following, we derive an {\em attainable} lower bound on the  weighted MSE. To this purpose,  we define the set $\set{LAC}(\tilde{\vec{t}}_{0})$ of 
local asymptotic covariant sequences of measurements at the point $\tilde{\vec{t}}_0 \in \tilde{\set{\Theta}}$
for the nuisance parameter model $\tilde{{\cal M}}$, and
focus on the minimum value 
\begin{align}\Label{global-boundBY}
\map{C}_{{\rm N},\tilde{{\cal M}}}(W,\tilde{\vec{t}}_{0})
:=\min_{\seq{m}\in \set{LAC}(\tilde{\vec{t}}_{0}) }
 \tr W V [\wp_{\tilde{\vec{t}}_{0},\tilde{\vec{t}} | \seq{m}}].
\end{align}

\begin{theo}[Weighted MSE bound with nuisance parameters]\Label{theo-nuisance}
The models $\tilde{\spc{S}}$ and  $\spc{S}'$ are given in the same way as Lemma \ref{NLGT},
and the notations is the same as Lemma \ref{NLGT}.
Also, we assume that $(\tilde{J}_{\vec{t}'_0})^{-1}$ exists.
Consider any sequence of locally asymptotically covariant measurements $\seq{m}:=\{M_n\}$
for the nuisance parameter model $\tilde{\spc{S}}$.
Then, the MSE matrx of 
the limiting distribution is evaluated as follows.
There exists a $k \times k'$ matrix $P$ such that
\begin{align}
P_{ij}& =\delta_{ij} \qquad\hbox{ for }1\le i\le k, 1\le j \le k' \Label{NFE}\\
V[\wp_{\tilde{\vec{t}}_0,\tilde{\vec{t}}| \seq{m}}] 
& \ge P (\tilde{J}_{\vec{t}'_0})^{-1} P^T.\Label{NHB0}
\end{align}
When the model $\spc{S}'$ is D-invariant, 
we have the bound for the weighted MSE with weight matrix $W\ge 0$
of the limiting distribution as
\begin{align}\Label{NHB}
\tr W V[\wp_{\tilde{\vec{t}}_0,\tilde{\vec{t}}| \seq{m}}] 
\ge \map{C}_{{\rm NH},{\cal M}}(W,\tilde{\vec{t}}_{0}).
\end{align}
When the model $\spc{S}'$ is a D-invariant qudit model and the state $\rho_{\vec{t}_0}$ is not degenerate,
we have
\begin{align}\Label{NHB2}
\map{C}_{{\rm N},\tilde{{\cal M}}}(W,\tilde{\vec{t}}_{0})
= \map{C}_{{\rm NH},\tilde{{\cal M}}}(W,\tilde{\vec{t}}_0).
\end{align}
Moreover, if $W>0$
and $\wp_{\tilde{\vec{t}}_0,\tilde{\vec{t}}| \seq{m}}$ has a differentiable PDF, 
the equality in (\ref{NHB}) holds if and only if $\wp_{\vec{t}_0,\vec{t}| \seq{m}}$ is the normal distribution with average zero
and covariance
\begin{align}\label{V-condition-W}
V_{\vec{t}_0|W}:=Re(Z_{{\vec{t}}_0}(\vec{X}))+ 
\sqrt{W}^{-1} |\sqrt{W} {\sf Im} ( Z_{\vec{t}_0}(\vec{X}))\sqrt{W}| \sqrt{W}^{-1},
\end{align}
where $\vec{X}$ is the vector to realize the minimum \eqref{NH-quantity}. 
Further, when the models $\tilde{\spc{S}}$ and  $\spc{S}'$ are qudit-models with $C^2$ continuous parametrization,
the equality in \eqref{NHB} holds, i.e.,
there exist a sequence of POVMs $M_W^{\vec{t}_0,n}$,
a compact set $K$, and constant $c(\vec{t}_0)$
 such that
\begin{align}
\limsup_{n \to \infty}\sup_{\vec{t}_0 \in K}
\sup_{\vec{t} \in \set{\Theta}_{n,x,c(\vec{t}_0)} } n^{\kappa}\| \wp^n_{\vec{t}_0,\vec{t}| M_W^{\vec{t}_0,n}}
- N[\vec{t}, V_{\vec{t_0}|W}]
\|_1 <\infty.\Label{DJI3}
\end{align}
Here  $\kappa$ is a parameter to satisfy $\kappa\ge 0.027$.
\end{theo}

Before proving Theorem \ref{theo-nuisance},
we discuss a linear subfamily of 
$k'$-dimensional Gaussian family 
$\{G[\vec{t}', \vec{\gamma}]\}_{\vec{t}'\in \mathbb{R}^{k'}}$.
Consider a linear map $T$ from $\mathbb{R}^{(k+s)}$ to 
$\mathbb{R}^{k'}$.
We have the subfamily $\tilde{{\cal M}}:=\{G[ T(\vec{t},\vec{p}), \vec{\gamma}]\}_{(\vec{t},\vec{p})\in \mathbb{R}^{k+s}}$ as a nuisance parameter model.
Then, the covariance condition is extended as follows.
\begin{defi}
A POVM $M$ is {\it unbiased} for the nuisance parameter model
$\{\rho_{(\vec{t},\vec{p})}\}$
when 
\begin{align*}
\vec{t}=E_{\vec{t},\vec{p}}(M):=\int \vec{x} \Tr \rho_{(\vec{t},\vec{p})} M(d \vec{x})
\end{align*} 
holds for any parameter $(\vec{t},\vec{p})$.
A POVM $M$ is a {\it covariant} estimator for the nuisance parameter model
$\{G[ T(\vec{t},\vec{p}), \vec{\gamma}]\}$
when
the distribution $\wp_{(\vec{t},\vec{p})|M}$
satisfies the condition $\wp_{\vec{0},\vec{0}|M}(\set{B})=\wp_{\vec{t},\vec{p}|M}(\set{B}+\vec{t})$.
\end{defi}

Then, we have the following corollary of Lemma \ref{ANH}.

\begin{cor}\Label{LFD3}
For any weight matrix $W\ge 0$, 
the nuisance parameter model
$\tilde{\cal M}=\{G[ T(\vec{t},\vec{p}), \vec{\gamma}]\}$ with $C^1$ continuous parametrization
satisfies
\begin{align}
& \inf_{M \in \set{UB}_{\tilde{{\cal M}}}}\tr W V_{\vec{t}}(M) 
=\inf_{M \in \set{CUB}_{\tilde{{\cal M}}}}\tr W V_{\vec{t}}(M) 
= \map{C}_{{\rm NH},\tilde{\cal M}}(W,\vec{t}),
\Label{DIT6}
\end{align}
where $\set{UB}_{\tilde{{\cal M}}}$ and $\set{CUB}_{\tilde{{\cal M}}}$ 
are the sets of unbiased estimators and covariant unbiased estimators
of the nuisance parameter model $\tilde{\cal M}$, respectively.
Further, 
when $W>0$, 
we choose a vector $\vec{X}$ to realize the minimum in  \eqref{H-quantity11}.
The above infimum is attained by 
the covariant unbiased estimators $M_W$ whose output distribution is 
the normal distribution with average $\vec{t}$ and covariance matrix
 ${\sf Re} ( (Z_{\vec{t}}(\vec{X}))+
\sqrt{W}^{-1} |\sqrt{W} {\sf Im} ( Z_{\vec{t}}(\vec{X}))\sqrt{W}| \sqrt{W}^{-1}$.
\end{cor}

This corollary can be shown as follows.
The inequality $ \inf_{M \in \set{UB}_{\tilde{{\cal M}}}}\tr W V_{\vec{t}}(M) 
\ge \map{C}_{{\rm NH},\tilde{\cal M}}(W,\vec{t})$
follows from the condition \eqref{MBG}.
Similar to Corollary \ref{LFD2}, 
Proposition \ref{LGT} guarantees that 
the latter part of the corollary with $W>0$ follows from \eqref{NH-quantity} and Lemma \ref{ANH}.
Hence, we obtain this corollary for $W>0$.
The case with non strictly positive $W$ can be shown by considering $W_\epsilon$ in the same way as Corollary \ref{LFD}.

\begin{proofof}{Theorem \ref{theo-nuisance}}
\par\noindent{\it Impossibility part (Proofs of \eqref{NHB0} and \eqref{NHB}):}

We denote the $\epsilon$-difference Fisher information matrix of 
$\{\wp_{\tilde{\vec{t}}_0,\tilde{\vec{t}} |\seq{m}}\}_{\tilde{\vec{t}}}$
by $J_{\tilde{\vec{t}}_0,\epsilon}^{\seq{m}}$.
Due to \eqref{3-3-6V}, there exists a $(k+s) \times k'$ matrix $\tilde{P}$ satisfying the following conditions.
\begin{align}\label{AAGGG}
\tilde{P}_{ij}& =\delta_{ij} \hbox{ for } 1\le i, j\le k+s , \quad
\left(P(\tilde{J}_{\vec{t}_0'}^{[\epsilon]})^{-1}P^T\right)^{-1}
 \ge J_{\vec{t}_0,\epsilon}^{\seq{m}}.
\end{align}
We define the $k \times (k+s)$ matrix $\bar{P}$ by
\begin{align*}
\bar{P}_{ij}& =\delta_{ij} \hbox{ for } 1\le i \le k, 1\le j\le k+s .
\end{align*}
Now, we extend Theorem \ref{LO5}.
Let $E_{(\vec{0},\vec{0})}[\hat{\vec{t}}]$
be the expectation of $\hat{\vec{t}}$ under the distribution
$\wp_{\vec{t}_{0},(\vec{0},\vec{0}) |\seq{m}}$.
We denote the Radon-Nikod\'{y}m derivative of 
$\wp_{\tilde{\vec{t}}_0,\epsilon e_i |\seq{m}}$
with respect to $\wp_{\tilde{\vec{t}}_0,0|\seq{m}}$ by $p_i$.
Then, for two vectors $\vec{a} \in \mathbb{R}^k$ and $\vec{b}\in \mathbb{R}^{k+s}$,
we apply Schwartz inequality to the two variables
$X:= \sum_j b_j(\hat{\vec{t}}-E_{(\vec{0},\vec{0})}[\hat{\vec{t}}])_j $
and $Y:= \sum_i  \frac{a_i}{\epsilon} (p_j(x)-1)$.
Using $\tilde{X}:= \sum_j b_j\vec{t}_j $,
we obtain 
\begin{align}
& \langle \vec{b} |V[\wp_{\tilde{\vec{t}}_{0},(\vec{0},\vec{0})| \seq{m}}] |\vec{b} \rangle
\langle \vec{a} |J_{\tilde{\vec{t}}_0,\epsilon}^{\seq{m}}|\vec{a} \rangle
=
\int \tilde{X}(\vec{x})^2 \wp_{\tilde{\vec{t}}_0,(\vec{0},\vec{0})|\seq{m}}(d\vec{x})
\int Y(\vec{x})^2 \wp_{\tilde{\vec{t}}_0,(\vec{0},\vec{0})|\seq{m}}(d\vec{x}) \nonumber \\
\ge &
\int X(\vec{x})^2 \wp_{\tilde{\vec{t}}_0,(\vec{0},\vec{0})|\seq{m}}(d\vec{x})
\int Y(\vec{x})^2 \wp_{\tilde{\vec{t}}_0,(\vec{0},\vec{0})|\seq{m}}(d\vec{x}) \nonumber \\
\ge &
\Big|\int X(\vec{x})Y(\vec{x}) \wp_{\tilde{\vec{t}}_0,(\vec{0},\vec{0})|\seq{m}}(d\vec{x})\Big|^2
=\Big|\langle \vec{b} | \bar{P} |\vec{a}\rangle\Big|^2,
\end{align}
where the final equation follows from 
the fact that 
the expectation of the variable $\hat{\vec{t}}-E_{(\vec{0},\vec{0})}[\hat{\vec{t}}]$ 
equals $\vec{t}= \bar{P} \tilde{\vec{t}}$
under the distribution $\wp_{\tilde{\vec{t}}_{0},\tilde{\vec{t}} |\seq{m}}$,
which can be shown by the covariance condition for the distribution family 
$\{\wp_{\tilde{\vec{t}}_{0},\tilde{\vec{t}} |\seq{m}}\}_{\tilde{\vec{t}}}$.

Choosing $\vec{a}:= (J_{\tilde{\vec{t}}_0,\epsilon}^{\seq{m}})^{-1} \bar{P}^T
\vec{b}$, we have
\begin{align}
\langle \vec{b} |
V[\wp_{\tilde{\vec{t}}_{0},(\vec{0},\vec{0})| \seq{m}}] 
| \vec{b} \rangle
\ge \langle \vec{b} |
\bar{P} (J_{\tilde{\vec{t}}_0,\epsilon}^{\seq{m}})^{-1}\bar{P}^T| \vec{b} \rangle,\Label{3-3-8T}
\end{align}
which implies
\begin{align}
V[\wp_{\tilde{\vec{t}}_{0},\tilde{\vec{t}}| \seq{m}}] 
=V[\wp_{\tilde{\vec{t}}_{0},(\vec{0},\vec{0})| \seq{m}}] 
\ge \bar{P} (J_{\tilde{\vec{t}}_0,\epsilon}^{\seq{m}})^{-1}\bar{P}^T
.\Label{3-3-8E}
\end{align}
Combining the above with Eq. (\ref{AAGGG}),
since $P:=\bar{P} \tilde{P}$ satisfies the condition \eqref{NFE},
we obtain \eqref{NHB0}.

When the model ${\cal M}'$ is D-invariant, 
since
\begin{align*}
\tr W V[\wp_{\tilde{\vec{t}}_0,\tilde{\vec{t}}| \seq{m}}] 
\ge
\tr P^T W P\tilde{J}_{\vec{t}_0'}^{-1},
\end{align*}
we obtain \eqref{NHB} by using the expression \eqref{NHo2}
in the same way as \eqref{HB}.

\noindent{\it Achievability part (Proof of \eqref{NHB2}):}

Since $\rho_{\vec{t}_0}$ is not degenerate, 
we can show the achievability part in the same way as Theorem \ref{thm-MSE}
because we can apply Q-LAN (Theorem \ref{Th3}) for the model ${\cal M}'$.
The difference is the following.
Instead of  Corollary \ref{LFD}, we employ Corollary \ref{LFD3} to choose 
the covariant POVM $M_{P|W}^{\tilde{J}_{\vec{t}'_0}^{-1}}$. 
Then, we obtain the desired statement.
\end{proofof}

\subsection{Nuisance parameter with D-invariant model}
Next, we discuss the nuisance parameters when the model is D-invariant.

\begin{lem}
When $\tilde{\spc{S}}=\{\rho_{(\vec{t},\vec{p})}\}_{(\vec{t},\vec{p})\in \tilde{\set{\Theta}}}$  
is a D-invariant $k+s$-parameter nuisance parameter model
and ${J}_{\vec{t}_0}^{-1} $ exists,
we have
\begin{align}\Label{NNHB2}
\map{C}_{{\rm NH},\tilde{{\cal M}}}(W,\tilde{\vec{t}}_{0})
=\tr \tilde{W}(J_{\tilde{\vec{t}}_0}^{-1})+\frac12\tr\left|\sqrt{\tilde{W}}J_{\tilde{\vec{t}}_0}^{-1}D_{\tilde{\vec{t}}_0}J_{\tilde{\vec{t}}_0}^{-1}\sqrt{\tilde{W}}\right|.
\end{align}
\end{lem}

A few comments are in order.   First, the nuisance parameter bound (\ref{NHB}) reduces to the bound (\ref{RLDB}), when the parameters to estimate are \emph{orthogonal} to the nuisance parameters
in the sense that the RLD Fisher information matrix $\tilde{J}_{\tilde{\vec{t}}_0}$ is block-diagonal.
This orthogonality is equivalent to the condition that 
the SLD Fisher information matrix $J_{\tilde{\vec{t}}_0}$ and the D-matrix take the block diagonal forms
\begin{align}\Label{orthomatrix}
J_{\tilde{\vec{t}}_0}=\left(\begin{array}{cc} J_{\vec{t}_0} & 0\\ 0 & J_{\rm N}\end{array}\right)\quad D_{\tilde{\vec{t}}_0}=\left(\begin{array}{cc} D_{\vec{t}_0} & 0\\ 0 & D_{\rm N}\end{array}\right). 
\end{align}
This is the case, for instance, of simultaneous estimation of the spectrum and the Hamiltonian-generated phase of a two-level system.
Under such circumstances, the inverse of the Fisher information matrix can be done by   inverting $J_{\vec{t}_0}$ and $J_{\rm N}$ independently. The same precision bound is thus obtained with or without introducing nuisance parameters, and we have the following lemma.
\begin{lem}\Label{lem-orthogonal}
When all nuisance parameters are orthogonal to the parameters of interest,  the bound with nuisance parameters (\ref{NHB}) coincides with the D-invariant MSE bound (\ref{RLDB}).
\end{lem}
In the case of orthogonal nuisance parameters, the estimation of nuisance parameters does not affect the precision of estimating the parameters of interest, which does not hold for the generic case of non-orthogonal nuisance parameters. Thanks to this fact, one can achieve the bound (\ref{NHB}) by first measuring the nuisance parameters and then constructing the optimal measurement based on the estimated value of the nuisance parameters. On the other hand,  an RLD bound [cf. Eq. (\ref{BFR})] can be attained if and only if its model is D-invariant. Combining these arguments with Lemma \ref{lem-orthogonal}, we obtain a characterization of the attainability of RLD bounds as follows.
\begin{cor}\Label{prop-orthogonal}
An RLD bound can be achieved if and only if it has an orthogonal nuisance extension, i.e. Eq. (\ref{orthomatrix}) holds for some choice of nuisance parameters.
\end{cor}
The above corollary offers a simple criterion for the important problem of the attainability of RLD bounds.  In Section \ref{subsec-multiphase}, we will  illustrate the application of this criterion  with  a concrete example. 

The bound (\ref{NHB})  can be straightforwardly computed even for complex models; for D-invariant models, the SLD operators have an uniform entry-wise expression and one only needs to shot it into a program to yield the bound (\ref{NHB}). Moreover, the bound does not rely on the explicit choice of nuisance parameters. To see this, one can consider another parameterization $\vec{x}'$ of the D-invariant model. The bound (\ref{NHB}) comes from the RLD bound for the D-invariant model, and the RLD quantum Fisher information matrices $\widetilde{J}_{\vec{t}_0'}$ and $\widetilde{J}_{\vec{x}'_0}$ for two parameterizations are connected by the equation $\widetilde{J}_{\vec{t}_0'}=A\widetilde{J}_{\vec{x}'_0}A^T$, where $A$ is the Jacobian $\left(\partial\vec{x}'/\partial\vec{t}'\right)$ at $\vec{t}'_0$. Since both parameterizations are extensions of the same model $\spc{S}$ satisfying $P_0\vec{t}_0'=P'_0\vec{x}'_0=\vec{t}_0$, the Jacobian takes the form 
$$A=\left(\begin{array}{cc} I_{k}&A'\\ 0 &A''\end{array}\right).$$
Then we have $\widetilde{J}_{\vec{x}'_0}^{-1}=A^T \widetilde{J}^{-1}_{\vec{t}_0'}A$, which implies that the upper-left $k\times k$ blocks of $\widetilde{J}_{\vec{x}'_0}$ and $\widetilde{J}_{\vec{t}_0'}$ are equal. The bound (\ref{NHB}) thus remains unchanged.

\subsection{Precision bound for joint measurements.}
A useful implication of Theorem \ref{theo-nuisance} is a bound on MSEs of several jointly measured observables. Consider a set $\{O_i\}$ of $k$ bounded observables. The goal is to jointly measure their expectations
\begin{align}\Label{ob-define}
o_i:=\Tr\rho O_i\qquad i=1,\dots,k.
\end{align}
The main result of this subsection is the following corollary:
\begin{cor}\Label{prop-disturbance}
Define  the SLD gap of $o_i$  as
\begin{align}\Label{disturbance-def}
\Delta_{o_i}:=\mathbf{MSE}_{o_i}-\left(J^{-1}\right)_{ii},
\end{align}
where $\mathbf{MSE}_{o_i}$ denotes the MSE of $o_i$ under joint measurement and $J$ is the SLD quantum Fisher information.
The sum of the SLD gaps for all observables satisfies the attainable bound:
\begin{align}\Label{disturbance-bound}
\sum_{i=1}^d \Delta_{o_i}\ge \frac12\tr\left|\left(J^{-1}DJ^{-1}\right)_{k\times k}\right|,
\end{align}
where $D$ is the D-matrix.
\end{cor}
The right hand side of Eq. (\ref{disturbance-bound}) is exactly the gap between the SLD bound and the ultimate precision limit. It shows a typical example where the SLD bound is not attainable.

\begin{proof}
 Substituting $W'$ in Eq. (\ref{NHB}) by the projection into the subspace $\R^k$, we obtain a bound for the MSE $\{\mathbf{MSE}_{o_i}\}$ of the limiting distributions:
\begin{align}\Label{joint1}
\sum_{i=1}^d \mathbf{MSE}_{o_i}\ge \sum_{i=1}^d \left(J^{-1}\right)_{ii}+\frac12\tr\left|\left(J^{-1}DJ^{-1}\right)_{k\times k}\right|.
\end{align}
 Here $J$ and $D$ are the SLD Fisher information and D-matrix for the extended model, and $\left(A\right)_{k\times k}$ denotes the upper-left $k\times k$ block of a matrix $A$. 
Substituting the above definition into Eq. (\ref{joint1}), we obtain Corollary \ref{prop-disturbance}.\qed
\end{proof}

Specifically, for the case of two parameters, the  bound (\ref{disturbance-bound}) reduces to 
\begin{align}\Label{disturbance-bound-two}
\Delta_{o_1}+\Delta_{o_2}\ge \left|\Tr\rho_{\vec{\theta}}[\hat{L}_1,\hat{L}_2]\right|,
\end{align}
where $\hat{L}_j:=\sum_{j=1}^{k'}\left(J_{\vec{\theta}'}\right)_{ji}^{-1}L_i$ are the SLD operators in the dual space. Next, taking partial derivative with respect to $o_j$ on both sides of Eq. (\ref{ob-define}) and substituting in the definition of RLD operators, the observables satisfy the orthogonality relation with the SLD operators as
\begin{align*}
\frac12\Tr\left(\rho L_j+L_j\rho\right)O_i=\delta_{ij}.
\end{align*}
By uniqueness of the dual space, we have
\begin{align*}
\hat{L_i}=O_i-o_i I\qquad i=1,\dots,k'
\end{align*}
and the bound becomes
\begin{align}\Label{uncertainty}
\Delta_{o_1}+\Delta_{o_2}&\ge|\<[O_1,O_2]|\>|.
\end{align}
Another bound expressing the tradeoff between $\Delta_{o_1}$ and $\Delta_{o_2}$ was obtained by Watanabe et al.  \cite{watanabe} as
\begin{align}\Label{watanabe}
\Delta_{o_1}\Delta_{o_2}\ge|\<[O_1,O_2]\>|^2/4.
\end{align} 
Now, substituting $O_2$ by $\alpha O_2$ for a variable $\alpha\in\R$ in Eq. (\ref{uncertainty}), we have
\begin{align*}
\Delta_{o_1}+\Delta_{\alpha o_2}=\Delta_{o_1}+\alpha^2\Delta_{o_2}\ge \alpha|\<[O_1,O_2]|\>|.
\end{align*}
For the above quadratic inequality to hold for any $\alpha\in\R$, its discriminant must be non-positive, which immediately implies the bound (\ref{watanabe}). Notice that the bound (\ref{watanabe}) was derived under asymptotic unbiasedness \cite{watanabe}, and thus
 it was not guaranteed to be attainable.
Here, instead, since our bound (\ref{uncertainty}) is always attainable, the bound (\ref{watanabe}) can also be achieved in any qudit model under the asymptotically covariant condition.

\subsection{Nuisance parameters versus fixed parameters}
It is intuitive to ask what is the relationship between the nuisance parameter  bound (\ref{NHB}) and the general bound (\ref{HB}).
To see it, 
let $\spc{S}=\{\rho_{\vec{t}}\}_{\vec{t}\in \set{\Theta}}$  be a generic $k$-parameter qudit model and let $\tilde{\spc{S}}$ be a 
$(k+s)$-parameter $D$-invariant model containing $\spc{S}$. 
When $\rho_{\vec{t}_0}$ is non-degenerate,
we notice that the QCR bound with nuisance parameters (\ref{NHB}) can be rewritten as
\begin{align}\Label{NHB1}
 \map{C}_{{\cal M}}(W,\vec{t}_{0})= \tr P_0^TWP_0(J_{\tilde{\vec{t}}_0}^{-1})+\frac12\tr\left|\sqrt{P_0^TWP_0}J_{\tilde{\vec{t}}_0}^{-1}D_{\tilde{\vec{t}}_0}J_{\tilde{\vec{t}}_0}^{-1}\sqrt{P_0^TWP_0}\right|,
\end{align}
where $P_0$ is a $k\times (k+s)$ matrix satisfying the constraint $(P_0)_{ij}:=\delta_{ij}$ for any $i,j\le k+s$.
By definition, $P_0$ is a special case of $P$, and it follows straightforwardly from comparing Eq. (\ref{NHB1}) with Eq. (\ref{HB}) that the general MSE bound is upper bounded by the MSE bound for the nuisance parameter case.  This observation agrees with the obvious intuition that  having additional information on the system is helpful for (or at least, not detrimental to) estimation.
At last, since $J_{\tilde{\vec{t}}_0}$ and $D_{\tilde{\vec{t}}_0}$ are block-diagonal in the case of orthogonal nuisance parameters, we have 
\begin{align*}
P(J_{\tilde{\vec{t}}_0})^{-1}P^T=J^{-1}_{\vec{t}_0}\quad PD_{\tilde{\vec{t}}_0}P^T=D_{\vec{t}_0}
\end{align*}
for any $k\times (k+s)$ matrix satisfying the constraint $(P)_{ij}:=\delta_{ij}$ for $i,j\le k$. This implies that the general bound (\ref{HB}) coincides with the nuisance parameter bound (\ref{NHB}) when the nuisance parameters are orthogonal.

\section{Tail property of the limiting distribution}\Label{sec-tail}
In previous discussions, we focused on the MSE of the limiting distribution. Here, instead, we consider the behavior of the limiting distribution itself. The characteristic property is the tail property: 
Given a weight matrix $W\ge 0$ and a constant $c$,  
we define the tail region $\set{T}_{W,c}(\vec{t})$ as
\begin{align*}
\set{T}_{W,c}(\vec{t}):=\left\{\vec{x}~|~\,(\vec{x}^T-\vec{t}^T)W(\vec{x}-\vec{t})\ge c\right\}.
\end{align*}
For a measurement $\seq{m}=\{M_n(\hat{\vec{t}}_n)\}$, the probability that the estimate $\hat{\vec{t}}_n$ is in the tail region can be approximated by the tail probability of the limiting distribution, i.e.
\begin{align*}
\mathbf{Prob}\left(\hat{\vec{t}}_n\in\set{T}_{W,c}(\vec{t})\right)=\wp_{\vec{t}_0,\vec{t}|\seq{m}}\left(\set{T}_{W,c}(\vec{t})\right)+\epsilon_n,
\end{align*}
up to $\epsilon_n$ being a term vanishing in $n$.
The tail property is usually harder to characterize than   the MSE. Nevertheless, here we show that, under certain conditions, there exists a good bound on the tail property of the limiting distribution.
\subsection{Tail property of Gaussian shift models}
Just like in the previous sections, the tail property of $n$-copy qudit models can  be analyzed by studying the tail property of Gaussian shift models.
In this subsection, we first derive a bound on the tail probability of Gaussian shift models. The result has an interest in its own and can be used for further analysis of qudit models using Q-LAN.

Consider a Gaussian shift model $\{G[\vec{\alpha},\vec{\gamma}]\}$ with $G[\vec{\alpha},\vec{\gamma}]=N[\vec{\alpha}^C,\Gamma]\otimes \Phi[\vec{\alpha}^Q,\vec{\beta}]$ and a measurement  $M^{\rm G}  (\hat \alpha)$.  Then, define    the probability $\wp_{\vec{\alpha}|M^G}\left(\set{T}_{W,c}(\vec{\alpha})\right)$, where $\set{T}_{W,c}(\vec{\alpha})$ is the tail region around $\vec{\alpha}$ defined as
\begin{align*}
\set{T}_{W,c}(\vec{\alpha}):=\left\{\vec{x}~|~\,(\vec{x}^T-\vec{\alpha}^T)W(\vec{x}-\vec{\alpha})\ge c\right\}.
\end{align*}
 Then, for covariant POVMs, the tail probability is independent of $\vec{\alpha}$ and is given by:
\begin{align*}	
  \wp_{\vec{\alpha}|M^G}\left(\set{T}_{W,c}(\vec{\alpha})\right)=\Tr N[\vec{0},\Gamma]\otimes\Phi[\vec{0},\vec{\beta}] M\left(\set{T}_{W,c}(\vec{0})\right).
\end{align*}
When the measurement is covariant, we have the following bound on the tail probability, which can be attained by a certain covariant POVM:

\begin{lem}\Label{L2B}
Consider a Gaussian model $G[\vec{\alpha},\vec{\gamma}]=N[\vec{\alpha}^C,\Gamma]\otimes \Phi[\vec{\alpha}^Q,\vec{\beta}]$  with $s'$ classical parameters and 
$s$ pairs of quantum parameters. 
Assume that a POVM $\{M^G(\set{B}) \}_{\set{B} \subset \mathbb{R}^{s'} \times \mathbb{R}^{2s}}$ is covariant
and the weight matrix $W$ has the following form;
\begin{align}\Label{W-diag}
W=\left(\begin{array}{cccc} W^C &  & & \\ & w_{s'+1}I_2 & & \\ & & \ddots  &\\ & & & w_{s'+s}I_2\end{array}\right)
\end{align}
with $W^C\ge 0$. 
Then,  the tail probability of the limiting distribution is bounded as
\begin{align}
  \wp_{\vec{\alpha}|M^G}\left(\set{T}_{W,c}(\vec{\alpha})\right)\ge
N\left[\vec{0},\Gamma \otimes E_s\left(e^{-\vec{\beta}}+\vec{e}/2 \right)\right]\left(\set{T}_{W,c}(\vec{0})\right),
\Label{ER5B}
\end{align}
where $\vec{e}$ is the $2s$-dimensional vector with  all entries equal to 1.
For the definition of $E_s\left(e^{-\vec{\beta}}+\vec{e}/2 \right)$, see \eqref{M1}.
When the POVM $M^G$ is given as $M^G(B)=\int_{B}| \alpha_1, \ldots, \alpha_s\rangle \langle \alpha_1, \ldots, \alpha_s| d \vec{\alpha}$, 
the equality in \eqref{ER5B} holds.
\end{lem}

The proof can be found in Appendix \ref{app-Gaussian}.
When the model has a group covariance,
similar evaluation might be possible.
For example, 
similar evaluation was done 
in the $n$-copy of full pure states family \cite{H98}
and
in the $n$-copy of squeezed states family \cite[Section 4.1.3]{hayashi2017group}.

\subsection{Tail property of D-invariant qudit models}
For a $k$-parameter D-invariant model $\{\rho_{\vec{t}}\}$,
using Lemma \ref{L2B} and Q-LAN,
we have the following theorem.
\begin{theo}\Label{thm-tail}
Let $\{\rho_{\vec{t}}\}_{\vec{t}\in\set{\Theta}}$ be a $k$-parameter D-invariant model. 
Assume that 
$(\tilde{J}_{\vec{t}'_0})^{-1}$ exists,
$\rho_{\vec{t}_0}$ is a non-degenerate state,
and a sequence of measurements $\seq{m}:=\{M_n\}$ satisfies local asymptotic covariance at $\vec{t}_{0} \in \set{\Theta}$.
When $J_{\vec{t}_0}^{-1/2}W J_{\vec{t}_0}^{-1/2}$ commutes with
$J_{\vec{t}_0}^{-1/2}D_{\vec{t}_0} J_{\vec{t}_0}^{-1/2}$,
we have
\begin{align}
 \wp_{\vec{t}_0,\vec{t}|\seq{m}}\left(\set{T}_{W,c}(\vec{t})\right)\ge 
N\Big[0, W^{1/2}J_{\vec{t}_0}^{-1}W^{1/2}
+\frac12\left|W^{1/2}J_{\vec{t}_0}^{-1}D_{\vec{t}_0}J_{\vec{t}_0}^{-1}
W^{1/2}\right|\Big]\left(
\set{T}_c\right)\Label{RLDV}
\end{align}
for $\set{T}_c:=\{\vec{x}\in\R^k~|~\|\vec{x}\| \ge c\}$. 
The equality holds if and only if
$\wp_{\vec{t}_0,\vec{t}| \seq{m}}$ is the normal distribution with average zero
and covariance $V_{\vec{t_0}|W}$ as defined in Eq. (\ref{V-tw}).
\end{theo}
We note that bounds on the probability distributions are usually more difficult to obtain and more informative than the MSE bounds, as the MSE can be determined by the probability distribution.
Theorem \ref{thm-tail} provides an attainable bound of the tail probability, which  can be used to determine the maximal probability that the estimate falls into a confidence region $\set{T}_{W,c}$ as well as the optimal measurement. 

Our proof of Theorem \ref{thm-tail} needs some preparations.
First, we introduce the concept of 
simultaneous diagonalization in the sense of symmetric transformation.
Two $2k \times 2k$ real symmetric matrices $A_1$ and $A_2$ are called
{\it simultaneously symplectic diagonalizable}
when there exist a symplectic matrix $S$ 
and two real vectors $\vec{\beta}_1$ and $\vec{\beta}_2$ such that
such that
\begin{align}
S^T A_1 S= E_k(e^{-\vec{\beta}_1}),\quad
S^T A_2 S= E_k(e^{-\vec{\beta}_2})\Label{SYAA2}
\end{align}
with $E_k$ defined in Eq. (\ref{M1}).
Regarding the simultaneous diagonalization, we have the following property, whose proof can be found in Appendix \ref{app-diag}:
\begin{lem}\Label{L4}
The following conditions are equivalent for two $2k \times 2k$ real symmetric matrices $A_1$ and $A_2$.
\begin{description}
\item[(i)]
$A_1$ and $A_2^{-1}$ are
simultaneously symplectic diagonalizable.

\item[(ii)]
$\Omega_k A_2 A_1=A_1 A_2 \Omega_k$, where $\Omega_k$ is defined in Eq. (\ref{M2}).
\end{description}
\end{lem}


Using Lemma \ref{L4}, we obtain the following lemma.

\begin{lem}\Label{L5}
Let $A_1$ be $
\left|J_{\vec{t}_0}^{-1}D_{\vec{t}_0}J_{\vec{t}_0}^{-1}\right|^{1/2}$.
Assume that 
$A_1^{-1}J_{\vec{t}_0}^{-1}D_{\vec{t}_0}J_{\vec{t}_0}^{-1}
A_1^{-1}=\Omega_k$.
When
$J_{\vec{t}_0}^{-1/2}W J_{\vec{t}_0}^{-1/2}$ commutes with
$J_{\vec{t}_0}^{-1/2}D_{\vec{t}_0} J_{\vec{t}_0}^{-1/2}$,
$( A_1 W A_1)^{-1}$ and 
$ A_1^{-1} J_{\vec{t}_0}^{-1} A_1^{-1}$
are simultaneously symplectic diagonalizable.
\end{lem}

\begin{proof}
From $[J_{\vec{t}_0}^{-1/2}W J_{\vec{t}_0}^{-1/2},J_{\vec{t}_0}^{-1/2}D_{\vec{t}_0} J_{\vec{t}_0}^{-1/2}]=0$, we get
$[J_{\vec{t}_0}^{-1}W, (J_{\vec{t}_0}^{-1}D_{\vec{t}_0}J_{\vec{t}_0}^{-1})]=0$.
Next, noticing that $(J_{\vec{t}_0}^{-1}D_{\vec{t}_0}J_{\vec{t}_0}^{-1})=A_1\Omega_k A_1$, we have
\begin{align*}
&J_{\vec{t}_0}^{-1}W
A_1 \Omega_k A_1
=J_{\vec{t}_0}^{-1}W (J_{\vec{t}_0}^{-1}D_{\vec{t}_0}J_{\vec{t}_0}^{-1}) 
 =(J_{\vec{t}_0}^{-1}D_{\vec{t}_0}J_{\vec{t}_0}^{-1})W J_{\vec{t}_0}^{-1} 
 =A_1 \Omega_k A_1W J_{\vec{t}_0}^{-1}.
\end{align*}
The above equalities show that $J_{\vec{t}_0}^{-1}W$ commutes with $A_1 \Omega_k A_1$, which further implies that
\begin{align*}
&A_1^{-1}J_{\vec{t}_0}^{-1} A_1^{-1} A_1WA_1 \Omega_k 
=\Omega_k A_1 W A_1 A_1^{-1}J_{\vec{t}_0}^{-1} A_1^{-1}.
\end{align*}
Using Lemma \ref{L4}, we get
$(A_1 W A_1 )^{-1}$
and $A_1^{-1}J_{\vec{t}_0}^{-1} A_1^{-1}$
are simultaneously symplectic diagonalizable.\qed
\end{proof}

\begin{proofof}{Theorem \ref{thm-tail}}
Define 
$A_1:=
\left|J_{\vec{t}_0}^{-1}D_{\vec{t}_0}J_{\vec{t}_0}^{-1}\right|^{1/2}$
and
$A_2:= 
A_1^{-1}J_{\vec{t}_0}^{-1}D_{\vec{t}_0}J_{\vec{t}_0}^{-1}
A_1^{-1}$.
Applying a suitable orthogonal matrix, we assume that $A_2=\Omega_k$ without loss of generality.

\noindent{\bf Step 1:}
For simplicity, we assume that there is no classical part.
First, we choose an orthogonal matrix $S'$ such that
${S'}^TA_2S'=D$. 
Using Lemma \ref{L5} guarantees that
the condition $[J_{\vec{t}_0}^{-1/2}W J_{\vec{t}_0}^{-1/2},J_{\vec{t}_0}^{-1/2}D_{\vec{t}_0} J_{\vec{t}_0}^{-1/2}]=0$ allows us to simultaneously diagonalize $W$ and $D_{\vec{t}_0}$.
That is, we can choose symplectic matrix $S$ such that
$S^{T} {S'}^T ( A_1 W A_1)^{-1}{S'} S$ and 
$S^{T} {S'}^T A_1^{-1} J_{\vec{t}_0}^{-1} A_1^{-1}{S'} S$
are diagonal matrixces
$ E_{k}(\vec{w})^{-1}$
and
$ E_{k}(\vec{\beta})$ for $k\in\N^*$.
We introduce the local parameter
${\vec{t}'}:=
S^{T} {S'}^T A_1^{-1} \vec{t}$.
Then, $ \vec{t}\cdot W \vec{t}=
\vec{t}'\cdot E_{k}(\vec{w}) \vec{t}'$.

For a sequence of measurement 
$\seq{m}:=\{M_n\}$ to satisfy local asymptotic covariance at $\vec{t}_{0} \in \set{\Theta}$,
according to 
Theorem \ref{CTh3},
we choose a covariant POVM $\widetilde{M}^G$ to satisfy \eqref{Gaussian-limiting}.
Applying Lemma \ref{L2B} to the POVM $\widetilde{M}^G$,
we obtain the desired statement.

\noindent{\bf Step 2:}
We consider the general case.
Now, we choose the local parameter 
$\vec{t}':= J_{\vec{t}_0}^{-1/2} \vec{t}$.
In this coordinate, 
The inverse of the RLD quantum Fisher information is 
 $I+ J_{\vec{t}_0}^{-1/2}D_{\vec{t}_0} J_{\vec{t}_0}^{-1/2}$.
 Since
$J_{\vec{t}_0}^{-1/2}D_{\vec{t}_0} J_{\vec{t}_0}^{-1/2}$ commutes with $J_{\vec{t}_0}^{-1/2}W J_{\vec{t}_0}^{-1/2}$,
the weight matrix has no cross term between the classical and quantum parts.
Using the above discussion and  Lemma \ref{L2B},
 we obtain the desired statement.
\end{proofof}

\section{Extension to global estimation and generic cost functions}\Label{sec-global}
In the previous sections, we focused on local models and cost functions of the form $\tr WV[\wp_{\vec{t}_0,\vec{t}|\seq{m}}]$. In this section, our treatment  will be extended to  global models 
$\{\rho_{\vec{t}}\}_{\vec{t} \in \set{\Theta}}$.
(where the parameter to be estimated is not restricted to a local neighborhood) and to generic cost functions. 

\subsection{Optimal global estimation via local estimation}

Our optimal global estimation is given by combining the two-step method and 
local optimal estimation.
That is, the first step is the application of
full tomography proposed in  \cite{haah2017sample} on $n^{1-x/2}$ copies
with the outcome $\hat{\vec{t}}_0$ for a constant $x\in(0,2/9)$, and
the second step is the local optimal estimation at $\hat{\vec{t}}_0$, given in Section \ref{s63}, on $$a_{n,x}:=n-n^{1-x/2}$$ copies.
Before its full description, we define the
neighborhood $\set{\Theta}_{n,x}(\vec{t})$ of $\vec{t}\in\set{\Theta}$ as
\begin{align}
\Label{neighborhood-def}
\set{\Theta}_{n,x}(\vec{t})&:=\left\{\vec{y}~|~\|\vec{y}-\vec{t}\|\le n^{-\frac{1-x}2}\right\} .
\end{align}

Given a generic model ${\cal M}=\{\rho_{\vec{t}}\}_{\vec{t} \in \set{\Theta}}$ that does not contain
any degenerate state and a weight matrix $W> 0$,
we describe the full protocol as follows.

\begin{description}
\item[(A1)] {\em Localization:}
Perform full tomography proposed in  \cite{haah2017sample} on $n^{1-x/2}$ copies, which is described by a POVM $\{M^{\rm tomo}_{n^{1-x/2}}\}$, for a constant $x\in(0,2/9)$. 
The tomography outputs the first estimate $\hat{\vec{t}}_0$ so that
\begin{align}
&
\Tr \rho_{\vec{t}}^{\otimes n^{1-x/2}}
M^{\rm tomo}_{n^{1-x/2}}
\left(\set{\Theta}_{n,x}(\vec{t}_{\rm g})\right)=1-O\left(e^{-n^{x/2}}\right)\Label{tomo-faithful}
\end{align}
for any true parameter $\vec{t}$.

\item[(A2)] {\em Local estimation:}
Based on the first estimate $\hat{\vec{t}}_0$, 
apply the optimal local measurement $M^{\hat{\vec{t}}_0,a_{n,x}}_{W}$ given in Theorem \ref{prop-fixed}
with the weight matrix $W$. 
If the measurement outcome $\hat{\vec{t}}_1$ of $M^{\hat{\vec{t}}_0,a_{n,x}}_{W}$
is in $\set{\Theta}_{n,x}(\hat{\vec{t}}_0)$, 
output the outcome $\hat{\vec{t}}_1$ as the final estimate; 
otherwise output $\hat{\vec{t}}_0$ as the final estimate.
\end{description}

Denoting the POVM of the whole process by 
$\seq{m}_{W}=\{M^n_{W}\}$,
we obtain the following theorem.

\begin{theo}\Label{THD}
Assume that
a qudit-model ${\cal M}=\{\rho_{\vec{t}}\}_{\vec{t} \in \set{\Theta}}$ does not contain
any degenerate state, 
the parametrization is $C^2$ continuous, 
$\set{\Theta}$ is an open set, and 
${J}_{\vec{t}_0}^{-1} $ exists.
(i) The sequence $\seq{m}_W$ satisfies local asymptotic covariance at any point $\vec{t}_0$ in the parameter space. 
(ii) 
The equation
\begin{align}\Label{global-bound}
\tr W V [\wp_{\vec{t}_{0},\vec{t} |\seq{m}_{W}}]=
\map{C}(W,\vec{t}_{0})
\end{align}
holds for any point $\vec{t}_0 \in \set{\Theta}$ and any $\vec{t}\in\Theta_{n,x,c(\vec{t}_0)}$ corresponding to a non-degenerate state, 
where $\map{C}_{\spc{S}}(W,\vec{t}_{0})$  is  the minimum weighted MSE as defined in Eq. (\ref{global-boundB}).
More precisely, we have
\begin{align}
\limsup_{n \to \infty}\sup_{\vec{t}_0 \in K}
\sup_{\vec{t} \in \Theta_{n,x,c(\vec{t}_0)} } n^{\kappa}
\left\| \wp^n_{\vec{t}_0,\vec{t}| M^n_{W}}
- N[\vec{t}, V_{\vec{t_0}|W}]
\right\|_1 <\infty\Label{DJI4}
\end{align}
for a compact set $K\subset \set{\Theta}$, where $V_{\vec{t_0}|W}$ is defined in Eq. (\ref{V-condition-W}) and $\Theta_{n,x,c(\vec{t}_0)}$ is defined in Eq. (\ref{neighborhood-c}).
Further, when the parameter set $\set{\Theta}$ is bounded and
$x<\kappa$,
we have the following relation.
\begin{align}
\lim_{n \to \infty}
\sup_{\vec{t}_0 \in K}
\sup_{\vec{t} \in \set{\Theta}_{n,x,c(\vec{t}_0)} } 
\left\| 
V [\wp^n_{\vec{t}_0,\vec{t}| M^n_{W}}]
-  V_{\vec{t_0}|W} \right\|_1=0.\Label{DJI5}
\end{align} 
\end{theo}

Here, we should remark the key point of the derivation.
The existing papers \cite{gill2000state,hayashi2005statistical} addressed the achievability 
of $\min_{M}\tr W J_{\vec{t}|M}^{-1}$ with the two-step method, where 
$J_{\vec{t}|M}$ is the Fisher information matrix of the 
distribution family $\{\wp_{\vec{t}|M}\}_{\vec{t}}$,
which expresses the bounds among separable measurement \cite[Exercise 6.42]{hayashi2017quantum}.
Hence it can be called the separable bound.
In the one-parameter case,
the separable bound equals the Holevo bound.
To achieve the separable bound, 
we do not consider the sequence of measurement.
Hence, we do not handle a complicated convergence.
The global achievability of the separable bound can be easily shown by the two-step method \cite{gill2000state,hayashi2005statistical}.
However, in our setting, 
we need to handle the sequence of measurement to achieve the local optimality.
Hence, we need to carefully consider the compact uniformity and the order estimate of 
the convergence in Theorem \ref{prop-fixed}. 
In the following proof,   
we employ our evaluation with such detailed analysis as in Eq. \eqref{DJI2}.  


\noindent\emph{Proof.}
\noindent{\bf Step 1:}  Define by $\vec{t}_{\rm g}:=\vec{t}_0+\frac{\vec{t}}{\sqrt{n}}$ the true value of the parameters.
By definition, we have $\|\vec{t}_{\rm g}- \hat{\vec{t}}_0\|\le n ^{-\frac{1-x}{2}}  $ with probability $1-O(e^{-n^{x/2}})$
and $\|\vec{t}_{\rm g}- \vec{t}_0\|\le c(\vec{t}_0) n ^{-\frac{1}{2}+x}  $ by definition. Since the error probability vanishes exponentially, it would not affect the scaling of MSE.
In this step,
we will show
\begin{align}
&\left\|\wp^{n}_{\vec{t}_0,\vec{t}| M^{\hat{\vec{t}}_0,a_{n,x}}_{W}}
-N[\vec{t}, V_{\vec{t}_0|W}]\right\|_1  =O(n^{-\kappa}). \Label{ATG}
\end{align}

Since $\|\vec{t}_{\rm g}- \hat{\vec{t}}_0\|\le n ^{-\frac{1-x}{2}}  $
and $\|\vec{t}_{\rm g}- \vec{t}_0\|\le c(\vec{t}_0) n ^{-\frac{1}{2}+x}  $ imply
$\|\vec{t}_0- \hat{\vec{t}}_0\|\le 2 c(\vec{t}_0) n ^{-\frac{1}{2}+x}  $, we have
\begin{align}
\| N[\vec{0}, V_{\vec{t}_0|W}] - N[\vec{0}, V_{\hat{\vec{t}}_0|W}]\|_1
&= O( n ^{-\frac{1}{2}+x}  ) .\Label{FRF}
\end{align} 
Eq. \eqref{DJI2} of Theorem \ref{prop-fixed} implies 
\begin{align}
\Big\|\wp^{a_{n,x}}_{\vec{t}_0,\vec{t}| M^{\hat{\vec{t}}_0,a_{n,x}}_{W}}
-N[\vec{0}, V_{\hat{\vec{t}}_0|W}]\Big\|_1&=O(n^{-\kappa}).
\end{align}
Since 
$\wp^{n}_{\vec{t}_0,\vec{t}| M^{\hat{\vec{t}}_0,a_{n,x}}_{W}}(\set{B})
=\wp^{a_{n,x}}_{\vec{t}_0,\vec{t}| M^{\hat{\vec{t}}_0,a_{n,x}}_{W}}\left(\frac{\sqrt{a_{n,x}}\set{B}}{\sqrt{n}}\right)$,
\begin{align}
&\Big\|\wp^{n}_{\vec{t}_0,\vec{t}| M^{\hat{\vec{t}}_0,a_{n,x}}_{W}}
-N[\vec{t}, \frac{n}{a_{n,x}}V_{\vec{t}_0|W}]\Big\|_1 
=
 \Big\|\wp^{a_{n,x}}_{\vec{t}_0,\vec{t}| M^{\hat{\vec{t}}_0,a_{n,x}}_{W}}
-N[\vec{t}, V_{\vec{t}_0|W}]\Big\|_1 \nonumber \\
\le &
\Big\|\wp^{a_{n,x}}_{\vec{t}_0,\vec{t}| M^{\hat{\vec{t}}_0,a_{n,x}}_{W}}
-N[\vec{0}, V_{\hat{\vec{t}}_0|W}]\Big\|_1
+\|N[\vec{t}, V_{\vec{t}_0|W}] - N[\vec{t}, V_{\hat{\vec{t}}_0|W}]\|_1 \nonumber \\
=& O(n^{-\kappa})+O( n ^{-\frac{1}{2}+x}  ) =O(n^{-\kappa}).
\Label{FJU}
\end{align}

As we have
\begin{align}
\| N[\vec{0}, V_{\vec{t}_0|W}]-N[\vec{0}, \frac{n}{a_{n,x}}V_{\vec{t}_0|W}]\|_1 
&=O(\frac{n}{a_{n,x}}-1) \nonumber \\
&=O((1-n^{x/2})^{-1}-1)=O(n^{x/2}),
\end{align}
we obtain
\begin{align}
&\|\wp^{n}_{\vec{t}_0,\vec{t}| M^{\hat{\vec{t}}_0,a_{n,x}}_{W}}
-N[\vec{t}, V_{\vec{t}_0|W}]\|_1\nonumber \\
\le
&\| N[\vec{t}, V_{\vec{t}_0|W}]-N[\vec{t}, \frac{n}{a_{n,x}}V_{\vec{t}_0|W}]\|_1 
+\|\wp^{n}_{\vec{t}_0,\vec{t}| M^{\hat{\vec{t}}_0,a_{n,x}}_{W}}
-N[\vec{t}, \frac{n}{a_{n,x}}V_{\vec{t}_0|W}]\|_1\nonumber \\
=&O(n^{x/2})+O(n^{-\kappa}).
\end{align}

\noindent{\bf Step 2:}
We will show \eqref{DJI4}.
First, we discuss two exceptional cases 
$\|\vec{t}_{\rm g}- \hat{\vec{t}}_0\| > n ^{-\frac{1-x}{2}}  $
and $\|\hat{\vec{t}}_1- \hat{\vec{t}}_0\|> n ^{-\frac{1-x}{2}}  $.
Eq. \eqref{tomo-faithful} guarantees that
\begin{align}
\Tr \rho_{\vec{t}_{\rm g}}^{\otimes n^{1-x/2}}
M^{\rm tomo}_{n^{1-x/2}}
\left( \Big\{\hat{\vec{t}}_0  \Big| \|\vec{t}_{\rm g}- \hat{\vec{t}}_0\| > n ^{-\frac{1-x}{2}} \Big\}
\right)
=O\left(e^{-n^{x/2}}\right).
\Label{tomo-faithful2}
\end{align}
Eq. \eqref{FJU} and the property of normal distribution implies
\begin{align}
& \Tr \rho_{\vec{t}_{\rm g}}^{\otimes a_{n,x}}
M^{\hat{\vec{t}}_0,a_{n,x}}_{W}
\left( \Big\{\hat{\vec{t}}_0  \Big| 
\|\hat{\vec{t}}_1- \hat{\vec{t}}_0\|> n ^{-\frac{1-x}{2}} \Big\}
\right) \nonumber \\
=&
O(n^{-\kappa})+
N[\vec{t}, \frac{n}{a_{n,x}}V_{\vec{t}_0,|W}]
( \{\vec{t}| ~ \|\vec{t} \|> n^{x/2} \}) \nonumber \\
=&
O(n^{-\kappa})+
O(e^{- O(n^{x/2}) })
=O(n^{-\kappa}).
\Label{TI2}
\end{align}
When 
$\|\vec{t}_{\rm g}- \hat{\vec{t}}_0\| \le n ^{-\frac{1-x}{2}}  $
and $\|\hat{\vec{t}}_1- \hat{\vec{t}}_0\| \le n ^{-\frac{1-x}{2}}  $,
Eq. \eqref{ATG} holds under the condition $\|\vec{t}_{\rm g}- \vec{t}_0\|\le c(\vec{t}_0) n ^{-\frac{1}{2}+x}  $,
which implies that
\begin{align}
\sup_{\vec{t} \in \set{\Theta}_{n,x,c(\vec{t}_0)} } 
\| \wp^n_{\vec{t}_0,\vec{t}| M^n_{W}}
- N[\vec{t}, V_{\vec{t_0}|W}]
\|_1 =O(n^{-\kappa}).\Label{DJI5B}
\end{align}
Since the above evaluation is compactly uniform with respect to $\vec{t}_0$,
we have \eqref{DJI4}.

\noindent{\bf Step 3:}
We will show 
\begin{align}
\lim_{n \to \infty}
\sup_{\vec{t}_0 \in K}
\sup_{\vec{t} \in \set{\Theta}_{n,x,c(\vec{t}_0)} } 
\Big\| 
V [\wp^n_{\vec{t}_0,\vec{t}| M^n_{W}}]
-  \frac{n}{a_{n,x}} V_{\vec{t_0}|W} \Big\|_1=0\Label{DJI6}
\end{align}
because Eq. \eqref{DJI6} implies \eqref{DJI5} due to the convergence $\frac{n}{a_{n,x}} \to 1$.
There are three cases.
(1) $\|\vec{t}_{\rm g}- \hat{\vec{t}}_0\| > n ^{-\frac{1-x}{2}}  $,
(2) $\|\hat{\vec{t}}_1- \hat{\vec{t}}_0\|> n ^{-\frac{1-x}{2}}  $ and
$\|\vec{t}_{\rm g}- \hat{\vec{t}}_0\| \le n ^{-\frac{1-x}{2}}  $,
and (3) the remaining case.

The compactness of $\set{\Theta}$ guarantees that 
the error $ n (\hat{\vec{t}} - \vec{t}_{\rm g})(\hat{\vec{t}} - \vec{t}_{\rm g})^T $
is bounded by $n C$ with a constant $C$.
Due to \eqref{tomo-faithful2}, the contribution of the first case is bounded by $
nC \cdot O(e^{-n^{x/2}})$, which goes to zero.

In the second case, since 
$\hat{\vec{t}}_0=\hat{\vec{t}}$,
the error $ n (\hat{\vec{t}} - \vec{t}_{\rm g})(\hat{\vec{t}} - \vec{t}_{\rm g})^T $
is bounded by $n (n ^{-\frac{1-x}{2}})^2=n^x $. 
Due to \eqref{TI2}, the contribution of the second case is bounded by 
$n^x \cdot O(n^{-\kappa})=O(n^{x-\kappa})$, which goes to zero.

In the third case, 
since 
$\| {\vec{t}}_{\rm g}-\hat{\vec{t}}_1\| \le
\| {\vec{t}}_{\rm g}-\hat{\vec{t}}_0\| +\| \hat{\vec{t}}_0-\hat{\vec{t}}_1\| \le
2 n ^{-\frac{1-x}{2}}$,
the error $ n (\hat{\vec{t}} - \vec{t}_{\rm g})(\hat{\vec{t}} - \vec{t}_{\rm g})^T $
is bounded by $ 2n (n^{-\frac{1-x}{2}})^2=2 n^x $. 
Due to \eqref{FJU}, the contribution of the second case is bounded by 
$2 n^x \cdot O(n^{-\kappa})=O(n^{x-\kappa})$, which goes to zero.
Therefore, we obtain \eqref{DJI6}.
\qed

\medskip

\subsection{Generic cost functions}
Finally, we show that results in this work hold also for any cost function $c(\hat{\vec{t}},\vec{t})$, which is bounded and has a symmetric expansion, in the sense of satisfying the following two conditions:
\begin{description}
\item[(B1)]  $c(\hat{\vec{t}},\vec{t})$ has a continuous third derivative, so that it can be expanded as  $c(\hat{\vec{t}},\vec{t})=
(\hat{\vec{t}}^T-\vec{t}^T)W_{\vec{t}}(\hat{\vec{t}}-\vec{t})+O(\|\hat{\vec{t}}-\vec{t}\|^3)$ 
as $\hat{\vec{t}}$ is close to $\vec{t} $,
where the matrix $W_{\vec{t}}\ge 0$ is a continuous function of $\vec{t}$.
\item[(B2)] $c(\hat{\vec{t}},\vec{t})\le C$ for a constant $C>0$ and for any $\hat{\vec{t}},\vec{t}\in\R^k$.
\end{description}
To adopt this situation, 
we replace the step (A2) by the following step (A2)':
\begin{description}
\item[(A2)'] 
Based on the first estimate $\hat{\vec{t}}_0$, 
apply the optimal local measurement $M^{\hat{\vec{t}}_0,a_{n,x}}_{W_{\hat{\vec{t}}_0}}$ 
given in Theorem \ref{prop-fixed}
with the weight matrix $W_{\hat{\vec{t}}_0}$. 
If the measurement outcome $\hat{\vec{t}}_1$ of $M^{\hat{\vec{t}}_0,a_{n,x}}_{W_{\hat{\vec{t}}_0}}$
is in $\set{\Theta}_{n,x}(\hat{\vec{t}}_0)$, 
output the outcome $\hat{\vec{t}}_1$ as the final estimate; 
otherwise output $\hat{\vec{t}}_0$ as the final estimate $\hat{\vec{t}}$.
\end{description}

Denoting the POVM of the whole process by 
$\seq{m}_{c}=\{M^n_c\}$,
we have the following result:

\begin{theo}\Label{theo-cost}
Assume the same assumption for 
the model ${\cal M}$ as Theorem \ref{THD}.
(i) When 
a sequence of measurements 
$\seq{m}:=\{M_{n}(d\hat{\vec{t}})\}$  
satisfies local  asymptotic covariance at $\vec{t}_{0} \in \set{\Theta}$
and a cost function $c$ satisfies condition (B1), 
the inequality
\begin{align}\Label{H7}
\lim_{n\to\infty} n\cdot c_{\vec{t}_{0}+\frac{\vec{t}}{\sqrt{n}}}(\seq{m})
\ge 
\map{C}(W_{\vec{t}_{0}},\vec{t}_{0})
\end{align}
holds, where $\map{C}_{\spc{S}}(W_{\vec{t}_{0}},\vec{t}_{0})$ is defined in Eq. (\ref{global-boundB}) and
\begin{align*}
c_{\vec{t}_{0}+\frac{\vec{t}}{\sqrt{n}}}(\seq{m})
:=\int c\left(\hat{\vec{t}}_n,\vec{t}_{0}+\frac{\vec{t}}{\sqrt{n}}\right)\,\Tr\rho_{\vec{t}_{0}+\frac{\vec{t}}{\sqrt{n}}}^{\otimes n}M_{n}\left(d\hat{\vec{t}}_{n}\right).
\end{align*}
(ii) In addition, if $c$ also satisfies Condition (B2), 
$\seq{m}_{c}=\{M^n_c\}$ is locally asymptotically covariant and attains the equality in (\ref{H7}) at any point $\vec{t}_{0} \in \set\Theta$ corresponding to a non-degenerate state.
\end{theo}
Theorem \ref{theo-cost} is reduced to a bound for the (actual) MSE when $c(\hat{\vec{t}},\vec{t})=(\hat{\vec{t}}^T-\vec{t}^T)W(\hat{\vec{t}}-\vec{t})$ for $W\ge 0$.
Therefore, bounds in this work, Eqs. (\ref{HB}) and (\ref{NHB}) for instance, are also attainable bounds for the MSE of any locally asymptotically unbiased measurement.

\medskip
\noindent\emph{Proof.}
\noindent{\bf Step 1:}
We prove (1).
Consider any sequence of asymptotically covariant measurements $\seq{m}_{\vec{t}_0}:=\{M_{n,\vec{t}_0}\}$ at $\vec{t}_0$. Denote by $\vec{t}_{\rm g}:=\vec{t}_0+\frac{\vec{t}}{\sqrt{n}}$ the true value of the parameters.
For a cost function $c$ satisfying (ii), we have
\begin{align*}
 &\lim_{n\to\infty} n\cdot c_{\vec{t}_{0}+\frac{\vec{t}}{\sqrt{n}}}(\seq{m})\\
 =&\lim_{n\to\infty}n\cdot\int c(\hat{\vec{t}}_n,\vec{t}_{\rm g})\,\Tr\rho_{\vec{t}_{\rm g}}^{\otimes n}M_{n,\vec{t}_{0}}\left(d\hat{\vec{t}}_{n}\right)\\
=&\lim_{n\to\infty}\int \left(\vec{t}'^TW_{\vec{t}_{\rm g}}\vec{t'}+\frac{1}{\sqrt{n}}O(\|\vec{t'}\|^3)\right)\ \wp^n_{\vec{t}_0,\vec{t}|M_{n,\vec{t}_{0}}}\left(d\vec{t}'\right) \quad
\big(\vec{t}':=\sqrt{n}(\hat{\vec{t}}_n-\vec{t}_{\rm g}) \big)\\
=&\int \vec{t}'^TW_{\vec{t}_{0}}\vec{t'}\ \wp_{\vec{t}_{0},\vec{t}|\seq{m}_{\vec{t}_0}}(d\vec{t}')
 \ge \ \map{C}(W_{\vec{t}_{0}},\vec{t}_{0}).
\end{align*}

\noindent{\bf Step 2:}
We prove (2).
We replace $W $ by $W_{\vec{t}}$ in the proof of Theorem \ref{THD}.
In this replacement, \eqref{FRF} is replaced by 
\begin{align}
\| N[\vec{0}, V_{\vec{t}_0|W_{{\vec{t}_0}}}] - N[\vec{0}, V_{\hat{\vec{t}}_0|W_{\hat{\vec{t}}_0}}]\|_1
&= O( n ^{-\frac{1}{2}+x}  ) \Label{FRF2}
\end{align}
where $x\in(0,2/9)$.
Hence, the contributions of  the first and second cases of Step 3 of the proof of Theorem \ref{THD}
go to zero.

In the third case of Step 3 of the proof,
we have $\| {\vec{t}}_{\rm g}-\hat{\vec{t}}_1\| \le 2 n ^{-\frac{1-x}{2}}$,
Hence, 
\begin{align}
n c(\hat{\vec{t}}_1,\vec{t}_{\rm g})-
n (\hat{\vec{t}}_1^T-\vec{t}_{\rm g}^T)W_{\vec{t}_{\rm g}}(\hat{\vec{t}}_1-\vec{t}_{\rm g})
=n O(\|\hat{\vec{t}}_1-\vec{t}_{\rm g}\|^3)
=O (n ^{-\frac{1-3x}{2}}) \to 0.
\end{align} 
Hence, in the contribution of the third case, 
we can replace the expectation of $n c(\hat{\vec{t}}_1,\vec{t}_{\rm g})$
by the weighted MSE with weight $W_{\vec{t}_{\rm g}}$.
Hence, we obtain the part (2).
\qed

\medskip

\section{Applications.}\Label{sec:applications}
In this section, we show how to evaluate the MSE bounds in several concrete examples.

\subsection{Joint measurement of observables}\Label{subsec-ob}

Here we consider the fundamental problem  of the joint measurement of two observables. For simplicity we choose to analyze qubit systems, although the approach can be readily generalized to arbitrary dimension. The task  is to simultaneously estimate the expectation of two observables $A$ and $B$ in a qubit system. The observables can be  expressed as
$A=\vec{a}\cdot\vec{\sigma}$ and $B=\vec{b}\cdot\vec{\sigma}$ with $\vec{\sigma}=(\sigma_x,\sigma_y,\sigma_z)$ being the vector of Pauli matrices. We assume without loss of generality that $|\vec{a}|=|\vec{b}|=1$ and $\vec{a}\cdot\vec{b}\in[0,1)$. The state of an arbitrary qubit system can be expressed as
\begin{align*}
\rho:=\frac12\left(I+\vec{n}\cdot\vec{\sigma}\right),
\end{align*}
where $\vec{n}$ is the Bloch vector.

With this notation, the task is reduced to estimate the parameters
\begin{align*}
x:=\vec{a}\cdot\vec{n},\qquad y:=\vec{b}\cdot\vec{n}.
\end{align*}
  It is also convenient to introduce a third unit vector $\vec{c}$ orthogonal to $\vec{a}$ and $\vec{b}$ so that $\{\vec{a},\vec{b},\vec{c}\}$ form a (non-orthogonal) normalized basis  of $\R^3$.   In terms of this vector, we can  define  the parameter $z:=\vec{c}\cdot\vec{n}$.  
In this way, we extend the problem to the full model containing all qubit states, where $x$, $y$ are the parameters of interest and $z$ is a nuisance parameter. Under this parameterization, we can evaluate the SLD operators for $x,y,$ and $z$, as well as the SLD Fisher information matrix and the D matrix (see Appendix \ref{app-example1} for details), substituting which into the bound (\ref{NHB}) yields:
\begin{align}
 \tr V W\ge \tr W\left(\begin{array}{cc}\frac{1-|\vec{n}|^2+y'^2(1-s^2)+z^2}{1-|\vec{n}|^2+x'^2+2x'y's+y'^2+z^2}&-\frac{x'y'(1-s^2)+(1-|\vec{n}|^2+z^2)s}{1-|\vec{n}|^2+x'^2+2x'y's+y'^2+z^2} \\
\\
-\frac{x'y'(1-s^2)+(1-|\vec{n}|^2+z^2)s}{1-|\vec{n}|^2+x'^2+2x'y's+y'^2+z^2} &\frac{1-|\vec{n}|^2+x'^2(1-s^2)+z^2}{1-|\vec{n}|^2+x'^2+2x'y's+y'^2+z^2} \end{array}\right)\nonumber\\
\qquad+\frac{1}2\tr\left|\sqrt{W}\left(\begin{array}{cc}0&-2z\sqrt{1-s^2}
\\
2z\sqrt{1-s^2}&0\end{array}\right)\sqrt{W}\right|,\Label{bound-example1}
\end{align}
where  $s:=\vec{a}\cdot\vec{b}$, $x'=\frac{x-ys}{1-s^2}$, and $y'=\frac{y-xs}{1-s^2}$.

The tradeoff between the measurement precisions for  the two observables is of fundamental interest. Substituting the expressions of D-matrix and the SLD Fisher information matrix (see Appendix \ref{app-example1}) into Eq. (\ref{disturbance-bound-two}), we obtain
\begin{align*}
 \Delta_A+\Delta_B&\ge2|z|\sqrt{1-s^2},
\end{align*}
which characterizes the precision tradeoff in joint measurements of qubit observables.

\subsection{   Direction estimation in the presence of noise}
Consider the task of estimating a pure qubit state $|\psi\>=\cos\frac\theta2|0\>+e^{i\varphi}\sin\frac\theta2|1\>$, which can also be regarded as determining  a direction in space, as qubits are often realized in spin-$1/2$ systems. In a practical setup, it is necessary to take into account the effect of noise, under which the qubit becomes mixed. For noises with strong symmetry, like depolarization, the usual MSE bound produces a good estimate of the error. For other kind of noises, it is essential to introduce nuisance parameters, and to use the techniques introduced in this paper.

 \begin{figure}[t!]
\begin{center}
  \includegraphics[width=0.6\linewidth]{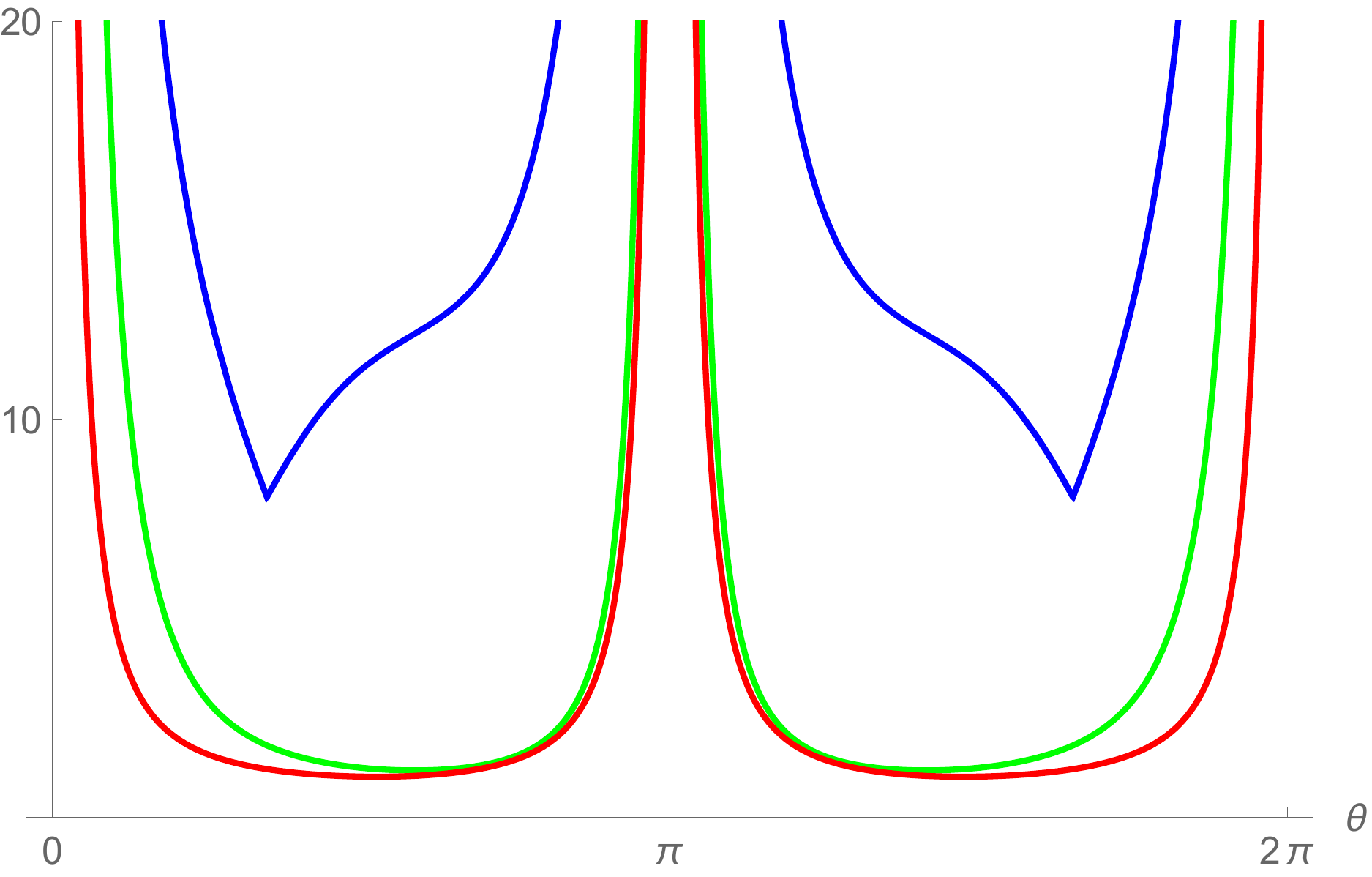}
  \end{center}
\caption{\Label{fig:damping}
  {\bf Amplitude damping model.}  Plot of $\mathbf{MSE}_{\theta}+\mathbf{MSE}_\varphi$ as a function of $\theta$ (the value is the same for all $\varphi$) for $\eta=0.9$ (red), $\eta=0.5$ (green), and $\eta=0.1$ (blue). Here $\mathbf{MSE}_{x}$ denotes the $(x,x)$-th element of the MSE matrix.}
\end{figure}
 
 As an illustration, we consider the amplitude damping noise as an example, which can be formulated as the channel
\begin{align*}
\map{A}_\eta(\rho):=A_0\rho A_0^\dag+A_1\rho A_1^\dag
\end{align*}
where $A_0=|0\>\<0|+\sqrt{\eta}|1\>\<1|$ and $A_1=\sqrt{1-\eta}|0\>\<1|$ are the Kraus operators. After the noisy evolution, the qubit state can be expressed as
\begin{align*}
\rho:=\frac12\left(I+\vec{n}\cdot\vec{\sigma}\right)
\end{align*}
with
$
 \vec{n}:=(\sqrt{\eta}\sin\theta\sin\varphi,\sqrt{\eta}\sin\theta\cos\varphi,1-\eta+\eta\cos\theta).
$ 
Now we can regard $\eta$ as a nuisance parameter, while $\theta$ and $\varphi$ are the parameters of interest.

Defining the derivative vector through the equation $\vec{p}_x\cdot\vec{\sigma}=\partial\rho/\partial x$, we can calculate the vectors 
\begin{align*}
\vec{p}_{\theta}&=(\sqrt{\eta}\cos\theta\sin\varphi,\sqrt{\eta}\cos\theta\cos\varphi,-\eta\sin\theta)\\
\vec{p}_\varphi&=(\sqrt{\eta}\sin\theta\cos\varphi,-\sqrt{\eta}\sin\theta\sin\varphi,0)\\
\vec{p}_\eta&=(\sin\theta\sin\varphi/(2\sqrt{\eta}),\sin\theta\cos\varphi/(2\sqrt{\eta}),-1+\cos\theta)
\end{align*}
In terms of the derivative vector, the SLD for the parameter $x\in  \{\theta,\varphi, \eta\}$ takes the form
\begin{align*}
L_x=-\frac{2\vec{p}_x\cdot\vec{n}}{1-|\vec{n}|^2}I+\left(2\vec{p}_x+\frac{2\vec{p}_x\cdot\vec{n}}{1-|\vec{n}|^2}\vec{n}\right)\cdot\vec{\sigma} \, .
\end{align*}

 After some straightforward calculations, we get 
 \begin{align*}
J=\left(\begin{array}{ccc}4\eta &0 & 2\sin\theta \\
\\
0 &4\eta\sin^2\theta &0\\
\\
2\sin\theta &0& \frac{(1-\eta)[\sin^2\theta+4\eta(1-\cos\theta)^2]+(2\eta-1)^2}{\eta(1-\eta)}\end{array}\right)
\end{align*}
and 
\begin{align*}
&D_{\theta,\varphi}=-D_{\varphi,\theta}=8\eta\sin\theta(\eta-\eta\cos\theta+\cos\theta)\\
&D_{\eta,\varphi}=-D_{\varphi,\eta}=4\sin^2\theta\left(1+\eta-\eta\cos\theta\right)\\
&D_{\theta,\eta}=-D_{\eta,\theta}=0.
\end{align*}
Then we have the MSE bound with nuisance parameter $\eta$. An illustration can be found in Figure \ref{fig:damping} with $W=I$ in Eq. (\ref{NHB}). The minimum of the sum of the $(x,x)$-th matrix element of the MSE matrix for $x=\theta,\varphi$ is independent of $\varphi$, which is a result of the symmetry of the problem:  the D-matrix does not depend on $\varphi$, and thus an estimation of $\varphi$ can be obtained without affecting the precisions of other parameters. Notice that when the state is close to $|0\>$ or $|1\>$, it is insensitive to the change of $\theta$, resulting in the cup-shape curves in Figure \ref{fig:damping}.
 
   \begin{figure}[h!]
\begin{center}
  \includegraphics[width=0.6\linewidth]{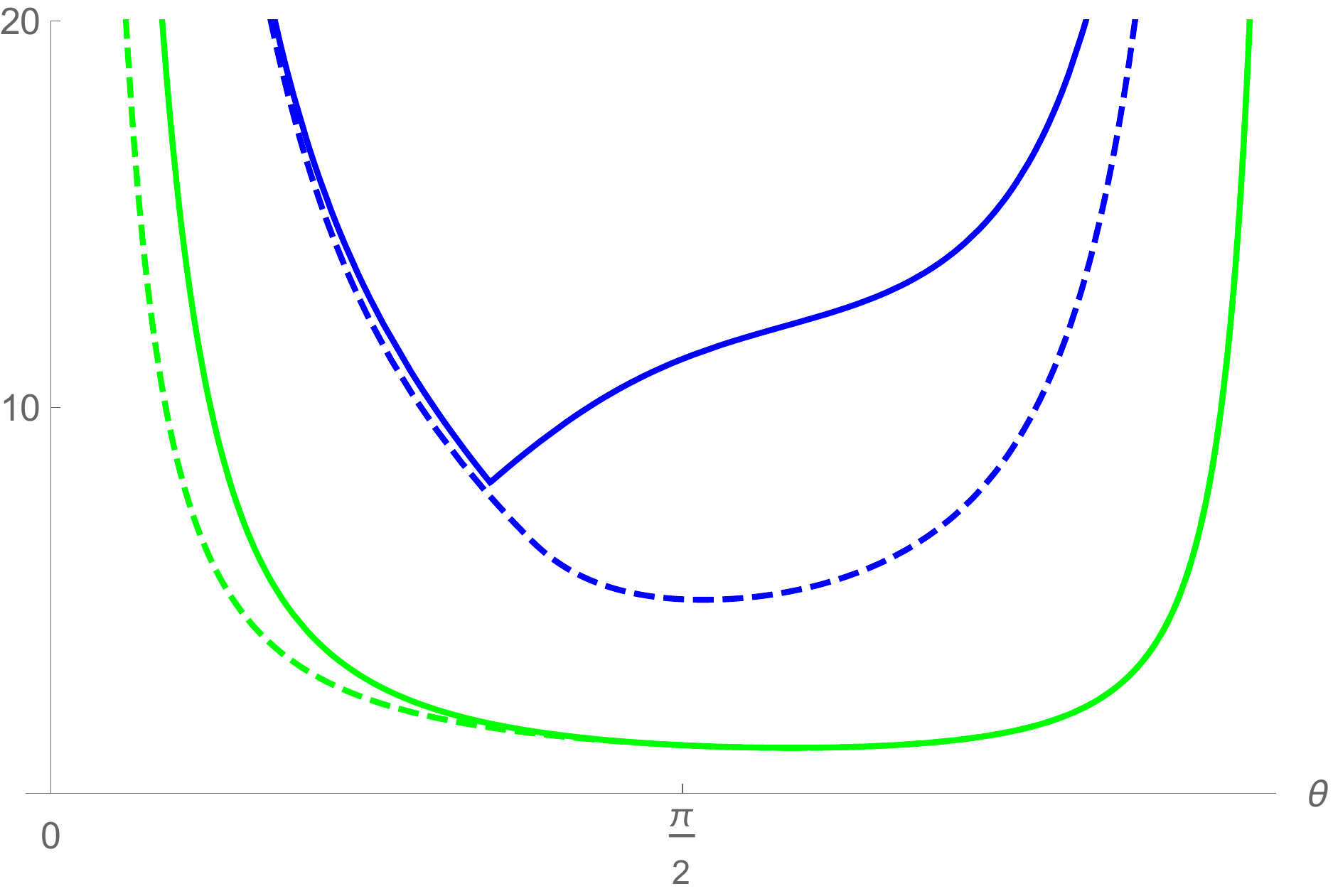}
  \end{center}
\caption{\Label{fig:fixed}
  {\bf  The nuisance parameter bound versus the fixed parameter bound.}  $\mathbf{MSE}_{\theta}+\mathbf{MSE}_\varphi$ as a function of $\theta$ (the value is the same for all $\varphi$) for $\eta=0.5$ (green) and $\eta=0.1$ (blue) are plotted. The solid curves correspond to the case when $\eta$ is a nuisance parameter, while the dashed curves correspond to the case when $\eta$ is a fixed parameter. Here $\mathbf{MSE}_{x}$ denotes the $(x,x)$-th element of the MSE matrix.}
\end{figure}

Next, we evaluate the sum of MSEs of $\varphi$ and $\theta$ when $\eta$ is a (known) fixed parameter using Eq. (\ref{HB}) and compare it to the nuisance parameter case. The result of the numerical evaluation is plotted in Fig. \ref{fig:fixed}.
It is clear from the plot that the variance sum is strictly lower when $\eta$ is treated as a fixed parameter, compared to the nuisance parameter case. This is a good example of how knowledge on a parameter ($\eta$) can assist the estimation of other parameters ($\varphi$ and $\theta$). It is also observed that, when the noise is larger (i.e. when $\eta$ is smaller), the gain of precision by knowing $\eta$ is also bigger.

\subsection{Multiphase estimation with noise}\Label{subsec-multiphase}
Here we consider a noisy version of the multiphase estimation setting \cite{humphreys2013quantum,yue2014quantum}. 
This problem was first studied by \cite{yue2014quantum}, where the authors derived a lower bound for the quantum Fisher information and conjectured that it was tight.  Under local asymptotic covariance, we can now derive an attainable bound and show its equivalence to the SLD bound using the orthogonality of nuisance parameters, which proves the conjecture. 

Our techniques also allow to resolve an open issue about the result of Ref. \cite{yue2014quantum}, where it was  unclear whether or not the best precision depended on the knowledge of the noise. Using Corollary \ref{prop-orthogonal}, we will also see that knowing {\em a priori} the strength of the noise does not help to decrease the estimation error.

   \begin{figure}[h!]
\begin{center}
  \includegraphics[width=0.7\linewidth]{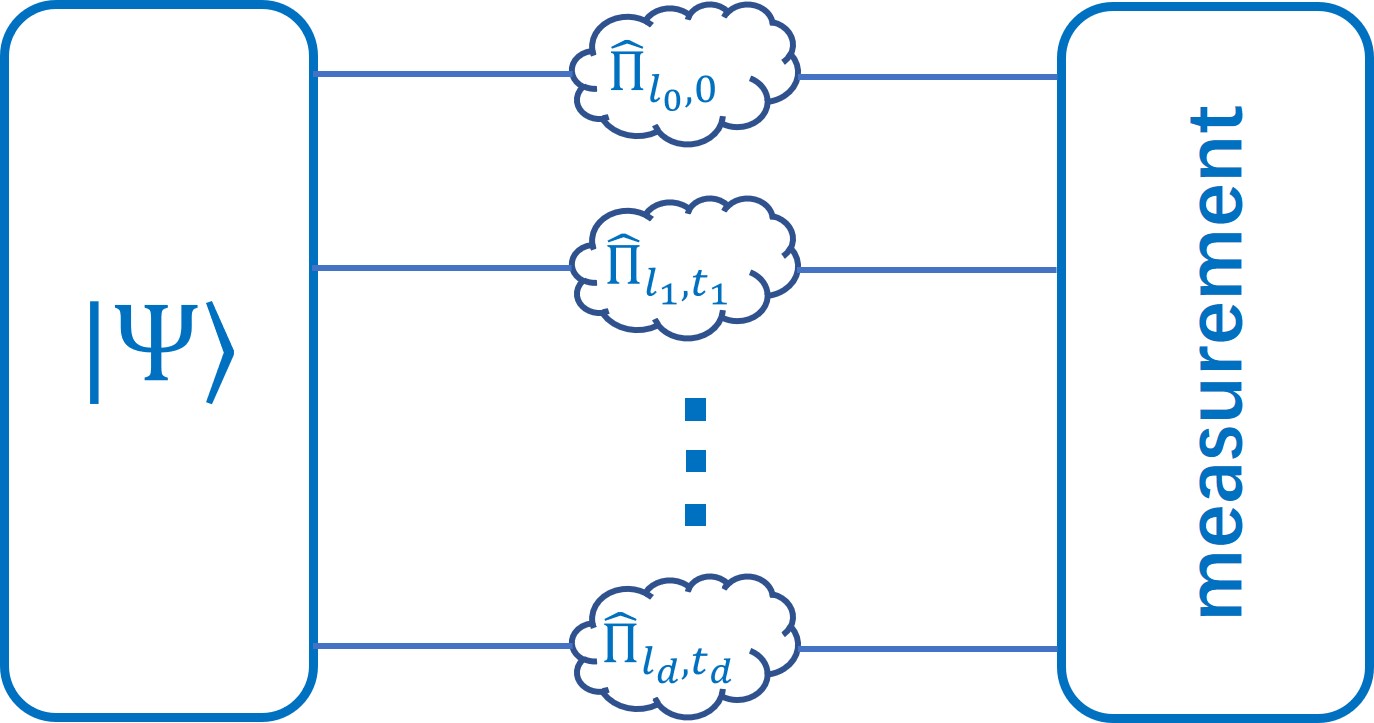}
  \end{center}
\caption{\Label{fig:multiphase}
  {\bf  Setup of  noisy multiphase estimation.} A input state $|\Psi\>$ passes through a noisy channel $\map{N}_{\vec{t}}$ consisting of $d+1$ modes and carrying $d$ parameters. The resultant state is then measured to obtain an estimate of $\vec{t}$. }
\end{figure}
The setting is illustrated in Fig. \ref{fig:multiphase}.
Due to photon loss, the phase-shift operation is no longer unitary. Instead, it corresponds to a noisy channel with the following Kraus form:
\begin{align*}
\hat{\Pi}_{l,\vec{t}}=\sqrt{\frac{(1-\eta)^{l}}{l!}}e^{i\hat{n}\vec{t}}\eta^{\frac{\hat{n}}{2}}\hat{a}^l
\end{align*}
\begin{align*}
\map{N}_{\vec{t}}(\rho):=\sum_{\vec{l}}\hat{\Pi}_{\vec{l},\vec{t}}\rho\hat{\Pi}_{\vec{l},\vec{t}}^\dag,
\qquad \hat{\Pi}_{\vec{l},\vec{t}}:=\hat{\Pi}_{l_j,0}\otimes\left(\bigotimes_{j=1}^d\hat{\Pi}_{l_j,\vec{t}_j}\right).
\end{align*}
Note that $\eta=0$ corresponds to the noiseless scenario. 
We consider a pure input state with $N$ photons and in the ``generalized NOON form'' as
\begin{align*}
|\Psi\>:=a|N\>_0+\frac{b}{\sqrt{d}}\sum_{j=1}^d |N\>_j\qquad |N\>_j:=|0\>^{\otimes j}\otimes|N\>\otimes |0\>^{\otimes (d-j)}.
\end{align*}
The output state from the noisy multiphase evolution would be
\begin{align*}
\map{I}\otimes\map{N}_{\vec{t}}(|\Psi\>\<\Psi|)=p_\eta\left|\psi_{\eta,\vec{t}}\right\>\left\<\psi_{\eta,\vec{t}}\right|+(1-p_\eta)\rho_\eta
\end{align*}
where
\begin{align*}
 \left|\psi_{\eta,\vec{t}}\right\>=\cos\alpha_\eta|N\>_0+\sin\alpha_\eta\sum_j\frac{e^{iN\vec{t}}|N\>_j}{\sqrt{d}},\qquad \cos\alpha_\eta=a/\sqrt{1-b^2(1-\eta^N)},
\end{align*}
 $p_\eta=1-b^2(1-\eta^N)$, and $\rho_\eta$ is independent of $\vec{t}$. Notice that the output state is supported by the finite set of orthonormal states $\{|n\>_j:\, j=0,\dots,d, n=0,\dots, N\}$, and thus it is in the scope of this work.

In this case, $\{\vec{t}_j\}$ are the parameters of interest, while $\alpha_\eta$  and $p_\eta$ can be regarded as nuisance parameters.
The SLD operators for these parameters can be calculated as
\begin{align*}
L_{\vec{t}_j}&=2i\left[N|N\>_j\<N|_j,|\psi_{\eta,\vec{t}}\>\<\psi_{\eta,\vec{t}}|\right]\\
L_{\alpha_\eta}&=2\left(|\psi_{\eta,\vec{t}}\>\<\psi^\perp_{\eta,\vec{t}}|+|\psi^\perp_{\eta,\vec{t}}\>\<\psi_{\eta,\vec{t}}|\right)\\
L_{p_\eta}&=\frac{1}{p_\eta}|\psi_{\eta,\vec{t}}\>\<\psi_{\eta,\vec{t}}|-\frac{1}{1-p_\eta}\wp_{\spc{H}_\perp},
\end{align*}
where $\wp_{\spc{H}_\perp}$ refers to the projection into the space orthogonal to $|\psi_{\eta,\vec{t}}\>$.
Notice that $p_\eta$ and $\alpha_\eta$ are orthogonal to other parameters, in the sense that 
\begin{align*}
\Tr\rho L_{\vec{t}_j}L_{p_\eta}=\Tr\rho L_{\alpha_\eta}L_{p_\eta}=0
\end{align*}
and 
\begin{align*}
\Tr\rho L_{\vec{t}_j}L_{\alpha_\eta}=\frac{2ip_\eta\sin2\alpha_\eta}{d}
\end{align*}
is purely imaginary. We also have
\begin{align*}
\Tr\rho L_{\vec{t}_j}L_{\vec{t}_k}&=4p_\eta N^2|\<\psi|N\>_j|^2\left(\delta_{jk}-|\<\psi|N\>_k|^2\right)\\
&=\frac{4p_\eta N^2\sin^2\alpha_\eta}{d}\left(\delta_{jk}-\frac{\sin^2\alpha_\eta}{d}\right).
\end{align*}
Therefore, the SLD Fisher information matrix and the D matrix are of the forms
\begin{align*}
J=\left(\begin{array}{ccc}J_{\vec{t}}& 0 & 0 \\
\\
0&J_{\alpha_\eta} & 0\\
\\
0& 0 & J_{p_\eta}\end{array}\right)\qquad D=\left(\begin{array}{ccc}0& D_{\vec{t},\alpha_\eta} & 0 \\
\\
D_{\vec{t},\alpha_\eta}^T&0 & 0\\
\\
0& 0 & 0\end{array}\right).
\end{align*}
Substituting the above into the bound (\ref{NHB}), we immediately get an attainable bound 
\begin{align}\Label{bound-multiphase}
\tr WV\left[\wp_{\vec{t}_0,\vec{t}|\seq{m}}\right]\ge\tr WJ_{\vec{t}}^{-1},\quad 
\left(J_{\vec{t}}\right)_{ij}=\frac{4p_\eta N^2\sin^2\alpha_\eta}{d}\left(\delta_{ij}-\frac{\sin^2\alpha_\eta}{d}\right)
\end{align}
for any locally asymptotically covariant measurement $\seq{m}$.
Taking $W$ to be the identity, one will see that for small $\eta$ the sum of the variances scales as $N^2/d^2$, while for $\eta\to1$ it scales as $N^2/d$, losing the boost in scaling compared to separate measurement of the phases.
The bound (\ref{bound-multiphase}) coincides with the SLD bound and the RLD bound. By Corollary \ref{prop-orthogonal}, we conclude that the SLD (RLD) bound can be attained in the case of joint estimation of multiple phases. In addition, we stress that the ultimate precision does not depend on whether or not the noisy parameter $\eta$ is known aprior: If $\eta$ is unknown, one can obtain the same precision as when $\eta$ is known by estimating $\eta$ without disturbing the parameters of interest.

\section{Conclusion}\Label{sec-conclusion}

\begin{table}[t!]
\begin{center}
    \begin{tabular}{ | l | l | l |}
    \hline
    {\bf Topic} & {\bf Our result} & {\bf Existing results} \\ \hline
    Class of  & Local asymptotic & Unbiased, \cite{helstrom-book} \\
    {estimators}  & covariance   & Locally unbiased \cite{holevo-book}, etc.\\
    \hline
Local    & Achieving Holevo bound &  \\ 
    optimal   & Order estimate & Achieving Holevo bound \cite{yamagata2013quantum} \\ 
    estimator &Uniform convergence  &     \\    
    \hline
Global    & Achieving Holevo bound &  \multirow{2}{*}{Achieving separable bound \cite{gill2000state,hayashi2005statistical}}\\ 
    estimator   & General cost function &  \\ 
    \hline
    Nuisance    & Achieving Holevo bound &  \multirow{3}{*}{Special cases, e.g. qubit models \cite{yang2017quantum,suzuki}}\\ 
parameter   & Order estimate &  \\ 
  & Uniform convergence &   \\ 
    \hline
Tail probability   & {Gaussian states} &  \multirow{2}{*}{Special pure states model \cite{H98,hayashi2017group}}\\ 
of limiting dist.   & {General D-inv. model} &  \\ 
    \hline

    \end{tabular}\caption{Comparison between our results and existing results.}
\end{center}
\end{table}

In this work, we completely solved the attainability problem of precision bounds for quantum state estimation under the local asymptotic covariance condition.
We provided an explicit construction for the optimal measurement which attains the bounds globally.
The key building block of the optimal measurement is the quantum local asymptotic normality, derived in \cite{qlan,guta-lan}  for a particular type of parametrization  and generalized here to arbitrary parameterizations. 
Besides the bound of MSE, we also derived a bound for the tail probability of estimation.
Our work provides a general tool of constructing benchmarks and optimal measurements in multiparameter state estimation.

Here, we should remark the relation with the results by Yamagata et. al. \cite{yamagata2013quantum},
which showed a similar statement for this kind of achievability in a local model scenario
by a kind of local quantum asymptotic normality.
In Theorem \ref{prop-fixed}, we have shown 
the compact uniformity with the order estimation in our convergence, but they did not such properties.
In the evaluation of global estimator, these properties for the convergence is essential.
The difference between our evaluation and their evaluation comes from the key tools.
The key tool of our derivation is Q-LAN (Proposition \ref{lem-QLAN}) by \cite{qlan,guta-lan},
which gives the state conversion, i.e.,
the TP-CP maps converting the states family with precise evaluation of the trace norm.
However, their method is based on the algebraic central limit theorem \cite{GW,Petz}, 
which gives only the behavior of the expectation of the function of operators $R_i$.
This idea of applying this method to the achievability of the Holevo bound
was first mentioned in \cite{hayashi2009quantum}.
Yamagata et. al. \cite{yamagata2013quantum} derived the detailed discussion in this direction.

Indeed, the algebraic version of Q-LAN by  \cite{GW,Petz} can be directly applied to 
the vector $\vec{X}$ of Hermitian matrices to achieve the Holevo bound
while use of the state conversion of Q-LAN requires some complicated procedure to handle the 
the vector $\vec{X}$ of Hermitian matrices, which is the disadvantage of our approach.
However, since the algebraic version of Q-LAN does not give a state conversion directly,
it is quite difficult to give the compact uniformity and the order estimate of the convergence.
In this paper, to overcome the disadvantage of our approach,
we have derived several advanced properties for Gaussian states 
in Sections \ref{s32} and \ref{s33} by using symplectic structure.
Using these properties, we could smoothly handle
complicated procedure to fill the gap between the full quit model and arbitrary submodel.

\begin{acknowledgement}
The authors thank the anonymous Referees of this paper for useful suggestions that helped improving the original manuscript.
This work is supported by the National Natural Science Foundation of China through Grant No. 11675136, 
by the Canadian Institute for Advanced Research
(CIFAR), by the Hong Kong Research Grant Council
through Grants No. 17300317 and 17300918, 
by the HKU Seed Funding for Basic Research, and 
by the Foundational Questions Institute through grant FQXi-RFP3-1325.   YY is supported by the Swiss National Science Foundation via the National Center for Competence in Research ``QSIT"
and by a Microsoft Research Asia Fellowship.
MH was supported in part by a JSPS Grant-in-Aid for Scientific Research 
(A) No.17H01280, (B) No. 16KT0017, the Okawa Research Grant and Kayamori
Foundation of Informational Science Advancement. YY acknowledges the hospitality of South University of Science and Technology of China during the completion of part of this work.
\end{acknowledgement}
\medskip
\bibliographystyle{unsrt}
\bibliography{ref} 

\appendix

\section{Proof of Lemma \ref{LGD}}\label{app:lemma1}
In this appendix, we show Lemma \ref{LGD}. For this aim,
we discuss the existence of PDF.
First, we show the following lemma. 
\begin{lem}\Label{L30}
Let $P$ be a probability measure on $\mathbb{R}$.
Define the location shift family $\{P_t\}$ as
$P_t(B):=P(B+t) $.
For an arbitrary disjoint decomposition ${\cal A}:=\{ A_i\}$ of $\mathbb{R}$,
we assume that the probability distribution family 
$\{P_{{\cal A},t}\}$ has finite Fisher information $J_{{\cal A},t}$,
where $P_{{\cal A},t}(i):= P_t(A_i)$.
We also assume that $J_{t} := \sup_{{\cal A}}J_{{\cal A},t}<\infty$.
Also, we define $x_+:=\inf \{x'| P(x',\infty)=0 \}$ and $x_-:=\sup \{x'| P((-\infty,x'])=0 \}$.

Then, for $x =x_-,x_+$,
the derivative 
$p(x):= \frac{d}{dx}P((-\infty,x])$ is zero.
For $x \in (x_-,x_+)$,
the derivative 
$p(x):= \frac{d}{dx}P((-\infty,x])$ exists.
For $x \in [x_-,x_+]$,
$p(x)$ is H\"{o}lder continuous with order $1/2$, i.e.,
\begin{align}
\limsup_{\epsilon\to 0} 
\Big|\frac{p(x+\epsilon)- p(x)}{\epsilon^{1/2}}\Big| < J_0^{3/2}.
\end{align}
Also, $p(x)$ is bounded, i.e., $\sup_x p(x)< \sqrt{J_0}$.
\end{lem}

\begin{proof}
Assume that $x_+< \infty$.
We choose $A_1:=(-\infty,x_+]$, $A_2:=(x_+,\infty)$, and ${\cal A}=\{A_1,A_2\}$.
The fidelity between $P_{{\cal A},0}$ and $P_{{\cal A},t}$ is
\begin{align}
&\sqrt{P((-\infty,x_+])} \sqrt{P((-\infty,x_+-t])}
=
P((-\infty,x_+]) \sqrt{1-\frac{P((x_+-t,x_+])}{P((-\infty,x_+])}} \nonumber \\
\cong&
P((-\infty,x_+]) 
\bigg(1-\frac{1}{2}\frac{P((x_+-t,x_+])}{P((-\infty,x_+])}
-\frac{1}{8} \bigg(\frac{P((x_+-t,x_+])}{P((-\infty,x_+])}\bigg)^2\bigg) \nonumber  \\
=&
1-\frac{1}{2}P((x_+-t,x_+])
-\frac{1}{8} P((x_+-t,x_+])^2
\end{align}
Hence,
we have
\begin{align}
\lim_{t \to 0}
\frac{8}{2 t^2}P((x_+-t,x_+])
=J_{\{ (-\infty,x_+], (x_+,\infty)\},0},
\end{align}
which implies the existence of $p(x_+)$ and $p(x_+)=0$.

Similarly, we can show that 
there exists $p(x_-)$ and $p(x_-)=0$ when $x_-> -\infty$.

Next, we choose $x \in (x_-,x_+)$.
We choose $A_1:=(-\infty,x]$, $A_2:=(x,\infty)$, and 
${\cal A}:=\{A_1,A_2\}$.
The fidelity between $P_{{\cal A},0}$ and $P_{{\cal A},t}$ is
\begin{align}
&\sqrt{P((-\infty,x])} \sqrt{P((-\infty,x+t])}
+\sqrt{P((x,\infty))} \sqrt{P((x+t,\infty))}\nonumber  \\
=&
P((-\infty,x]) \sqrt{1+\frac{P((x,x+t])}{P((-\infty,x])}}
+P((x,\infty)) \sqrt{1-\frac{P((x,x+t])}{P((x,\infty))}} \nonumber \\
\cong&
P((-\infty,x]) 
\bigg(1+\frac{1}{2}\frac{P((x,x+t])}{P((-\infty,x])}
-\frac{1}{8} \bigg(\frac{P((x,x+t])}{P((-\infty,x])}\bigg)^2\bigg) \nonumber \\
&+P((x,\infty)) 
\bigg(1
-\frac{1}{2}\frac{P((x,x+t])}{P((x,\infty))}
-\frac{1}{8}\bigg(\frac{P((x,x+t])}{P((x,\infty))}\bigg)^2
\bigg) \nonumber \\
=&
P((-\infty,x]) 
+\frac{1}{2}P((x,x+t])
-\frac{1}{8} \frac{P((x,x+t])^2}{P((-\infty,x])} \nonumber \\
&+P((x,\infty)) 
-\frac{1}{2} P((x,x+t])
-\frac{1}{8}\frac{P((x,x+t])^2}{P((x,\infty))}
\nonumber  \\
=&
1
-\frac{1}{8} \frac{P((x,x+t])^2}{P((-\infty,x])} 
-\frac{1}{8}\frac{P((x,x+t])^2}{P((x,\infty))}.
\end{align}
Hence,
we have
\begin{align}
\lim_{t \to 0}\frac{P((x,x+t])^2}{t^2P((-\infty,x])} 
+\frac{P((x,x+t])^2}{t^2P((x,\infty))}
=J_{\{ (-\infty,x], (x,\infty)\},0},
\end{align}
which implies the existence of $p(x)$.
Thus, 
\begin{align}
& |p(x)|=
\sqrt{\frac{J_{\{ (-\infty,x], (x,\infty)\},0} }{
\frac{1}{P((-\infty,x])} +\frac{1}{P((x,\infty))} } }
\le \sqrt{\frac{J_{0} }{
\frac{1}{P((-\infty,x])} +\frac{1}{P((x,\infty))} }} \nonumber \\
\le &\sqrt{\frac{J_{0} }{
\min\big( \frac{1}{P((-\infty,x])} ,\frac{1}{P((x,\infty))} \big) }} 
= \sqrt{J_{0} \max(P( (-\infty,x]),P((x,\infty)) ) } 
\le \sqrt{J_{0} }.
\end{align}
Hence, $p(x)$ is bounded.

For $x \in [x_-,x_+]$,
we consider the sets ${\cal A}_{x,d}:= \{(x,x+d],(x,x+d]^c\}$.
Then, we have
\begin{align}
&J_{{\cal A}_{x,d},0}
=
\frac{ (\frac{d}{dt}P((x+t,x+d+t]))^2 }{P((x,x+d]) }
+
\frac{ (\frac{d}{dt}P((x+t,x+d+t]^c))^2 }{P((x,x+d]^c) }
\nonumber \\
=&
\frac{ p((x,x+d])^2 }{P((x,x+d]) }
+
\frac{ p((x,x+d]^c)^2 }{P((x,x+d]^c) }.
\end{align}
Hence, we have
\begin{align}
\frac{ (p(x +d)-p(x))^2 }{P((x,x+d]) }
<J_{0}.
\end{align}
Hence, when $d \to 0$
\begin{align}
\frac{ (p(x +d)-p(x))^2 }{p(x)d  }
<J_{0},
\end{align}
i.e., 
\begin{align}
\frac{ (p(x +d)-p(x))^2 }{d  }
<p(x) J_{0} \le J_0^{3/2},
\end{align}
which implies that $p(x)$ is
H\"{o}lder continuous with order $1/2$.
\end{proof}

Using the previous Lemma, we are in position to prove Lemma \ref{LGD}. 

\begin{proofof}{Lemma \ref{LGD}}
Let $ {\cal A}:= \{A_i\}$ be an arbitrary disjoint finite decomposition of $\mathbb{R}$.
Let ${\cal G}_{{\cal A}}$ be the coarse-graining  map from a distribution on
$\mathbb{R}$ to a distribution on the meshes $A_i$.
Then, 
the Fisher information $J_{{\cal A},n,t}$ of $\{{\cal G}_{{\cal A}}(\wp^n_{t_0,t|M_n})\}_t$
is not greater than the Fisher information $J_{t}$ 
of $  \left\{\rho_{t_0,t}^{n} \right\}_{t  \in  \set \Delta_n} $.
Hence, the Fisher information $J_{{\cal A},t}$ of 
$\{{\cal G}_{{\cal A}}(\wp_{t_0,t|\seq{m}})\}_t$
satisfies
\begin{align}
J_{{\cal A},t}=\lim_{n \to \infty} J_{{\cal A},n,t}
\le J_{t_0}.
\end{align}
Therefore, we can apply Lemma \ref{L30} to $\wp_{t_0,t|\seq{m}}$.
Lemma \ref{L30} guarantees the existence of the PDF of the limiting 
distribution $\wp_{t_0,t|\seq{m}}$,
\end{proofof}
  
 \section{Lemmas used for asymptotic evaluations}\label{app:asymptoticstuff}
In this appendix, we prepare two lemmas for  asymptotic evaluations of information quantity 
of probability distributions.
\begin{lem}\Label{ALL8}
Assume that two sequences of probability distributions $\{(P_{n},Q_n)\}$ 
on $\mathbb{R}$ converges to 
a pair of probability distributions $(P,Q)$ on $\mathbb{R}$, respectively.
Then, the inequality
\begin{align}
F(P||Q)
\ge \limsup_{n \to \infty}
F(P_n||Q_n) \label{LJT}
\end{align}
holds.
\end{lem}

 \begin{proof}
 Let $p$ and $q$ be the Radon-Nikod\'{y}m derivative of $P$ and $Q$  
with respect to $P+Q$.
Given $N>0$, let ${\cal G}_{N}$ be the coarse-grained map from
a distribution on $\mathbb{R}$ to a distribution on 
the Borel subsets $ B \subset (-N, N]$ and the subset $(-N, N]^c$.
Given $N>0$ and $N'>0$, let 
${\cal G}_{N',N}$ be the coarse-grained map from
a distribution on $\mathbb{R}$ to a distribution on meshes 
$ (i/N', (i+1)/N']$ with $-NN'\le i \le NN'-1 $
and its complement $(-N, N]^c$.

We define
$\bar{\cal G}_{N',N} (P)(x):=
\frac{{\cal G}_{N',N} (P)(i)}{(P+Q)((i/N', (i+1)/N'])}$
for $x \in (i/N', (i+1)/N']\subset (-N, N]$.
Then, we have
$\bar{\cal G}_{N',N} (P)(x)\le 
\frac{{\cal G}_{N',N} (P)(i)}{P((i/N', (i+1)/N'])}=1/2$
and
$\bar{\cal G}_{N',N} (Q)(x)\le 1/2$.
Since 
$\sqrt{\bar{\cal G}_{N',N} (P)(x)}
\sqrt{\bar{\cal G}_{N',N} (Q)(x)}
\le \sqrt{1/2}\sqrt{1/2}=1/2 $ for $x \in (-N, N]$,
Lebesgue convergence theorem guarantees that
\begin{align}
& \lim_{N' \to \infty} 
\int_{-N}^N
\sqrt{\bar{\cal G}_{N',N} (P)(x)}
\sqrt{\bar{\cal G}_{N',N} (Q)(x)}
(P+Q)(dx) \nonumber \\
=&
\int_{-N}^N
\sqrt{p(x)}
\sqrt{q(x)}
(P+Q)(dx)\label{BFD}.
\end{align}
Since 
\begin{align}
&F\left({\cal G}_{N',N} (P)||  {\cal G}_{N',N} (Q)\right) \nonumber \\
=&
\int_{-N}^N
\sqrt{\bar{\cal G}_{N',N} (P)(x)}
\sqrt{\bar{\cal G}_{N',N} (Q)(x)}
(P+Q)(dx) \nonumber \\
&+\sqrt{P((-N, N]^c)}\sqrt{Q((-N, N]^c)},
\end{align}
Eq. \eqref{BFD} implies that
\begin{align}
\lim_{N' \to \infty} 
F\left({\cal G}_{N',N} (P)||  {\cal G}_{N',N} (Q)\right)
=
F\left({\cal G}_N (P)||  {\cal G}_N (Q)\right)
\Label{AL34}.
\end{align}
Also, we have
\begin{align}
\lim_{N \to \infty} 
F\left({\cal G}_N (P)||  {\cal G}_N(Q)\right)
=
F\left(P ||  Q\right)
\Label{AL34B}.
\end{align}
Then, information processing inequality for the fidelity
yields that
 \begin{align*}
F\left(P_n||Q_n\right)
\le 
F\left({\cal G}_{N',N}(P_n)||{\cal G}_{N',N}(Q_n)\right) .
 \end{align*}
Since the number of meshes is finite,
we have
\begin{align*}
\limsup_{n \to \infty}
F\left({\cal G}_{N',N}(P_n)||{\cal G}_{N',N}(Q_n)\right) 
=
F\left({\cal G}_{N',N} (P)||  {\cal G}_{N',N} (Q)\right)
\end{align*}
for every $N',N>0$.
Hence, we have
 \begin{align}
 \limsup_{n \to \infty} 
 F\left(P_n||Q_n\right)
\le 
F({\cal G}_{N',N} (P)||  {\cal G}_{N',N} (Q))
.\label{AL333}
 \end{align}
Hence, using \eqref{AL34}, \eqref{AL34B} and \eqref{AL333}, we have
\begin{align*}
\limsup_{n \to \infty}F(P_n||Q_n)
\le
F(P||Q).
\end{align*}
 \end{proof}

\begin{lem}\Label{LL8}
Let $\Theta$ be an open subset of $\mathbb{R}^{k'}$.
Assume that a sequence of probability distributions 
$\{P_{\vec{t},n}\}_{\vec{t} \in \Theta}$ 
on $\mathbb{R}^k$ converges to a family of probability distributions 
$\{P_{\vec{t}}\}_{\vec{t}\in \Theta}$ on $\mathbb{R}^k$.
We denote their $\epsilon$-difference Fisher information matrices
by $J_{\vec{t},\epsilon}^{n} $ and $J_{\vec{t},\epsilon}$, respectively.
For an vector $\vec{t}$ and $\epsilon >0$,
we also assume that there exists a Hermitian matrix $J_{\epsilon}$ such that
$J_{\vec{t},\epsilon}^{n} \le J_{\epsilon}$.
Then, $P_{\vec{t}+\epsilon e_j }$ is absolutely continuous with respect to $P_{\vec{t}}$
for $j=1, \ldots, k'$,
and 
the inequality
 \begin{align}
 \liminf_{n \to \infty} 
\langle \vec{a}| J_{\vec{t},\epsilon}^{n} -J_{\vec{t},\epsilon}|\vec{a}\rangle
\ge 0 
\Label{3-3-5B}
 \end{align}
holds for any complex vector $\vec{a} \in \mathbb{C}^{k'}$.
\end{lem}
\begin{proofof}{Lemma \ref{LL8}}
Since $J_{\vec{t},\epsilon}^{n}$ and $J_{\vec{t},\epsilon}$
are real matrices,
it is sufficient to show \eqref{3-3-5B} for a real vector $\vec{a}$.
In this proof, we fix the vector $\vec{t}$.

\noindent {\bf Step (1):} 
We show that $P_{\vec{t}+\epsilon e_j }$ is absolutely continuous with respect to $P_{\vec{t}}$
for $j=1, \ldots, k'$ by contradiction.
Assume that there exists an integer $j$ such that
$P_{\vec{t}+\epsilon e_j }$ is not absolutely continuous with respect to $P_{\vec{t}}$.
There exists a Borel set $B \subset \mathbb{R}^k$ such that
$P_{\vec{t}+\epsilon e_j }(B)>0$ and $P_{\vec{t}}(B)=0$.
Let ${\cal G}$ be the coarse-grained map from
a distribution $P$ on $\mathbb{R}^k$ to a binary distribution 
$(P(B),P(B^c))$ on two events $\{B,B^c\}$.
Let 
$J_{\vec{t},B,\epsilon}$ 
and $J_{\vec{t},B,\epsilon}^n$ 
be the $\epsilon$-difference Fisher information matrices of
$\{{\cal G}(P_{\vec{t}})\}$
and
$\{{\cal G}(P_{\vec{t},n})\}$, respectively.
Information processing inequality implies that
$
J_{\vec{t},B,\epsilon}^n\le
J_{\vec{t},\epsilon}^{n} \le J_{\epsilon}$.
Also, 
$J_{\vec{t},B,\epsilon}^n\to J_{\vec{t},B,\epsilon}$ as $n \to \infty$.
Hence, $J_{\vec{t},B,\epsilon}\le J_{\epsilon}$. 
However, the $j$-th diagonal element of $J_{\vec{t},B,\epsilon}$ is infinity.
It contradicts the assumption of contradiction.

\noindent {\bf Step (2):} 
Let $p_{\vec{t}+\epsilon e_j } $ be the Radon-Nikod\'{y}m derivative of $P_{\vec{t}+\epsilon e_j }$ 
with respect to $P_{\vec{t}}$.
We show that
\begin{align}
\lim_{R \to \infty}P_{\vec{t}}(
\{\vec{x}| p_{\vec{t}+\epsilon e_j}(\vec{x})> R\})
= 0 \Label{LPT}
\end{align}
for $N$, $\epsilon>0$, and any integer $j=1, \ldots, k'$
by contradiction.
We denote the LHS of \eqref{LPT} by $C_j$ and assume 
there exists an integer $j$ such that $C_j>0$.

We set $R= \sqrt{J_{\epsilon;j,j}/ C_j}+2$.
Setting $B$ to be $\{\vec{x}| p_{\vec{t}+\epsilon e_j}(\vec{x})> R\}$,
we repeat the same discussion as Step (1).
Then, we obtain the contradiction as follows.
\begin{align}
& J_{\epsilon;j,j} \ge J_{\vec{t},B,\epsilon;j,j}
\ge \int_{ B_j } (p_{\vec{t}+\epsilon e_j}(\vec{x})-1)^2 P_{\vec{t}}(d\vec{x})
\ge \int_{ B_j } (R-1)^2 P_{\vec{t}}(d\vec{x}) \nonumber \\
=&(R-1)^2/C_j=\Big(\sqrt{J_{\epsilon;j,j}/ C_j}+1\Big)^2/C_j >J_{\epsilon;j,j}.
\end{align}

\noindent {\bf Step (3):}
We show \eqref{3-3-5B} for a real vector $\vec{a}$.
We define the subsets 
\begin{align}
{\cal C}_R &:=\{\vec{x}| \exists j, p_{\vec{t}+\epsilon e_j}(\vec{x})> R\}
\\
{\cal C}_{N,R}& :=((-N, N]^k)^c 
\cup {\cal C}_R.
\end{align}
Given $R>0$,
let ${\cal G}_{R}$ be the coarse-grained map from
a distribution on $\mathbb{R}^k$ to a distribution on 
the family of measurable sets
$\{ B \subset \mathbb{R}^k\setminus {\cal C}_R \} \cup \{{\cal C}_R  \}$,
where $B$ is any Borel set in $\mathbb{R}^k\setminus {\cal C}_R$.
Given $N>0$ and $R>0$, let ${\cal G}_{N,R}$ 
be the coarse-grained map from
a distribution on $\mathbb{R}^k$ to a distribution on 
the family of measurable sets
$\{ B \subset \mathbb{R}^k\setminus {\cal C}_{N,R} \}\cup \{{\cal C}_{N,R}\}$, 
where $B$ is any Borel subset in $\mathbb{R}^k\setminus {\cal C}_{N,R}$.
Given $N>0$, $R>0$, and $N'>0$, let 
${\cal G}_{N',N,R}$ be the coarse-grained map from
a distribution on $\mathbb{R}^k$ to a distribution on meshes 
$ (\prod_{j=1}^k (i_j/N', (i_j+1)/N'])\setminus {\cal C}_{N,R}$ with 
$-NN'\le i_j \le NN'-1 $
and the complement ${\cal C}_{N,R}$.

We define
$$\bar{\cal G}_{N',N,R} (p)(\vec{x}):=
\frac{{\cal G}_{N',N,R} (p)(i_1, \ldots, i_k)}
{P_{\vec{t}}((\prod_{j=1}^k (i_j/N', (i_j+1)/N'])\setminus {\cal C}_{N,R})}$$
for $\vec{x} \in (\prod_{j=1}^k (i_j/N', (i_j+1)/N'])\setminus 
{\cal C}_{N,R}$.
Let $ J_{\vec{t},\epsilon,N',N,R}$
be the $\epsilon$-difference Fisher information matrix
for the distribution family
$\{ {\cal G}_{N',N,R}(p_{\vec{t}})\}_{\vec{t}}$.
Let $ J_{\vec{t},\epsilon,N,R}$
be the $\epsilon$-difference Fisher information matrix
for the distribution family
$\{ {\cal G}_{N,R}(p_{\vec{t}})\}_{\vec{t}}$.
Let $ J_{\vec{t},\epsilon,N}$
be the $\epsilon$-difference Fisher information matrix
for the distribution family
$\{ {\cal G}_{R}(p_{\vec{t}})\}_{\vec{t}}$.

Since 
$  
\bar{\cal G}_{N',N,R} (p_{\epsilon e_{j}})(\vec{x})
\bar{\cal G}_{N',N,R} (p_{\epsilon e_{j'}})(\vec{x})
\le
R^2$ for $\vec{x} \in {\cal C}_{N,R}^c$
and
\begin{align}
&\lim_{N'\to \infty}
\bar{\cal G}_{N',N,R} (p_{\vec{t}+\epsilon e_{j}})(\vec{x})
\bar{\cal G}_{N',N,R} (p_{\vec{t}+\epsilon e_{j'}})(\vec{x}) \nonumber \\
= &
{\cal G}_{N,R} (p_{\vec{t}+\epsilon e_{j}})(\vec{x})
{\cal G}_{N,R} (p_{\vec{t}+\epsilon e_{j'}})(\vec{x}),
\end{align}
Lebesgue convergence theorem guarantees that
\begin{align}
\lim_{N' \to \infty} ( J_{\vec{t},\epsilon,N',N,R})_{j,j'} =
( J_{\vec{t},\epsilon,N,R})_{j,j'} \Label{L33}
\end{align}
in the same way as \eqref{AL34}.
Since 
$p_{\vec{t}}(((-N, N]^k)^c) \to 0 $ 
as $N \to \infty$,
we have
\begin{align}
\lim_{N \to \infty} 
 J_{\vec{t},\epsilon,N,R} 
 =
 J_{\vec{t},\epsilon,R} 
\Label{L33C}.
\end{align}
Using \eqref{LPT}, we have
\begin{align}
\lim_{R \to \infty} 
 J_{\vec{t},\epsilon,R} 
 =
 J_{\vec{t},\epsilon} 
\Label{L33B}.
\end{align}

Let $J_{\vec{t},\epsilon,N',R,N}^{n}$ be
the $\epsilon$-difference Fisher information matrix
for the distribution family
$\{{\cal G}_{N',R,N}(p_{\vec{t},n})\}_{\vec{t}}$. 
Then, information processing inequality 
\eqref{MMM}
for the $\epsilon$-difference Fisher information matrix
yields that
\begin{align}
\langle \vec{a}| J_{\vec{t},\epsilon}^{n} -
J_{\vec{t},\epsilon,N',R,N}^{n}
|\vec{a}\rangle
\ge 0 .\label{BFT1}
\end{align}
Since the number of meshes is finite,
we have
\begin{align}
\lim_{n \to \infty}
J_{\vec{t},\epsilon,N',R,N}^{n}
=J_{\vec{t},\epsilon,N',R,N} .\label{BFT2}
\end{align}
The combination of 
\eqref{L33},
\eqref{L33C},
\eqref{L33B}
\eqref{BFT1},
and \eqref{BFT2} implies 
 \begin{align}
 \liminf_{n \to \infty} 
\langle \vec{a}| J_{\vec{t},\epsilon}^{n} -
J_{\vec{t},\epsilon,N',R,N}
|\vec{a}\rangle
\ge 0 .\Label{L333}
 \end{align}
Hence, using \eqref{L33}, \eqref{L33C}, \eqref{L33B}, and \eqref{L333}, we have
 $\liminf_{n \to \infty} 
\langle \vec{a}| J_{\vec{t},\epsilon}^{n} -
J_{\vec{t},\epsilon}
|\vec{a}\rangle
\ge 0 $.
 \end{proofof}

\section{Proof of Lemma \ref{3LL}}\Label{A3LL}
Before starting our proof of Lemma \ref{3LL}, we prepare the following lemmas.
\begin{lem}\Label{lem-QLAN2}
Consider a canonical quantum Gaussian states family $\{\Phi[\vec{\theta},\vec{\beta}]\}$.
When a symplectic matrix $S$ satisfies 
\begin{align}
SE_q\left(e^{-\vec{\beta}}\right)S^T=E_q\left(e^{-\vec{\beta}}\right)
\Label{C2},
\end{align}
where $E_q$ is the matrix defined in Eq. (\ref{M1}),
there exists a unitary operator $U_S$ such that
\begin{align*}
\Phi[\vec{\theta},\vec{\beta}]=U_S\Phi[S \vec{\theta},\vec{\beta}]U_S^\dag.
\end{align*}
\end{lem}
\begin{proof}
Consider any coordinate $\vec{\theta}'=(\vec{\theta}'^C,\vec{\theta}'^Q)$, where $\vec{\theta}'^Q$ obtained by a reversible linear transformation $S$ on the Q-LAN coordinate $\vec{\theta}^Q$, i.e. $\vec{\theta}'^Q=S\vec{\theta}^Q$.

Define $\hat{q}_j=\frac{1}{\sqrt{2}}(\hat{a}_j+\hat{a}_j^\dag)$, $\hat{p}_j=\frac{1}{i\sqrt{2}}(\hat{a}_j-\hat{a}_j^\dag)$, and $\vec{x}=(\hat{q}_1,\hat{p}_1,\dots,\hat{q}_q,\hat{p}_q)^T$. We have
\begin{align}
\Phi[\vec{\theta}_V^Q,\vec{\beta}]&=\Phi[S^{-1}\vec{\theta}'^Q,\vec{\beta}_V]\nonumber\\
&=Z_\beta \exp[i\vec{x}^TS^{-1}\vec{\theta}'^Q]\exp\left[-\frac{\vec{x}^TE_q(\vec{\beta}_V)\vec{x}}2\right]\exp\Phi[-i\vec{x}^TS^{-1}\vec{\theta}'^Q]\nonumber\\
&=Z_\beta \exp[i\vec{y}^T\vec{\theta}'^Q]\exp\left[-\frac{\vec{y}^TSE_q(\vec{\beta}_V)S^T\vec{y}}2\right]\exp\Phi[-i\vec{y}^T\vec{\theta}'^Q]\Label{app1-inter}
\end{align}
where $ \vec{y}:=(S^{-1})^T\vec{x}$ and $Z_\beta>0$ is a normalizing constant.
Now, by the definition of $E_q(\vec{x})$ in Eq. (\ref{M1}) and $SE_q(e^{-\vec{\beta}_V})'S^T=E_q\left(e^{-(\vec{\beta}_V)'}\right)$, $S$ must be of the block diagonal form $S=\bigoplus_{i} O_{\set{s}_i}$. Here $\{\set{s}_i\}$ is a partition of $\{1,\dots,2q\}$ and $j,k\in\set{s}_i$ if and only if $\beta'_j=\beta'_k$, and $O_{\set{s}_i}$ is an orthogonal matrix acting on any component $j\in\set{s}_i$. Since $\vec{\beta}_V'$, $\vec{\beta}_V$ and $\ln\vec{\beta}_V$ are in one-to-one correspondence, we $SE_q(e^{-\vec{\beta}_V})S^T=E_q(e^{-\vec{\beta}_V})$. Substituting it into Eq. (\ref{app1-inter}), we have
\begin{align*}
\nonumber\Phi[\vec{\theta}_V^Q,\vec{\beta}_V]&=Z_\beta \exp[i\vec{y}^T\vec{\theta}'^Q]\exp\left[-\frac{\vec{y}^TE_q( \vec{\beta}_V)\vec{y}}2\right]\exp\Phi[-i\vec{y}^T\vec{\theta}'^Q,\vec{\beta}_V].
\end{align*}
That is, $(S^{-1})^T$ can be regarded as a transformation of $\vec{x}$. Finally, $S$ is symplectic since $SDS^T=D$, and there exists a unitary $U_S$ such that \cite{hayashi2017group}
\begin{align}
\Phi[\vec{\theta}_V^Q,\vec{\beta}_V]&=U_S \Phi[\vec{\theta}'^Q,\vec{\beta}_V]U_S^\dag\Label{app1-inter2}.
\end{align}
Therefore we have $\Phi[\vec{\theta}_V^Q,\vec{\beta}_V]\cong \Phi[\vec{\theta}'^Q,\vec{\beta}_V]$ as desired. \qed
\end{proof}

\begin{proofof}{Lemma \ref{3LL}}
Using the imaginary part ${\sf Im}(\Gamma)$, 
we distinguish the classical and the quantum parts.
Specifically,  the kernel and  support of ${\sf Im}(\Gamma)$ are 
\begin{align}\Label{kernel}
\Ker {\sf Im} (\Gamma):=\left\{\vec{x}\in\R^k: {\sf Im} (\Gamma)\vec{x}=0\right\}
\end{align}
and
\begin{align}\Label{image}
\Supp {\sf Im} (\Gamma):=\left(\Ker {\sf Im} (\Gamma)\right)^\perp,
\end{align}  
respectively. We introduce the classical parameters $\vec{\theta}^C$ and the quantum parameters $\vec{\theta}^Q$ in
$\Ker {\sf Im} (\Gamma)$
and $\Supp {\sf Im} (\Gamma)$, respectively.
That is, 
the classical parameter $\vec{\theta}^C$ 
and 
the quantum parameter $\vec{\theta}^Q$
are given by an invertible linear transformation $T'$ such that 
$\vec{\theta}':=(\vec{\theta}^C,\vec{\theta}^Q)=T'\vec{t}$ satisfies
\begin{align}
\Ker {\sf Im} (\Gamma)&=
\{T'^{-1}(\vec{\theta}^C,0)| \vec{\theta}^C \in \R^{d^C}\}
\Label{reparameteriation1}
\\
\Supp {\sf Im} (\Gamma)&=
\{T'^{-1}(0,\vec{\theta}^Q)| \vec{\theta}^C \in \R^{2d^Q}\}.
\Label{reparameteriation2}
\end{align}
Since the above separation is unique up to the linear conversion
and any classical Gaussian states can be converted to each other
via scale conversion,
the remaining problem is to show the desired statement for the quantum part.

Next, we focus on the quantum part 
$(({T'}^{-1})^T\Gamma T'^{-1})^Q$
of the Hermitian matrix $ ({T'}^{-1})^T\Gamma T'^{-1}$
It is now convenient to define the matrix 
\begin{align}\Label{A}
A:=| {\sf Im} ((({T'}^{-1})^T\Gamma T'^{-1})^Q)|^{1/2}.
\end{align}
The role of $A$ is to normalize the $D$-matrix. Indeed,  since 
${\sf Im}((({T'}^{-1})^T\Gamma T'^{-1})^Q)$ is skew symmetric, 
$A^{-1}{\sf Im}((({T'}^{-1})^T\Gamma T'^{-1})^Q)A^{-1}$ 
is similar to $\Omega_{d^Q}$, namely that there exists an orthogonal matrix $S_0$ so that 
$S_0A^{-1}{\sf Im}((({T'}^{-1})^T\Gamma T'^{-1})^Q) A^{-1}S_0^T
=\Omega_{d^Q}$. 
Moreover, since $S_0A^{-1}
{\sf Re}((({T'}^{-1})^T\Gamma T'^{-1})^Q)
A^{-1}S_0^T$ is a real symmetric matrix, there exists a symplectic matrix $S$  
and a vector $\vec{\beta}$
such that \cite{williamson1936algebraic}
\begin{align}
S^T S_0A^{-1}
{\sf Im} ((({T'}^{-1})^T\Gamma T'^{-1})^Q)
A^{-1}S_0^T S= E_{d^Q}(e^{-\vec{\beta}}).\Label{SYAA}
\end{align}
Meanwhile, we have $SS_0 A^{-1}{\sf Im}((({T'}^{-1})^T\Gamma T'^{-1})^Q) A^{-1}
S_0S^{T}=
\Omega_{d^Q}$ since $S$ is symplectic.
Overall, when $T$ is given as $(I \oplus (SS_0A^{-1}))T'$,
the desired requirement is satisfied.

The uniqueness of $\vec{\beta}$ is guaranteed by the 
uniqueness of symplectic eigenvalues.
Hence, when two linear conversions $T$ and $\tilde{T}$ satisfies the condition of the statement,
$T \Gamma T^{T}=
\tilde{T} \Gamma \tilde{T}^{T}$. 
Thus, Lemma \ref{lem-QLAN2} guarantees that
the canonical Gaussian states 
$G( T^{-1}\vec{\alpha}, T \Gamma T^{T})$
and
$G( \tilde{T}^{-1}\vec{\alpha}, \tilde{T} \Gamma \tilde{T}^{T})$
are unitarily equivalent.
\end{proofof}

\section{Proof of Lemma \ref{LNS}}\Label{ALNS}
(3) $\Rightarrow$ (1):
When a Gaussian states family is given in the RHS of \eqref{FHY},
it is clearly $D$-invariant.

(1) $\Leftrightarrow$ (2):
Assume that the Gaussian states $G[ \vec{\alpha}, {\Gamma}]$ is generated by the operators
$\vec{R}=(R_1, \ldots, R_d)$.
Due to Lemma \ref{lem-QLAN-QFI1}, the SLDs of $\{G[ T(\vec{t}), {\Gamma}]\}$ are 
$L_j:=\sum_{k,k'} T_{k,j} (A^{-1})_{k,k'} R_{k'}$.
Lemma \ref{lem-QLAN-QFI1} guarantees that 
\begin{align*}
{\cal D}(L_j)&= 
\sum_{k,k'} T_{k,j} (A^{-1})_{k,k'} {\cal D}(R_{k'})
=\sum_{k,k'} T_{k,j} (A^{-1})_{k,k'} \sum_{j'} 2 B_{k',j'}R_{j'}\\
&=-2 \sum_{k}R_k (B A^{-1} T)_{k,j}.
\end{align*}
Hence, the $D$-invariance is equivalent to the condition (2).

(2) $\Rightarrow$ (3):
First, we separate the system into the classical and the quantum parts.
In the Gaussian states family $\{G[ \vec{\alpha}, {\Gamma}]\}$
this separation can be done by considering the Kernel of $Im (\Gamma)$ as in the proof of Lemma \ref{3LL}.
In the Gaussian states family $G[ T(\vec{t}), {\Gamma}]$
this separation can be done by considering the Kernel of the $D$-matrix $D_{\vec{0}}$ in the same way.
Since the relation \eqref{FHY} for the classical part is easily done,
we show the relation \eqref{FHY} when only the quantum part exists.

Under the above assumption,
we define the $k \times (d-k)$ matrix $T'$ such that
$F:=(T \oplus T')$ is invertible and $T'^T A^{-1} T=0$.
Then, Lemma \ref{3LL} guarantees that 
$G[ F(\vec{t},\vec{t}'), \Gamma]$ is unitarily equivalent to 
$G[ (\vec{t},\vec{t}'), F^{-1} \Gamma (F^T)^{-1}]$.
Since $T'^T A^{-1} T=0$, we have
\begin{align}
 &F^{-1} \Gamma (F^T)^{-1}
=
F^{-1} A (F^T)^{-1} F^T A^{-1}\Gamma A^{-1} F  F^{-1 } A(F^T)^{-1} \nonumber \\
=&
(F^T T'^T F)^{-1} F^T A^{-1}\Gamma A^{-1} F  (F^T T'^T F)^{-1}
=  {\Gamma}_T \oplus {\Gamma}_0,
\end{align}
where
\begin{align*}
{\Gamma}_0:= ((T'^T A^{-1} T')^{-1}+ i (T'^T A^{-1} T')^{-1}(T'^T B T')(T'^T A^{-1} T')^{-1}).
\end{align*}
Hence, $G[ (\vec{t},\vec{t}'), F^{-1} \Gamma (F^T)^{-1}]
=G[ \vec{t}, \Gamma_T]
\otimes G[ \vec{t}', \Gamma_0]$.
Putting $\vec{t}'=\vec{0}$, we obtain the condition (3).

\section{Proof of Lemma \ref{ANH}}\Label{AANH}

Since Lemma \ref{3LL} shows that general Gaussian states can be reduced to the canonical Gaussian states,
we discuss only the canonical Gaussian states.

\noindent{\bf Step 1:}
We show the statement when we have only the quantum part and $\vec{X}=\vec{R}$.
For a given state $\rho$, we define the POVM $M_\rho$
by 
\begin{align}
M_\rho(\set{B}):=\int_{\set{B}} T_{\vec{\alpha}} \rho T_{\vec{\alpha}}^\dagger d \vec{\alpha}.
\end{align}
When $\rho$ is a squeezed state with $ \Tr \rho Q_j=\Tr \rho_j P=0$,
the output distribution 
$\wp_{\vec{\alpha}| M[\rho]}$
of $M[\rho]$ is the $2d^Q$-dimensional normal distribution of average $\vec{\alpha}$ 
and the following covariance matrix \cite{holevo-book}; 
\begin{align}
E_{d^Q}(\vec{\beta})+ V_\rho,
\hbox{ with }
V_{\rho}:= 
\left(
\begin{array}{cc}
(\Tr Q_i Q_j\rho)_{i,j} & (\Tr Q_i P_j\rho)_{i,j} \\
( \Tr P_i Q_j \rho)_{i,j} &  (\Tr P_i P_j \rho)_{i,j}
\end{array}
\right). \Label{BDE}
\end{align}

In the single-mode case, without loss of generality, we can assume that 
$W$ is 
a diagonal matrix $
\left(
\begin{array}{cc}
w_1 & 0 \\
0  & w_2
\end{array}
\right)$
because this diagonalization can be done by applying the orthogonal transformation between $Q$ and $P$.
Then, 
\begin{align}
\frac{1}{2}\sqrt{W}^{-1} |\sqrt{W} \Omega_{d^Q} \sqrt{W}|\sqrt{W}^{-1}
=
\left(
\begin{array}{cc}
\frac{\sqrt{w_1}}{2 \sqrt{w_1}} & 0 \\
0  & \frac{\sqrt{w_1}}{2 \sqrt{w_2}}
\end{array}
\right).
\end{align}
We define the squeezed state $\rho[w]$ by
$ V_\rho=\left(
\begin{array}{cc}
\frac{\sqrt{w}}{2 } & 0 \\
0  & \frac{1}{2 \sqrt{w}}
\end{array}
\right)$.
Then the squeezed state $\rho[\frac{{w_2}}{ {w_1}}]$
satisfies the condition 
$V_{\rho[\frac{{w_2}}{ {w_1}}]}= \frac{1}{2}\sqrt{W}^{-1} |\sqrt{W} \Omega_{d^Q} \sqrt{W}|\sqrt{W}^{-1}
$.
Hence, the POVM $M[\rho[\frac{{w_2}}{ {w_1}}]]$ satisfies the requirement.

In the multiple-mode case,
we choose a symplectic matrix $S$ such that 
$S W S^T$
is a diagonal matrix
with diagonal element $w_1,w_2, \ldots, w_{2d^Q}$.
The matrix $$\frac{1}{2}S\sqrt{W}^{-1} |\sqrt{W} \Omega_{d^Q} \sqrt{W}|\sqrt{W}^{-1}S^T$$
is the diagonal matrix 
with diagonal element 
$\frac{\sqrt{w_1}}{2 \sqrt{w_2}},\frac{\sqrt{w_1}}{2 \sqrt{w_2}}, \ldots,
\frac{\sqrt{w_{2d^Q}}}{2 \sqrt{w_{2d^Q-1}}},
\frac{\sqrt{w_{2d^Q-1}}}{2 \sqrt{w_{2d^Q}}}$.
Employing symplectic representation $U(S)$ given in \cite[Section 7.8]{hayashi2017group},
we apply the unitary operator $U(S)$.
Then, the state 
$\rho[\frac{{w_2}}{ {w_1}}] \otimes\cdots\otimes
\rho[\frac{{w_{2d^Q}}}{{w_{2d^Q-1}}}]$
satisfies the condition 
$V_{U(S)
\rho[\frac{{w_2}}{ {w_1}}] \otimes\cdots\otimes
\rho[\frac{{w_{2d^Q}}}{{w_{2d^Q-1}}}]
U(S)^\dagger}
=\frac{1}{2}\sqrt{W}^{-1} |\sqrt{W} \Omega_{d^Q} \sqrt{W}|\sqrt{W}^{-1}$.
Therefore, 
Hence, the POVM $M[U(S)
\rho[\frac{{w_2}}{ {w_1}}] \otimes\cdots\otimes
\rho[\frac{{w_{2d^Q}}}{{w_{2d^Q-1}}}]U(S)^\dagger
]$ satisfies the requirement.



\noindent{\bf Step 2:}
As the next step, we show the statement when we have only quantum part and 
the vector $\vec{X}$ is given in the following way with an integer $\tilde{k}$;
\begin{align}
X_1&=Q_1, ~X_2=P_1, ~\ldots,~ 
X_{2\tilde{k}-1}=Q_{\tilde{k}},~
X_{2\tilde{k}}=P_{\tilde{k}}
\Label{DHY}\\
X_{2\tilde{k}+1}&= Q_{\tilde{k}+1},~ 
X_{2\tilde{k}+2}= Q_{\tilde{k}+2}, ~\ldots,~
X_{k}= Q_{k-\tilde{k}}.
\Label{DHY2}
\end{align}
For the latter $k-2\tilde{k} $ parameters, we just apply the measurements for 
the observables $Q_{\tilde{k}+1},\ldots, Q_{k-\tilde{k}}$.
We denote the part of the initial $2\tilde{k}$ parameters for $W$ by $W_Q$.
We substituting $W_Q$ into $W$ in the construction of the POVM $M[U(S)
\rho[\frac{{w_2}}{ {w_1}}] \otimes\cdots\otimes
\rho[\frac{{w_{2d^Q}}}{{w_{2d^Q-1}}}]U(S)^\dagger
]$.
Then, as the output distribution, it 
realizes the normal distribution with average $\vec{t}$ and covariance matrix
 ${\sf Re} ( (Z_{\vec{t}}(\vec{X}))+
\sqrt{W}^{-1} |\sqrt{W} {\sf Im} ( Z_{\vec{t}}(\vec{X}))\sqrt{W}| \sqrt{W}^{-1}$.

\noindent{\bf Step 3:}
Next, we consider the case when we have only quantum part and 
$\vec{X}$ does not have the form \eqref{DHY}, \eqref{DHY2}.
Let $2 \tilde{k}$ be the rank of $ {\sf Im} ( Z_{\vec{t}}(\vec{X}))$.
We choose an invertible matrix $T'$ such that 
non-zero entries of the matrix $T' {\sf Im} ( Z_{\vec{t}}(\vec{X})){T'}^{T}$
are limited in the first $2 \tilde{k}$ components
and are given as $\Omega_{\tilde{k}}$.
When we apply the matrix $T'$ to the outcome of $M$,
we obtain another POVM, which is denoted by $M'$.
$M'$ can be considered as a measurement for $X':=T'(\vec{X})$ instead of $\vec{X}$.
Hence, we have the relation
$ V_{\vec{\alpha}}(M') = T' V_{\vec{\alpha}}(M){T'}^T$.
Thus, it is sufficient to discuss the case with $\vec{X}'$ and the weight matrix
$ {T'}^T WT'$.

Then, we define $Y_j:=(T' \vec{X})_j$ for $j=1, \ldots, 2\tilde{k}$ and 
$Y_{2\tilde{k}+2j'-1}:= (T' \vec{X})_{2\tilde{k}+j'}$ for $j'= 1, \ldots, k-2 \tilde{k}$.
Next, we choose $k'-k$ operators $Y_{2\tilde{k}+2j}$ for $j= 1, \ldots, k-2 \tilde{k}$
and $Y_{2k-2 \tilde{k}+j'}$ for $j'=1, \ldots, k'$
as linear combinations of
$Q_j$ and $P_j$ with $j=1, \ldots, k'/2$ such that
$ {\sf Im} ( Z_{\vec{t}}(\vec{Y}))=\Omega_{k'/2}$. 
Hence, there exists a symplectic matrix $S$ such that 
$ \vec{Y}= S \vec{Q}$, where the vector $\vec{Q}$ is defined as $(\vec{Q})_{2j-1}=Q_j$
and $(\vec{Q})_{2j}=P_j$.
Employing symplectic representation $U(S)$ given in \cite[Section 7.8]{hayashi2017group},
we apply the unitary operator $U(S)$ to the Hilbert space so that
$U(S)(T' \vec{X})_j$ satisfies the conditions\eqref{DHY} and \eqref{DHY2}.
Hence, our problem is reduced to Step 2.

\noindent{\bf Step 4:}
We consider the case when our system is composed of the classical and the quantum parts.
$X_i$ is given as a linear combination of
the operators $Q_j$ and $P_j$, 
and the classical random variables $Z_{j'}$.
Then, we divide $\vec{X}=(X_i)$ to the sum of 
the quantum part $\vec{X}^Q=(X_i^Q)$ and the classical part 
$\vec{X}^C=(X_i^C)$.
Then, we have $Z_{\vec{\alpha}}(\vec{X})=Z_{\vec{t}}(\vec{X}^Q)+Z_{\vec{\alpha}}(\vec{X}^C) $,
which implies that
$
\tr W  {\sf Re} ( Z_{\vec{\alpha}}(\vec{X}))+
\tr |\sqrt{W} {\sf Im} ( Z_{\vec{\alpha}}(\vec{X}))\sqrt{W}| 
=
\tr W  {\sf Re} ( Z_{\vec{\alpha}}(\vec{X}^Q))+
\tr |\sqrt{W} {\sf Im} ( Z_{\vec{\alpha}}(\vec{X}^Q))\sqrt{W}| 
\tr |\sqrt{W} {\sf Im} ( Z_{\vec{\alpha}}(\vec{X}^C))\sqrt{W}| $.
Hence, 
when the outcome is given by the sum of 
 $\vec{X}^C$ and the outcome of Step 3 with $\vec{X}^Q$,
the desired properties are satisfied.

\section{Proof of Eq. (\ref{QFI-QLAN})}\Label{app-QFI}
By additivity of quantum Fisher information, the quantum Fisher information of $\{\rho_{\vec{\theta},n}\}$ is equal to the quantum Fisher information of $\{\rho_{\vec{\theta}_0+\vec{\theta}}\}$. For the latter, solving the equations
\begin{align}
\frac{\partial{\rho_{\vec{\theta}_0+\vec{\theta}}}}{\partial \theta^{\rm R/I}_{j,k}}=\frac12\left(L_{\theta^{\rm R/I}_{j,k}}\rho_{\vec{\theta}_0+\vec{\theta}}+\rho_{\vec{\theta}_0+\vec{\theta}}L_{\theta^{\rm R/I}_{j,k}}\right),
\end{align}
we get 
\begin{align}\Label{SLD-qudit}
L_{\theta^{\rm R}_{j,k}}=\frac{2\sqrt{\theta_{0,j}-\theta_{0,k}}}{(\theta_{0,j}+\theta_{0,k})}\cdot T^{\rm I}_{j,k}\qquad
L_{\theta^{\rm I}_{j,k}}=-\frac{2\sqrt{\theta_{0,j}-\theta_{0,k}}}{(\theta_{0,j}+\theta_{0,k})}\cdot T^{\rm R}_{j,k}.
\end{align}
Using Eq. (\ref{SLD-qudit}), we can evaluate the SLD quantum Fisher information 
\begin{align}
(J_{\vec{\theta}_0})_{j,k}^{\rm R}=(J_{\vec{\theta}_0})_{j,k}^{\rm I}=\frac{4(1-\beta_{j,k})}{1+\beta_{j,k}}\qquad (J_{\vec{\theta}_0})_{j,k}^{\rm (R,I)}=(J_{\vec{\theta}_0})_{j,k}^{\rm (I,R)}=0
\end{align}
and the D-matrix
\begin{align}
(D_{\vec{\theta}_0})_{j,k}^{\rm R}=(D_{\vec{\theta}_0})_{j,k}^{\rm I}=0\qquad (D_{\vec{\theta}_0})_{j,k}^{\rm (R,I)}=-(D_{\vec{\theta}_0})_{j,k}^{\rm (I,R)}=-\frac{8(1-\beta_{j,k})^2}{(1+\beta_{j,k})^2},
\end{align}
having used the definition $\beta_{j,k}:=\theta_{0,k}/\theta_{0,j}$. It can be immediately verified that $\left(\widetilde{J}_{\vec{\theta}_0}\right)^{-1}=(J_{\vec{\theta}_0})^{-1}+\frac{i}{2}(J_{\vec{\theta}_0})^{-1}D_{\vec{\theta}_0}(J_{\vec{\theta}_0})^{-1}$ matches Eq. (\ref{QFI-QLAN}).

The case with the displaced thermal state $\rho_{\theta^{\rm R}_{j,k}+i\theta^{\rm I}_{j,k},\beta_{j,k}}$
follows from Lemma \ref{lem-QLAN-QFI1}.

\section{Proof of Lemma \ref{L54}}\Label{app-lemma12}
Denote by $Q(\vec{y}):=\frac1{\pi^{k/2}}\<\vec{y}| 
F|\vec{y}\>$ the Q-function of $F$ \cite{husimi1940some}. Expanding displaced thermal states into a convex combination of coherent states, Eq. (\ref{Gaussian-limiting2}) can be rewritten as
\begin{align}
f(\vec{\alpha})&=\int d\vec{y}\,\left(\prod_{j}\sqrt{\frac{1-\gamma_j}{\gamma_j}}\,e^{-(1-\gamma_j)(y_j-\alpha_j)^2/\gamma_j}\right)\,Q(\vec{y}). \Label{HFD}
\end{align}
Taking the Fourier transform 
$\map{F}_{\vec{y}\to\vec{\xi}}(g):=\int d\vec{y}\ e^{i\vec{y}\cdot\vec{\xi}} g$ 
on both sides, we get
\begin{align}
&\map{F}_{\vec{\alpha}\to\vec{\xi}}\left(
f(\vec{\alpha})\right)\nonumber\\
=&\int\int d\vec{\alpha} d\vec{y}\,\left(\prod_{j}\sqrt{\frac{1-\gamma_j}{\gamma_j}}\,e^{i\xi_j\alpha_j- (1-\gamma_j)(y_j-\alpha_j)^2/\gamma_j}\right)\,Q(\vec{y})\nonumber\\
=&\left(\prod_{j}\int d\alpha_j\sqrt{\frac{1-\gamma_j}{\gamma_j}}\,e^{-\frac{1-\gamma_j}{\gamma_j}\left(\alpha_j-y_j-\frac{i\xi_j\gamma_j}{2(1-\gamma_j)}\right)^2}\right)\int d\vec{y}\,e^{i\xi_jy_j-\sum_j\frac{\gamma_j\xi_j^2}{4(1-\gamma_j)}}\,Q(\vec{y})\nonumber\\
=&\sqrt{\pi^k}e^{-\frac14\sum_j\frac{\gamma_j\xi_j^2}{1-\gamma_j}}\map{F}_{\vec{y}\to\vec{\xi}}\left(Q(\vec{y})\right).
\Label{KGT}
\end{align}
In addition, we know that  the P-function $P(\vec{y})$ \cite{glauber1963coherent} of $F$  can be evaluated via the Q-function as (see, for instance, \cite{schleich2011quantum})
\begin{align}
P(\vec{y})=\map{F}^{-1}_{\vec{\xi}\to\vec{y}}\left(e^{\frac14\sum_j\xi_j^2}\map{F}_{\vec{\alpha}\to\vec{\xi}}\left(
Q(\vec{\alpha})\right)\right).
\Label{KGT2}
\end{align}
The combination of \eqref{KGT} and \eqref{KGT2} yields
\begin{align*}
P(\vec{y})=\map{F}^{-1}_{\vec{\xi}\to\vec{y}}\left(\sqrt{\pi^k}e^{\frac14\sum_j\frac{\xi_j^2}{1-\gamma_j}}\map{F}_{\vec{\alpha}\to\vec{\xi}}\left( f(\vec{\alpha})\right)\right).
\end{align*}
By definition of the P-function $P(\vec{y})$, 
$F$ satisfies
\begin{align}\Label{app-MG6}
F=\int d\vec{y} P(\vec{y})|\vec{y}\>\<\vec{y}| 
=\int d\vec{y}\,\map{F}^{-1}_{\vec{\xi}\to\vec{y}}\left(\sqrt{\pi^k}e^{\frac14\sum_j\frac{\xi_j^2}{1-\gamma_j}}\map{F}_{\vec{\alpha}\to\vec{\xi}}\left( 
f(\vec{\alpha})
\right)\right)|\vec{y}\>\<\vec{y}|.
\end{align}

Conversely, we assume that
$F$ is given by \eqref{Gaussian-limiting3}.
Then, we choose the function $Q(\vec{\alpha})$ to satisfy
\begin{align}
\map{F}_{\vec{\alpha}\to\vec{\xi}}\left(
f(\vec{\alpha})\right)
=\sqrt{\pi^k}e^{-\frac14\sum_j\frac{\gamma_j\xi_j^2}{1-\gamma_j}}\map{F}_{\vec{y}\to\vec{\xi}}\left(Q(\vec{y})\right).\Label{TGY}
\end{align}
Since $F= \int d\vec{y} 
\map{F}^{-1}_{\vec{\xi}\to\vec{y}}\left(e^{\frac14\sum_j\xi_j^2}\map{F}_{\vec{\alpha}\to\vec{\xi}}\left(
Q(\vec{\alpha})\right)\right)
|\vec{y}\>\<\vec{y}|$,
we have
$Q(\vec{y})=\frac1{\pi^{k/2}}\<\vec{y}| 
F|\vec{y}\>$.
Expanding displaced thermal states into a convex combination of coherent states, 
we have
\begin{align}
\Tr F G[\vec{\alpha},\vec{\gamma}]=
\int d\vec{y}\,\left(\prod_{j}\sqrt{\frac{1-\gamma_j}{\gamma_j}}\,e^{-(1-\gamma_j)(y_j-\alpha_j)^2/\gamma_j}\right)\,Q(\vec{y})\Label{GDE}.
\end{align}
Applying the inverse of $\map{F}_{\vec{\alpha}\to\vec{\xi}}$ to \eqref{TGY}, we obtain
\eqref{HFD}.
The combination of \eqref{HFD} and \eqref{GDE} implies \eqref{Gaussian-limiting2}.

\section{Proof of Lemma \ref{L2B}}\Label{app-Gaussian}
First we focus on the quantum part and we show that,
when a POVM $\{M_0^G \}_{B \subset \mathbb{R}^{2s}}$ is covariant, 
\begin{align}
  \wp_{\vec{\alpha}^Q|M_0^G}\left(\set{T}_{W^Q,c}(\vec{\alpha}_0^Q)\right)\ge
N\left[\vec{0}, E_s\left(e^{-\vec{\beta}}+\vec{e}/2 \right)\right]\left(\set{T}_{W^Q,c}(\vec{0})\right)
\Label{ER5}
\end{align}
where $\vec{\alpha}^Q=(\vec{\alpha}^R,\vec{\alpha}^I)$ and $W^Q$ is a diagonal matrix with eigenvalues
$w_j>0$.

Since $M^Q$ is covariant, we have $\wp_{\vec{\alpha}^Q|M_0^G}\left(\set{T}_{W^Q,c}(\vec{\alpha}_0^Q)\right)=\wp_{\vec{0}|M_0^G}\left(\set{T}_{W^Q,c}(\vec{0})\right)$.
There exists a state $ \tau$ such that
\begin{align*}
 M^Q(\set{B})=\frac{1}{\pi^s}\int_{\set{B}}d \vec{\alpha}^Q\ T^{\rm Q}_{\vec{\alpha}^Q}\ \tau (T^{\rm Q}_{\vec{\alpha}^Q})^\dagger,
 \end{align*}
 where $T^{\rm Q}_{\vec{\alpha}^Q}$ is defined as in Eq. (\ref{D-alpha}).	
Let $\hat{N}_j$ be the number operator on the $j$-th system.
Since $\hat{N}_j \Phi[\vec{0},\vec{\beta}]= \Phi[\vec{0},\vec{\beta}]\hat{N}_j$, 
for the set $\set{T}_{W^Q,c}(\vec{0})$,
we have
\begin{align*}
& \Tr \Phi[\vec{0},\vec{\beta}]M  (\set{T}_{W^Q,c}(\vec{0}))
 \\
 =&
\Tr \Phi[\vec{0},\vec{\beta}]
e^{i t \hat{N}_j} M (\set{T}_{W^Q,c}(\vec{0}))
  e^{-i t \hat{N}_j} \\
=& 
\Tr \Phi[\vec{0},\vec{\beta}]e^{it \hat{N}_j}\frac{1}{\pi^s}\int_{\set{T}_{W^Q,c}(\vec{0})} D_{\vec{\alpha}^Q} \tau D_{\vec{\alpha}^Q}^\dagger 
d \vec{\alpha}^Q
e^{-i t \hat{N}_j} \\
=& 
\Tr \Phi[\vec{0},\vec{\beta}]
\frac{1}{\pi^s}\int_{ \set{T}_{W^Q,c}(\vec{0})} 
D_{\vec{\alpha}^Q}
e^{it \hat{N}_j} \tau e^{-it \hat{N}_j} 
D_{\vec{\alpha}^Q}^\dagger 
d \vec{\alpha}^Q
 \end{align*}
Since this relation holds with any $\hat{N}_j$.
Hence, $\tau$ can 
be replaced by
\begin{align*}
\frac{1}{(2\pi)^s}
\int_{[0,2\pi]^s} e^{i \sum_{j=1}^s t_j \hat{N}_j} \tau e^{-i \sum_{j=1}^st_j \hat{N}_j} dt_1\ldots d t_s \, , 
\end{align*} which commutes with $\hat{N}_j$ with any $j$.

Now, we consider the case with $\tau=|k_1,\ldots, k_s\rangle \langle k_1,\ldots, k_s| $.
Then,
\begin{align}
&\Tr \Phi[\vec{0},\vec{\beta}] M  (\set{T}_{W^Q,c}(\vec{0})) \nonumber \\
=&
\frac{1}{\pi^s}\int_{ \set{T}_{W^Q,c}(\vec{0})} 
\langle k_1,\ldots, k_s|
D_{\vec{\alpha}^Q}^\dagger 
 \Phi[\vec{0},\vec{\beta}] D_{\vec{\alpha}^Q}
|k_1,\ldots, k_s\rangle 
d \vec{\alpha}^Q\nonumber  \\
=&
\frac{1}{\pi^s}\int_{ \set{T}_{W^Q,c}(\vec{0})} 
\langle k_1,\ldots, k_s|
 \Phi[\vec{\alpha}^Q,\vec{\beta}]
|k_1,\ldots, k_s\rangle 
d \vec{\alpha}^Q\nonumber  \\
=&
\frac{1}{\pi^s}\int_{ \set{T}_{W^Q,c}(\vec{0})} 
\int_{}
|\langle k_1,\ldots, k_s|\alpha_1', \ldots, \alpha_s'\rangle |^2
\frac{1}{\pi^s N_1 \ldots N_s}e^{-\frac{\sum_{j=1}^s|\alpha_j'-\alpha_j|^2}{N_j}}
d \vec{\alpha}'
d \vec{\alpha}^Q\nonumber \\
=&
\frac{1}{\pi^s}\int_{ \set{T}_{W^Q,c}(\vec{0})} 
\int_{}
e^{-\sum_{j=1}^s |\alpha_j'|^2}
\frac{|\alpha_1'|^{2k_1} \cdots |\alpha_s'|^{2 k_s}}{k_1 ! \cdots k_s !}
\frac{1}{\pi^s N_1 \ldots N_s}e^{-\frac{\sum_{j=1}^s|\alpha_j'-\alpha_j|^2}{N_j}}
d \vec{\alpha}'
d \vec{\alpha}^Q\nonumber \\
=&
\frac{1}{\pi^s}
\int_{}
e^{-\sum_{j=1}^s |\alpha_j'|^2}
\frac{|\alpha_1'|^{2 k_1} \cdots |\alpha_s'|^{2 k_s}}{k_1 ! \cdots k_s !}
f( |\alpha_1'|^2, \ldots, |\alpha_s'|^2)
d \vec{\alpha}'
\nonumber \\
=&
\frac{1}{\pi^s}\int_{[0,\infty)^s}
e^{-\sum_{j=1}^s r_j}
\frac{r_1^{k_1} \cdots r_s^{k_s}}{k_1 ! \cdots k_s !}
f( r_1, \ldots, r_s)
 d_1 \ldots d r_s
\nonumber \\
=&
\int_{[0,\infty)^s}
f( r_1, \ldots, r_s)
p_{k_1}( r_1) \cdots p_{ k_s}(r_s) d r_1 \ldots d r_s,\Label{e10-18}
\end{align}
where $\alpha_j:=\alpha^R_j+i\alpha^I_j$,
$f( r_1, \ldots, r_s):=  
\int_{\set{T}_{W^Q,c}(\vec{0})} 
\frac{1}{\pi^s N_1 \ldots N_s}
e^{-\frac{\sum_{j=1}^s|\alpha_j'-\alpha_j|^2}{\hat{N}_j}}
d \vec{\alpha}^Q$ with $r_j=|\alpha_j'|^2$
and
$p_{k}( r)$
is the PDF defined as
$e^{-r}\frac{r^{k} }{\pi k !}$.
We find that $f( r_1, \ldots, r_s)$ is monotone increasing for $r_1, \ldots, r_s$.
Also,
$p_{k}( r) \ge p_0(r)$ for $ r \ge (k!)^{1/k} $,
and $p_{k}( r) < p_0(r)$ for $ r < (k!)^{1/k} $.
These facts show that
\begin{align}
& \int_{[0,\infty)^s}
f( r_1, \ldots, r_s)
p_{0}( r_1) \cdots p_{ 0}(r_s) d r_1 \ldots d r_s \nonumber \\
\le &
\int_{[0,\infty)^s}
f( r_1, \ldots, r_s)
p_{k_1}( r_1) \cdots p_{ k_s}(r_s) d r_1 \ldots d r_s.\Label{e10-17}
\end{align}
When $\tau$ is the vacuum state,
the equality in (\ref{ER5}) hold.
So, combining (\ref{e10-18}) and (\ref{e10-17}), we obtain (\ref{ER5}). 

In the classical case, 
the covariant measurement is unique.
So, we have the extension as in Lemma \ref{L2B}.

\section{Proof of Lemma \ref{L4}}\Label{app-diag}
(i)$\Rightarrow$ (ii):
Since
$S^{-1 } A_2 (S^{T})^{-1}=(S^T A_2^{-1} S)^{-1}$,
$S^T A_1 S$, and $D$
commute with each other,
we have
\begin{align}
& D S^{-1 } A_2  A_1 S
=
D S^{-1 } A_2 (S^{T})^{-1} S^T A_1 S\nonumber\\
= &
S^T A_1 S
S^{-1 } A_2 (S^{T})^{-1} D
= 
S^T A_1  A_2 (S^{T})^{-1} D.\Label{PPY}
\end{align}
Since $ S^{T}D S =D$, we have
$ S^{T}D  =DS^{-1}$
and
$ D S  =(S^{T})^{-1} D$.
Thus,
\begin{align}
& S^{T} D A_2  A_1 S
= 
S^T A_1  A_2  D S ,\Label{PPY2}
\end{align}
which implies (ii).

(ii)$\Rightarrow$ (i):
Let $S$ be a symplectic matrix to symplectically diagonalize $A_1$.
Combining (\ref{PPY}) and (\ref{PPY2}), we have
\begin{align*}
D S^{-1 } A_2 (S^{T})^{-1} S^T A_1 S 
= 
S^T A_1 S
S^{-1 } A_2 (S^{T})^{-1} D
\end{align*}
Since $D^2=-1$, we have
\begin{align*}
 S^{-1 } A_2 (S^{T})^{-1} S^T A_1 S D
= 
D S^T A_1 S
S^{-1 } A_2 (S^{T})^{-1} 
= 
 S^T A_1 S D
S^{-1 } A_2 (S^{T})^{-1} .
\end{align*}
Hence, 
$S^{-1 } A_2 (S^{T})^{-1} $ commute with  
$ S^T A_1 S D$.
There exists an orthogonal matrix $S'$ such that
$S S'$ is a symplectic matrix, and 
$ (S S')^T A_1 (S S')$
and 
$(S S')^{-1 } A_2 ((S S')^{T})^{-1} $ are diagonal matrices.
Considering the inverse of $A_2^{-1}$, we obtain (i).

\section{Derivation of Eq. (\ref{bound-example1})}\Label{app-example1}
First,  SLD operators for $x$, $y$, and the nuisance parameter $z$ can be calculated by solving the equation  $\frac{\partial \rho_{\vec{t}}}{\partial \vec{t}_j}=\frac12\left(\rho_{\vec{t}}{L}_j+{L}_j\rho_{\vec{t}}\right)$:
\begin{align*}
L_x&=-\frac{x'}{1-|\vec{n}|^2}I+\left(\vec{a}'+\frac{x'\vec{n}}{1-|\vec{n}|^2}\right)\cdot\vec{\sigma}\\
L_y&=-\frac{y'}{1-|\vec{n}|^2}I+\left(\vec{b}'+\frac{y'\vec{n}}{1-|\vec{n}|^2}\right)\cdot\vec{\sigma}\\
L_z&=-\frac{z}{1-|\vec{n}|^2}I+\left(\vec{c}+\frac{z\vec{n}}{1-|\vec{n}|^2}\right)\cdot\vec{\sigma},
\end{align*}
 where $s:=\vec{a}\cdot\vec{b}$, $\vec{a}'=\frac{\vec{a}-s\vec{b}}{1-s^2}$, $\vec{b}'=\frac{\vec{b}-s\vec{a}}{1-s^2}$,
$x'=\frac{x-ys}{1-s^2}$, and $y'=\frac{y-xs}{1-s^2}$.
By definition, the SLD Fisher information can be evaluated as $\left(J_{\vec{t}}\right)_{ij}:=\Tr\rho_{\vec{t}}(L_iL_j+L_jL_i)/2$ and the D-matrix can be evaluated as $(D_{\vec{t}})_{ij}:=i\Tr\rho_{\vec{t}}[L_i,L_j]$. Explicitly, we get
\begin{align*}
J=\left(\begin{array}{ccc}\frac{1}{1-s^2}+\frac{x'^2}{1-|\vec{n}|^2} &-\frac{s}{1-s^2}+\frac{x'y'}{1-|\vec{n}|^2} & \frac{x'z}{1-|\vec{n}|^2} \\
\\
-\frac{s}{1-s^2}+\frac{x'y'}{1-|\vec{n}|^2}  &\frac{1}{1-s^2}+\frac{y'^2}{1-|\vec{n}|^2} &\frac{y'z}{1-|\vec{n}|^2}\\
\\
\frac{x'z}{1-|\vec{n}|^2}&\frac{y'z}{1-|\vec{n}|^2} & 1+\frac{z^2}{1-|\vec{n}|^2}\end{array}\right), 
\end{align*}
\begin{align*}
  J^{-1}=\left(\begin{array}{ccc}\frac{1-|\vec{n}|^2+y'^2(1-s^2)+z^2}{1-|\vec{n}|^2+x'^2+2x'y's+y'^2+z^2}&-\frac{x'y'(1-s^2)+(1-|\vec{n}|^2+z^2)s}{1-|\vec{n}|^2+x'^2+2x'y's+y'^2+z^2} & -\frac{(x'+sy')z}{1-|\vec{n}|^2+x'^2+2x'y's+y'^2+z^2} \\
\\
-\frac{x'y'(1-s^2)+(1-|\vec{n}|^2+z^2)s}{1-|\vec{n}|^2+x'^2+2x'y's+y'^2+z^2} &\frac{1-|\vec{n}|^2+x'^2(1-s^2)+z^2}{1-|\vec{n}|^2+x'^2+2x'y's+y'^2+z^2}&-\frac{(sx'+y')z}{1-|\vec{n}|^2+x'^2+2x'y's+y'^2+z^2}\\
\\
-\frac{(x'+sy')z}{1-|\vec{n}|^2+x'^2+2x'y's+y'^2+z^2} & -\frac{(sx'+y')z}{1-|\vec{n}|^2+x'^2+2x'y's+y'^2+z^2} &\frac{1-|\vec{n}|^2+x'^2+2sx'y'+y'^2}{1-|\vec{n}|^2+x'^2+2x'y's+y'^2+z^2}
 \end{array}\right),
\end{align*} 
and
\begin{align*}
D=\frac{1}{\sqrt{1-s^2}}\left(\begin{array}{ccc} 0 & -2z & 2y \\
\\
2z&0 &-2x\\
\\
-2y&2x & 0\end{array}\right).
\end{align*}
Finally, substituting the above into Eq. (\ref{NHB}), we get Eq. (\ref{bound-example1}).

\end{document}